\documentclass[preprint,12pt]{elsarticle}

\usepackage{algorithm}
\usepackage{algpseudocode}
\usepackage{booktabs}
\usepackage[T1]{fontenc}
\usepackage{hyperref}
\usepackage{lineno}
\usepackage{xcolor}
\usepackage{siunitx}
\usepackage{adjustbox}
\usepackage{amsmath}
\usepackage{tikz}

\usetikzlibrary{matrix,calc, positioning,fit,patterns}
\usepackage{amssymb}
\usepackage{amsthm}
\theoremstyle{definition}
\newtheorem{example}{Example}[section]

\newtheorem{theorem}{Theorem}[section]
\newtheorem{lemma}[theorem]{Lemma}
\newtheorem{definition}{Definition}[section]
\DeclareMathAlphabet{\mathqhv}{OT1}{qhv}{m}{n}

\newcommand{\sa}{$GS\!A^{i}$}
\newcommand{\same}{GS\!A^{i}}
\newcommand{\gsa}{$S\!A_{\mathcal{F}^{i}}$}
\newcommand{\gsame}{S\!A_{\mathcal{F}^{i}}}
\newcommand{\exppset}{$\mathcal{F}^{i}_{exp}$}
\newcommand{\exppsetme}{\mathcal{F}^{i}_{exp}}

\newcommand{\unformat}[2]{\sisetup{scientific-notation = false, round-precision=#1}\num{#2}}

\journal{Special Issue I\&C}

\begin{document}

%datasets
\newcommand{\ill}[1]{\texttt{ill#1}}
\newcommand{\hg}[1]{\texttt{hg#1}}
\newcommand{\hifi}{\texttt{pbhf}}
\newcommand{\ecoli}{\texttt{ecg31k}}
\begin{frontmatter}

%% Title, authors and addresses

%% use the tnoteref command within \title for footnotes;
%% use the tnotetext command for theassociated footnote;
%% use the fnref command within \author or \address for footnotes;
%% use the fntext command for theassociated footnote;
%% use the corref command within \author for corresponding author footnotes;
%% use the cortext command for theassociated footnote;
%% use the ead command for the email address,
%% and the form \ead[url] for the home page:
%% \title{Title\tnoteref{label1}}
%% \tnotetext[label1]{}
%% \author{Name\corref{cor1}\fnref{label2}}
%% \ead{email address}
%% \ead[url]{home page}
%% \fntext[label2]{}
%% \cortext[cor1]{}
%% \affiliation{organization={},
%%             addressline={},
%%             city={},
%%             postcode={},
%%             state={},
%%             country={}}
%% \fntext[label3]{}

\title{Efficient Construction of the BWT for Repetitive Text using String Compression} %TODO Please add

%% use optional labels to link authors explicitly to addresses:
%% \author[label1,label2]{}
%% \affiliation[label1]{organization={},
%%             addressline={},
%%             city={},
%%             postcode={},
%%             state={},
%%             country={}}
%%
%% \affiliation[label2]{organization={},
%%             addressline={},
%%             city={},
%%             postcode={},
%%             state={},
%%             country={}}

\author[fin,cebib]{Diego D\'iaz-Dom\'inguez}
\author[chl,cebib]{Gonzalo Navarro}

\affiliation[fin]{organization={Department of Computer Science, University of Helsinki},%Department and Organization
%            addressline={}, 
%            city={},
%            postcode={}, 
%            state={},
            country={Finland}
}
\affiliation[chl]{organization={Department of Computer Science, University of Chile},%Department and Organization
%            addressline={}, 
%            city={},
%            postcode={}, 
%            state={},
            country={Chile}
}
\affiliation[cebib]{organization={CeBiB --- Center for Biotechnology and Bioengineering}, country={Chile}}

\begin{abstract}
We present a new semi-external algorithm that builds the Burrows--Wheeler transform variant of Bauer et al.\ (a.k.a., BCR BWT) in linear expected time. Our method uses compression techniques to reduce computational costs when the input is massive and repetitive. Concretely, we build on induced suffix sorting (ISS) and resort to run-length and grammar compression to maintain our intermediate results in compact form. Our compression format not only saves space but also speeds up the required computations. Our experiments show important space and computation time savings when the text is repetitive. In moderate-size collections of real human genome assemblies (14.2 GB - 75.05 GB), our memory peak is, on average, 1.7x smaller than the peak of the state-of-the-art BCR BWT construction algorithm (\texttt{ropebwt2}), while running 5x faster. Our current implementation was also able to compute the BCR BWT of 400 real human genome assemblies (1.2 TB) in 41.21 hours using 118.83 GB of working memory (around 10\% of the input size). Interestingly, the results we report in the 1.2 TB file are dominated by the difficulties of scanning huge files under memory constraints (specifically, I/O operations). This fact indicates we can perform much better with a more careful implementation of our method, thus scaling to even bigger sizes efficiently. 
\end{abstract}

%%Graphical abstract
%\begin{graphicalabstract}
%\includegraphics{grabs}
%\end{graphicalabstract}

%%Research highlights
\begin{highlights}
\item We introduce a new algorithm to build the Burrows-Wheeler Transform on massive and highly repetitive text collections.
\item Our method builds on Induced Suffix Sorting and uses grammar compression to maintain the intermediate results in compressed form.
\item Our experiments demonstrate that our particular format saves significant space and computation time. 
\end{highlights}

\begin{keyword}
BWT \sep string compression \sep repetitive text
%% keywords here, in the form: keyword \sep keyword

%% PACS codes here, in the form: \PACS code \sep code

%% MSC codes here, in the form: \MSC code \sep code
%% or \MSC[2008] code \sep code (2000 is the default)
\end{keyword}

\end{frontmatter}

%\linenumbers

%% main text
\section{Introduction} \label{sec:itr}

The Burrows--Wheeler transform (BWT)~\cite{bw94} is a reversible string transformation that reorders the symbols of a text $T$ according to the lexicographical ranks of its suffixes. The features of this transform have turned it into a critical component for text compression and indexing \cite{Ohl13,makinen2015genome}. In addition to being reversible, the reordering reduces the number of equal-symbol runs in $T$, thus improving the compressibility. The BWT is also the main component of the so-called FM-index~\cite{fe00op,g2018op}, a self-index that supports pattern matching in time proportional to the pattern length. Briefly, the FM-index encodes $T$ as its BWT and then uses its combinatorial properties~\cite{ga17wh} to look for patterns in the text efficiently. Popular bioinformatic tools~\cite{la09ul,li2010fast} rely on the FM-index to process data, as collections in this area are typically massive and repetitive, and the patterns to search for are short.

The FM-index is an important breakthrough as it dramatically reduces space usage compared to the classical suffix tree~\cite{we73li} and suffix array~\cite{MM93}. However, it still uses space proportional to $T$ (i.e., succinct), making it impractical for massive input. This problem is relevant as massive collections are nowadays standard in many fields. A fortunate coincidence is that massive collections are usually highly repetitive too, and in that case, the number of equal-symbol runs in the BWT (denoted $r$ in the literature) is considerably smaller than the text size. Gagie et al.~\cite{g2018op} exploited this property to design the $r$-index, a compressed self-index that requires $O(r)$ bits of space and still supports efficient pattern matching.

The $r$-index is a promising solution to compress and index massive collections, but it lacks efficient construction algorithms that scale well with the input size. This limitation, of course, hampers its adoption in practical applications. One of the most important challenges (although not the only one) is how to obtain the BWT of $T$. Several algorithms in the literature produce the BWT in linear time~\cite{ok2009li,bauer13lw,lou20gsuf,egidi19ext,bon20com}. Nevertheless, the computational resources their implementations require with large inputs are still too high. This limitation is particularly evident in Genomics, where the data can easily reach terabytes~\cite{st15big}. 

Some authors~\cite{kem19str,b2019pr,kemp19op,kemp20res,bou21com} have tackled the problem of computing big BWTs by exploiting the repetitiveness of the input. Their approach consists of extracting a set of representative strings from the text, performing calculations on them, and then extrapolating the results to the copies of those strings. For instance, the methods of Boucher et al.~\cite{b2019pr,bou21com} based on prefix-free parsing (PFP) use Karp--Rabin fingerprints~\cite{ka87ef} to create a dictionary of prefix-free phrases from $T$. They then create a parse by replacing the phrases in $T$ with metasymbols, and finally construct the BWT using the dictionary and the parse. Similarly, Kempa et al.~\cite{kem19str} consider a subset of positions in $T$ they call a string synchronizing set, from which they compute a partial BWT they then extrapolate to the whole text.

Although these repetition-aware techniques are promising, some are at a theoretical stage~\cite{kem19str, kemp19op, kemp20res}, while the rest~\cite{b2019pr,bou21com} have been empirically tested only under controlled settings, and their results depend on parameters that are not simple to tune. Thus, it is difficult to assess their performance under real circumstances. 

Recently, Nunes et al.~\cite{n2018gr} proposed a method called GCIS that adapts the concept of \emph{induced suffix sorting} (ISS) for compression. Their ideas are closely related to the linear-time BWT algorithm of Okanohara et al.~\cite{ok2009li}. Briefly, Okanohara et al.~cut the text into phrases using ISS, assigning symbols to the phrases, and then replacing the phrases with their symbols. They apply this procedure recursively until all the text symbols are unique. Then, when they return from the recursions, they induce an intermediate $BWT^{i}$ for the text of every recursion $i$ using the previous $BWT^{i+1}$. %Besides, GCIS stores the dictionaries ISS generates along the recursion in a context-free grammar.
The connection between these two methods is that GCIS captures in the grammar precisely the information that Okanohara et al.~use to compute the BWT. Additionally, D\'iaz-Dom\'inguez et al.~\cite{diaz21gram} recently demonstrated that ISS-based compressors such as GCIS require much less computational resources than state-of-the-art methods like RePair~\cite{lar00off} to encode the data while maintaining high compression ratios. The simple construction of ISS makes it an attractive alternative to processing high volumes of text. In particular, combining the ideas of Okanohara et al. with ISS-based compression is a promising alternative for computing big BWTs.

\textbf{Our contribution.} {\em Induced suffix sorting} (ISS)~\cite{ko2005space} has proved useful for compression~\cite{n2018gr, diaz21gram} and for constructing the BWT~\cite{ok2009li}. In this work, we show that compression can be incorporated into the internal stages of the BWT computation in a way that saves both working space and time. Okanohara et al.~\cite{ok2009li} use ISS to construct the BWT as follows: they build the texts $T^{1}, T^{2}, \ldots, T^{h}$, with $h=O(\log n)$, by applying recursive rounds of parsing that cut each $T^{i}$ into phrases and replace the phrases by new symbols. Then, when they return from the recursions, they induce the BWT of every $T^{i}$ from the BWT of the previous text $T^{i+1}$, generating the final BWT when they reach the first recursion level again. We use a technique similar to grammar compression to store the sets of phrases generated in the rounds of parsing, and run-length compression for the intermediate BWTs. This approach is shown not only to save the space required for those intermediate results but, importantly, the format we choose speeds up the computation of the final BWT as we return from the recursion because the factorizations that help save space also save redundant computations. Unlike Okanohara et al., we receive as input a string collection and output its BCR BWT~\cite{bauer13lw}, a variant for string collections. The reason is that massive datasets usually contain multiple strings, in which case the BCR BWT variant is simpler to construct.

Early versions of this work appeared in {\em Proc.~DCC'21}~\cite{diaz21gram} and {\em Proc. CPM'22}~\cite{DNcpm22.2}. In this extended article, we explain how to produce a smaller set of phrases in each recursion of parsing than the one we obtain by applying the regular ISS procedure. We aim to keep working memory and CPU consumption low, even for not-so-repetitive collections. Our technique to reduce the set of phrases is simple enough, so the extra time we spend in this step increases the overall performance. Additionally, we explain how to extend our ideas to compute the smallest BCR BWT (in terms of the number of runs) one can obtain by reordering the strings of the text. Finally, we empirically assessed our techniques in massive datasets.

Our experiments show that when the input is a collection of human genomes (a repetitive dataset), we are, on average, 5x faster than \texttt{ropebwt2}~\cite{li2010fast}, one of the most efficient implementations of the BCR BWT algorithm. In the same datasets, we also outperform the PFP-based methods \texttt{pfp-ebwt}~\cite{bou21com} and \texttt{r-pfpbwt}~\cite{ol23rec}, being on average 2.8x faster than them while using much less memory. We also report the construction of the BCR BWT for a collection of 1.2 TB using an amount of working memory that did not exceed 10\% of the input size, and a running time of less than 42 hours. Under not-so-repetitive scenarios (short Illumina reads), we are the second fastest tool and the second most space-efficient on average, being outperformed only by \texttt{ropebwt2} and \texttt{BCR\_LCP\_GSA}~\cite{bauer13lw}, respectively.
\section{Related Concepts}

\subsection{Grammar and Run-length Compression}\label{ssec:comp}

\emph{Grammar compression}~\cite{Kieffer2000} consists of encoding a text $T$ as a small context-free grammar $\mathcal{G}$ that produces only $T$. Formally, a grammar is a tuple $(V,\Sigma, \mathcal{R}, \mathqhv{S})$, where $V$ is the set of nonterminals, $\Sigma$ is the set of terminals, $\mathcal{R}$ is the set of replacement rules and $\mathqhv{S} \in V$ is the start symbol. The right-hand side of $\mathqhv{S} \rightarrow C \in \mathcal{R}$ is referred to as the compressed form of $T$. The size of $\mathcal{G}$ is usually measured in terms of the number of rules, the sum of the lengths of the right-hand sides of $\mathcal{R}$, and the length of the compressed string.

\emph{Run-length compression} encodes the equal-symbol~runs of maximal length in $T$ as a sequence $(c_1, \ell_1),(c_2,\ell_2),\ldots,(c_{n'}, \ell_{n'})$ of $n'\leq n$ pairs, where every $(c_i, \ell_i)$, with $i \in [1, n']$, stores the symbol $c_i \in \Sigma$ of the $ith$ run and its length $\ell_{i} \geq 1$. For instance, let $T[j..j']=cccc$ be a substring with four consecutive copies of $c$, where $T[j-1]\neq a$ and $T[j'+1]\neq c$. Then $T[j..j]$ compresses to $(c,4)$.

\subsection{The Suffix Array}\label{ssec:sa}

The \emph{suffix array}~\cite{MM93} of a string $T[1..n] \in \Sigma^{*}$ is a permutation $S\!A[1..n]$ that enumerates the suffixes $T[j..n]$ of $T$ in increasing lexicographic order, $T[S\!A[j]..n] < T[S\!A[j+1]..n]$, for $j \in [1..n-1]$. It is customary to divide $S\!A$ into $\sigma$ \emph{buckets}. Specifically, a bucket $c \in \Sigma$ is a contiguous range $S\!A[j_c..j_{c+1}-1]$ storing the text positions of the suffixes of $T$ prefixed by $c$.

The \emph{generalized suffix array}~\cite{shi1996} is a variant of $S\!A$ that enumerate the suffixes of a string collection $\mathcal{T}=\{T_{1},T_{2},\ldots,T_{k}\}$ over the alphabet $\Sigma$. Let $T=T_1\texttt{\$}_{1}T_{2}\texttt{\$}_{2}\cdots T_{k}\texttt{\$}_{k}$ be a string over the alphabet $\{\texttt{\$}_1, \texttt{\$}_2, \ldots, \texttt{\$}_k\} \cup \Sigma$ storing the concatenation of $\mathcal{T}$ such that each $T_x \in \mathcal{T}$ ends in $T$ with a unique sentinel $\texttt{\$}_{x}$. The values of the $k$ distinct sentinels $\texttt{\$}_{1}, \ldots, \texttt{\$}_{k} \notin \Sigma$ are chosen arbitrarily but are smaller than any symbol in $\Sigma$. The generalized suffix array of $\mathcal{T}$ is then a vector $GS\!A[1..n=|T|]$ equal to the suffix array of $T$. Put simply, $GS\!A$ sorts the suffixes in lexicographical order, breaking ties for equal suffixes according to the order of the strings in $\mathcal{T}$ that the sentinels induce. %For instance, given two strings $T_x[1..n_x], T_y[1..n_y] \in \mathcal{T}$ suffixed by the sequence $T_{x}[u..n_x]=T_{y}[u'..n_y]$, the relative order of $T_{x}[u..n_x]$ and $T_{y}[u'..n_y]$ in $GS\!A$ depends on the values we choose for $\texttt{\$}_{x}$ and $\texttt{\$}_{y}$.
Still, in practice, a construction algorithm for $GS\!A$ does not require keeping an explicit set of $k$ sentinels. One can get the same result by concatenating the elements of $\mathcal{T}$ in $T=T_1\texttt{\$}T_{2}\texttt{\$}\cdots T_{k}\texttt{\$}$ using the same symbol $\texttt{\$}$ as a boundary between strings, sorting the suffixes of $\mathcal{T}$ (substrings in $T$) lexicographically, and breaking ties for equal suffixes according to an arbitrary rule. Figure~\ref{fig:bct_bwt_examp} shows an example of $GS\!A$.

\subsection{The Burrows--Wheeler Transform}\label{ssec:bwt}

Let $T[1..n]$ be a string over the alphabet $\Sigma \cup \{\texttt{\$}\}$ where $\texttt{\$}$ is smaller than any symbol in $\Sigma$ and only occurs in $T[n]$.  The \emph{Burrows--Wheeler transform} (BWT)~\cite{bw94} of $T$ is a reversible string transformation that stores in $BWT[j]$ the symbol that precedes the $jth$ suffix of $T$ in lexicographical order, i.e., $BWT[j] = T[S\!A[j]-1]$ (assuming $T[0]=T[n]=\texttt{\$}$).

The mechanism to revert the transformation is the so-called $\mathsf{LF}$ mapping. Given an input position $BWT[j]$ that maps a symbol $T[u]$, $\mathsf{LF}(j) = j'$ returns the index $j'$ such that $BWT[j']=T[u-1]$ maps the preceding symbol of $T[u]$. Thus, spelling $T$ reduces to continuously applying $\textsf{LF}$ from $BWT[1]$, the symbol to the left of $T[n]=\texttt{\$}$, until reaching $BWT[j]=\texttt{\$}$. 

The BCR BWT~\cite{bauer13lw} is a reversible transformation that reorders the symbols of $\mathcal{T}=\{T_{1},T_{2},\ldots,T_{k}\}$. Consider again $T=T_1\texttt{\$}_1T_{2}\texttt{\$}_2\cdots T_{k}\texttt{\$}_k$, the sequence of length $n=|T|$ storing the concatenation of the strings in $\mathcal{T}$ separated by unique sentinels $\texttt{\$}_1 \ldots \texttt{\$}_k \notin \Sigma$. Additionally, let us define the vector $GS\!A[1..n]$ for $\mathcal{T}$ using the order $\texttt{\$}_1<\texttt{\$}_2\ldots <\texttt{\$}_k$. The BCR BWT of $\mathcal{T}$ is a vector $BWT_{bcr}[1..n]$ storing in $BWT_{bcr}[j]$ the symbol $T[GS\!A[j]-1]$. It is worth mentioning that when $T[GS\!A[j]]=T_{x}[1]$ maps to the leftmost symbol of a string $T_x \in \mathcal{T}$, $BWT_{bcr}[j]=\texttt{\$}_{x}$ is (theoretically) the sentinel at the end of $T_{x}$ in $T$. Spelling strings in $\mathcal{T}$ from $BWT_{bcr}$ works similarly to the procedure in the standard BWT. Successive rounds of $\mathsf{LF}$ operations starting from $BWT_{bcr}[1]$ and finishing when the symbol in $BWT_{bcr}[j']$ is a sentinel spells $T_{1}$ from right to left. The same procedure but starting from $BWT_{bcr}[2]$ spells $T_2$, and so on. Figure~\ref{fig:bct_bwt_examp} depicts an example of $BWT_{bcr}$.

A common measure of compression for a text is the number of equal-symbol runs in its BWT (denoted $r$ in the literature), but when the text is a collection, the value $r$ associated with its BCR BWT varies depending on the order of the special symbols $\texttt{\$}_1,\ldots,\texttt{\$}_k$. Thus, the optimal BCR BWT $(BWT_{opt})$ is the transform built with the sentinel ordering that minimizes $r$. Bentley et al.~\cite{ben20ont} proposed a linear-time procedure (referred to here as \textsf{CAO}\footnote{The word stands for \emph{Constraint Alphabet Ordering}, the original name Bentley et al. gave to the problem they were studying.}) that receives as input $BWT_{bcr}$ and produces $BWT_{opt}$. They consider the partition $(s_1, e_1), (s_2, e_2), \ldots, (s_x, e_x)$ of $BWT_{bcr}$ induced by equal suffixes of $\mathcal{T}$.~That is, every block $BWT_{bcr}[s_u..e_u]$, with $u \in [1..x]$, stores the left-context symbols of different suffixes of $\mathcal{T}$ that spell the same sequence. They regard the partition as a vector $\mathcal{A}$ where every $uth$ element is a tuple collapsing the symbols of $BWT_{bcr}[s_u..e_u]$ by their values. Thus, the $uth$ tuple is a sequence $\mathcal{A}[u] = (c_{1}, \ell_{1}),\ldots,(c_{b}, \ell_{b})$ of $1\leq b \leq |\Sigma \cup \{\texttt{\$}\}|$ pairs where $(c_{p}, \ell_{p})$, with $p \in [1,b]$, groups the $\ell_{p}$ occurrences of symbol $c_{p} \in \Sigma$ within $BWT_{bcr}[s_u..e_u]$.

\begin{figure}[!t]
\centering
\resizebox{0.9\textwidth}{!}{%
\begin{tikzpicture}[>=stealth,thick,baseline]

\matrix (m1) [matrix of nodes, ampersand replacement=\|,
row 1/.style={nodes={font=\ttfamily}},
row 2/.style={nodes={gray!70, font=\footnotesize\bfseries}},
     ] at (0,0)  { 
 $T$\| $=$ \| a \| a \| c \| t \| $\texttt{\$}_1$ \| a \| c \| c \| t \| $\texttt{\$}_2$ \| c  \| a  \| c  \| t  \| $\texttt{\$}_3$\\
    \|     \| $1$ \| $2$ \| $3$ \| $4$ \| $5$  \| $6$ \| $7$ \| $8$ \| $9$ \| $10$ \| $11$ \| $12$ \| $13$ \| $14$ \| $15$\\
};

\matrix (m2) [matrix of nodes, ampersand replacement=\|, below= 2mm of m1,
row 1/.style={nodes={gray!70, rotate=90, font=\footnotesize\ttfamily, anchor=base west}},
row 4/.style={nodes={font=\ttfamily}},
]  { 
 \| \| $\$_{1}$\| $\$_{2}$ \| $\$_{3}$ \| aact$\$_{1}$ \| acct$\$_{3}$\| act$\$_{1}$ \| act$\$_{3}$ \| cact$\$_{3}$  \| cct$\$_{2}$ \| ct$\$_{1}$ \| ct$\$_{2}$ \| ct$\$_{3}$  \| t$\$_{1}$ \| t$\$_{2}$ \| t$\$_{3}$ \\
$GS\!A$\| $=$\| 5 \| 10 \| 15 \| 1 \| 6 \| 2 \| 12 \| 11 \| 7 \| 3 \| 8 \| 13 \| 4 \| 9 \| 14\\
\||[white]|aaa\\ 
$BWT_{bcr}$\| $=$\| t \| t \| t \| \$\textsubscript{1} \| \$\textsubscript{2} \| a \| c \| \$\textsubscript{2} \| a \| a \| c \| a \| c \| c \| c \\
};

\draw ($(m2-4-3.base west)+(0,0.4)$) rectangle ($(m2-4-5.base east)-(0,0.22)$);
\draw ($(m2-4-5.base east)+(0,0.4)$) rectangle ($(m2-4-6.base east)-(0,0.22)$);
\draw ($(m2-4-6.base east)+(0,0.4)$) rectangle ($(m2-4-7.base east)-(0,0.22)$);
\draw ($(m2-4-7.base east)+(0,0.4)$) rectangle ($(m2-4-9.base east)-(0,0.22)$);
\draw ($(m2-4-9.base east)+(0,0.4)$) rectangle ($(m2-4-10.base east)-(0,0.22)$);
\draw ($(m2-4-10.base east)+(0,0.4)$) rectangle ($(m2-4-11.base east)-(0,0.22)$);
\draw ($(m2-4-11.base east)+(0,0.4)$) rectangle ($(m2-4-14.base east)-(0,0.22)$);
\draw ($(m2-4-14.base east)+(0,0.4)$) rectangle ($(m2-4-17.base east)-(0,0.22)$);

\matrix (m3) [matrix of nodes, ampersand replacement=\|, below= 2mm of m2,
]  { 
$\mathcal{A}$\| $=$\| $[(\texttt{t}, 3)]$ \| $[(\texttt{\$}, 1)]$  \| $[(\texttt{\$}, 1)]$ \| $[(\texttt{a}, 1), (\texttt{c}, 1)]$ \| $[(\texttt{\$}, 1)]$ \| $[(\texttt{a}, 1)]$ \| $[(\texttt{c}, 1), (\texttt{a}, 2)]$ \| $[(\texttt{c}, 3)]$\\
\| \| $1$ \| $2$  \| $3$ \| $4$ \| $5$ \| $6$ \| $7$ \| $8$\\
};
\end{tikzpicture}
}
\caption{The generalized suffix array and the BCR BWT for the collection $\mathcal{T}=\{\texttt{acct}, \texttt{acct}, \texttt{cact}\}$. The string $T$ is the concatenation of $\mathcal{T}$ separated by sentinel symbols. The boxes in $BWT_{bcr}$ represent the partition induced by equal suffixes of $\mathcal{T}$. The numbers below $\mathcal{A}$ map the blocks in the partition of $BWT_{bcr}$ to tuples in $\mathcal{A}$.}
\label{fig:bct_bwt_examp}
\end{figure}

The key observation to produce $BWT_{opt}$ is that reordering the symbol within each block $BWT_{bcr}[s_x..e_x]$ \emph{only} affects the order in which one spells strings in $\mathcal{T}$ from the transform. For example, permuting symbols within $BWT[s_x..e_x]$ might produce that the string $T_x$ the BCR BWT spells from $BWT_{bcr}[1]$ via $\mathsf{LF}$ operations is no longer $T_{1}$. However, $T_{x}$ is still a member of $\mathcal{T}$. Bentley et al. noticed that one could minimize $r$ by sorting the pairs of every tuple $\mathcal{A}[j]$,  maximizing the number of matches between adjacent tuples. A match between adjacent tuples occurs when the symbol in the rightmost pair of $\mathcal{A}[j]$ equals the symbol of the leftmost pair of $\mathcal{A}[j+1]$.

\begin{example}
Consider the $BWT_{bcr}$ of Figure~\ref{fig:bct_bwt_examp}, which has $r=9$ equal-symbol runs, and its associated vector $\mathcal{A}$ (also in Figure~\ref{fig:bct_bwt_examp}). By rearranging the tuple $\mathcal{A}[7]= [(\texttt{c},1), (\texttt{a}, 2)]$ to $[(\texttt{a},2), (\texttt{c},1)]$, the symbol of $(\texttt{a},2)$ matches the symbol in $\mathcal{A}[6]=[(\texttt{a},1)]$. On the other hand, the symbol of $(\texttt{c},1)$ matches the symbol in $\mathcal{A}[8]=[(\texttt{c},3)]$. The collapse of $\mathcal{A}$ then produces the run-length-compressed vector $BWT_{opt} = (\texttt{t},3)\ (\texttt{\$},2)\ (\texttt{a},1)\ (\texttt{c},1)\ (\texttt{\$},1)\ (\texttt{a}, 3)\ (\texttt{c}, 4)$, which has $r=7$ equal-symbol runs. 
\end{example}

\textsf{CAO} finds a sorting for $\mathcal{A}$ that maximizes the number of adjacent matches in linear time. Nevertheless, it does not compute $\mathcal{A}$; it receives it as input. In this regard, Bentley et al. do not give many details on how to compute this vector efficiently. Recently, Cenzato et al.~\cite{cen22com} pointed out that it is possible to use the algorithm \textsf{SA-IS} (Section~\ref{ssec:iss}) to get a bit vector marking the blocks in the partition of $BWT_{bcr}$ induced by equal suffixes (the SAP-array~\cite{cox2012large}), from which one can easily obtain $\mathcal{A}$. %Additionally, their experiments showed that the value of $r$ in the optimal BCR BWT can be considerably smaller than in the original BWT when $\mathcal{T}$ has many short strings. 

\subsection{Induced Suffix Sorting}\label{ssec:iss}

\emph{Induced suffix sorting} (ISS)~\cite{ko2005space} computes the lexicographical ranks of a subset of suffixes in $T$ and uses the result to induce the order of the rest. This method is the underlying procedure in several algorithms that build the suffix array~\cite{n2009li,n2013pr,l2017in, karkk17eng} and the BWT~\cite{ok2009li,b2019pr} in linear time. The ISS idea introduced by the suffix array algorithm \textsf{SA-IS} of Nong et al.~\cite{n2009li} is of interest to this work. The authors give the following definitions:

\begin{definition}\label{txt:def:ltype}
A symbol $T[j]$ is called L-type if $T[j]>T[j+1]$ or if $T[j]=T[j+1]$ and $T[j+1]$ is also L-type. On the other hand, $T[j]$ is said to be S-type if $T[j]<T[j+1]$ or if $T[j]=T[j+1]$ and $T[j+1]$ is also S-type. By definition, the symbol $T[n]$, the one with the sentinel, is S-type.
\end{definition}

\begin{definition}\label{txt:def:stype}
A symbol $T[j]$, with $j \in [1..n]$, is called leftmost S-type, or LMS-type, if $T[j]$ is S-type and $T[j-1]$ is L-type. 
\end{definition}

\begin{definition}\label{txt:def:lmstype}
An LMS substring is (i) a substring $T[j..j']$ with both $T[j]$ and $T[j']$ being LMS-type symbols, and there is no other LMS symbol in the substring, for $j \neq j'$; or (ii) the sentinel itself.
\end{definition}

$\textsf{SA-IS}$ recursively sorts a subset of suffixes of $\mathcal{T}$ using ISS and then uses the result to induce the relative order of another set of suffixes. The recursion continues until there are no more suffixes of $\mathcal{T}$ to sort. When that happens, it returns from the recursion merging the distinct subsets of suffixes it processed before, producing the final suffix array $S\!A$ of $T$ when it reaches the first level of the recursion again. The key idea that makes \textsf{SA-IS} linear-time is that every recursion level operates over a sequence that is at most half the length of the sequence in the previous level.

\begin{algorithm}[!t]
\caption{\textsf{SA-IS}}
\label{algo:sais}{}
\begin{algorithmic}[1]
\small 
\Require $T^{i}[1..n^{i}]$
    \State $S\!A^{i}[1..n_i] \gets $ empty array
    \State scan $T^{i}$ right to left an insert LMS-type positions in $S\!A^{i}$
    \State scan $S\!A^{i}$ twice to sort the LMS substrings with different sequences 
    \If{all the symbols in $T^{i}$ are distinct}
    \State \textbf{return} $S\!A^{i}$
    \EndIf
    \State scan $S\!A^{i}$ to get the set $\mathcal{F}^{i}$ sorted in LMS order 
    \State $o \gets 1$
    \For{$F \in \mathcal{F}^{i}$}
        \State assign symbol $o \in \Sigma^{i+1}$ to $F$
        \State $o \gets o + 1$
    \EndFor
    \State $T^{i+1} \gets$\ replace LMS substrings in $T^{i}$ with their assigned symbols in $\Sigma^{i+1}$
    \State $S\!A^{i+1} \gets \textsf{SA-IS}(T^{i+1})$ 
    \State $S\!A^{i} \gets$\ store LMS-type positions of $T^{i}$ as their symbol in $\Sigma^{i+1}$ appear in $S\!A^{i+1}$ 
    \State scan $S\!A^{i}$ twice again to the sort the suffixes of $T^{i}$
    \State \textbf{return} $S\!A^{i}$
\end{algorithmic}
\end{algorithm}

The suffixes that \textsf{SA-IS} processes in every level $i$ are those prefixed by LMS substrings. The level receives as input a string $T^{i}[1..n^i]$ over an alphabet $\Sigma^{i}$ (when $i=1$, $T^{i}=T$ and $\Sigma^{i}=\Sigma$) and initializes an empty suffix array $S\!A^{i}[1..n^{i}]$. Then, it scans $T^{i}$ from right to left to classify the symbols as L-type, S-type, or LMS-type. As it moves through the text, the algorithm records the text positions of the LMS substrings in $S\!A^{i}$. More specifically, if $T^{i}[j]=b$ is the leftmost symbol of an LMS substring, it inserts $j$ in the rightmost empty position in the bucket $b$ of $S\!A^{i}$. After scanning $T^{i}$, $\textsf{SA-IS}$ sorts the LMS substrings in $S\!A^{i}$ using two linear scans.%, which requires two linear scans of $S\!A^{i}$.

\paragraph{Scans of the Suffix Array}\label{par:scan} \textsf{SA-IS} sorts the suffixes of $T^{i}$ prefixed by different LMS substrings in two linear scans of $S\!A^{i}$. The first scan traverses the array left to right, and for each index $j$ such that $T^{i}[S\!A^{i}[j]-1]=c$ is L-type, it inserts the text position $S\!A^{i}[j]-1$ in the leftmost empty cell of bucket $c$ in $S\!A^{i}$. The second scan traverses $S\!A^{i}$ from right to left. This time, for each index $j$ such that $T^{i}[S\!A^{i}[j]]=c$ is S-type, it stores the text position $S\!A^{i}[j]-1$ in the rightmost empty cell in the bucket $c$ of $S\!A^{i}$.

ISS rearranges the LMS substrings in $S\!A^{i}$ in a slightly different way  from the lexicographic order. Let $F_x[1..n_x]$ and $F_y[1..n_y]$ be two strings over $\Sigma^{i}$ labelling LMS substrings of $T^{i}$, and let $\prec_{lex}$ be the operator that denotes lexicographical order. LMS order ($\prec_{LMS}$) is defined as

\[ F_x \prec_{LMS} F_y \Leftrightarrow
   \begin{cases}
     \text{$n_y<n_x$ and $F_x[1..n_y]=F_y$} \\
     F_x \prec_{lex} F_y\ \text{otherwise.} \\
   \end{cases}
\]

For instance, it holds that $\texttt{actca} \prec_{LMS} \texttt{actc}$, which differs from the standard $\texttt{actc} \prec_{lex} \texttt{actca}$. However, for two occurrences $F_y = T^{i}[j..j']$ and $F_{x}=T^{i}[u..u']$, the higher LMS order of $F_{y}$ implies that the suffix $T^{i}[j..n^i]$ is lexicographically greater than the suffix $T^{i}[u..n^i]$. The cause of this property is explained in Section 2 of Ko and Aluru~\cite{ko2005space}.

The idea now is to use the sorted LMS substrings to induce the order of the suffixes in $T^{i}$ that are not prefixed by LMS substrings. Still, the relative order of LMS substrings with the same sequence is unknown in $S\!A^{i}$. Nong et al.~solve this problem by creating a new string $T^{i+1}$ in which they replace the LMS substrings with their LMS orders. Let $\mathcal{F}^{i}$ be the set of strings labelling LMS substrings in $T^{i}$. If a string $F \in \mathcal{F}^{i}$ has LMS order $o$ in the set, then \textsf{SA-IS} replaces the LMS substrings of $T^{i}$ labelled $F$ by $o$. Notice that now it is possible to compute the LMS orders in one scan of $S\!A^{i}$. The resulting sequence $T^{i+1}$ is used as input for another recursive call $i+1$. The base case for the recursion is when all the suffixes in $S\!A^{i}$ are prefixed by different symbols, in which case the algorithm returns $S\!A^{i}$ without further processing.

When the $(i+1)th$ recursive call ends, all the suffixes of $T^{i}$ prefixed by the same LMS substrings are sorted in $S\!A^{i+1}$, so \textsf{SA-IS} proceeds to complete $S\!A^{i}$. For doing so, it resets $S\!A^{i}$, inserts the LMS substrings arranged as their respective symbols appear in $S\!A^{i+1}$, and performs the two scans of paragraph~\ref{par:scan} to reorder the unsorted suffixes of $T^{i}$. Once it finishes, it passes $S\!A^{i}$ to the previous recursion $i-1$. The final array $S\!A^{1}$ is the suffix array for $T$. Algorithm~\ref{algo:sais} explains the general idea of \textsf{SA-IS}. We also refer the reader to Figure 2 in Louza et al.~\cite{l2017in} for a detailed running example of a recursion level in \textsf{SA-IS}.
\section{Methods}

\subsection{Definitions}

We assume the standard string notation.~For a string $T[1..n]$ over an arbitrary alphabet, $T[1..j]$ is the $jth$ prefix of $T$, and $T[j..n]$ is the $jth$ prefix. Additionally, the operator $|T|=n$ represents the number of symbols in $T$ (i.e., its length). We will refer to $T[j..n]$ as a \emph{proper} suffix if $1<j \leq n$, and non-proper otherwise. When $T$ is run-length-compressed, we use the format $c^{\ell}$ to denote an equal-symbol run of $\ell$ copies of $c$. We also use the alternative notation $(c,\ell)$ with the same meaning. We choose one format or the other depending on the context. Additionally, we use $\varepsilon$ as the empty symbol.  

Let $\mathcal{S}=\{S_{1}, \ldots, S_{z}\}$ be a string set. We consider a suffix $S_x[j..n_x]$ to be \emph{left-maximal} if there is at least one other string $S_y[1..n_{y}] \neq S_x \in \mathcal{S}$ with a suffix $S_y[j'..n_{y}]$ such that (i) $S_y[j'..n_{y}]=S_x[j..n_{x}]$, and (ii) either both $S_y[j'..n_{y}]$ and $S_x[j..n_{x}]$ are proper suffixes with $S_y[j'-1] \neq S_x[j-1]$, or one of them ($S_x[j..n_x]$ or $S_y[j'..n_y]$) is not proper. We will use the operator $|\mathcal{S}|$ for the total number of strings in $\mathcal{S}$ and $||\mathcal{S}||$ to express the total number of symbols $\sum_{S_x \in \mathcal{S}} |S_x|$. 

Through the paper, we use the superscript $i$ to denote elements of the $ith$ recursion level in \textsf{SA-IS}. Thus, the string $T^{i}[1..n^{i}]$ over the alphabet $\Sigma^{i}[1..\sigma^{i}]$ is the input for level $i$. Notice that, for instance, $\sigma^{i}$ does not mean $\sigma=\sigma^{1}$ raised to $i$. The same idea applies to other elements with the $i$ superscript. The expressions $T$ and $T^{1}$ are equivalent. The same holds for $\Sigma$ and $\Sigma^{1}$.

For a given position $T^{i}[j]$, with $j \in [1..n_i]$, the operator $exp^{u}(T^{i}[j]) \in \Sigma^{i}$, with $u<i$, denotes the string we obtain by recursively replacing the symbol $o=T^{i}[j] \in \Sigma^{i}$ with its associated phrase $F=o_1{\cdots}o_r \in \mathcal{F}^{i-1}$ over the alphabet $\Sigma^{i-1}$. The formal recursive definition of $exp^{u}$ is as follows: 

\[ exp(o)^{u} \Leftrightarrow
   \begin{cases}
     \text{$o$ if $o \in \Sigma^{u}$} \\
     \text{$exp(o_1)^{u-1}{\cdots}exp^{u-1}(o_{r-1})$ if $exp^{1}(o_r)$ does not end with $\texttt{\$}$} \\
     \text{$exp(o_1)^{u-1}{\cdots}exp^{u-1}(o_r)$ if $exp^{1}(o_r)$ ends with $\texttt{\$}$}\\
   \end{cases}
\]

Note that $exp^{u}$ removes the overlap between consecutive LMS substrings of $T^{1}$. Additionally, the function $map^{u}(T^{i}[j])=T^{u}[z..z']$, with $u<i$, returns the substring in $T^{u}$ from where $T^{i}[j]$ was formed. Formally, the boundaries $(z,z')$ in $map^{u}(T^{i}[j])$ are 

\begin{align*}
z &= 1 + \sum_{q=1}^{j-1} |exp^{u}(T^{i}[q])|\\
z' &= z+|exp^{u}(T^{i}[j])|-1. 
\end{align*}

The expressions $exp(T^{i}[j])$ and $map(T^{i}[j])$ are equivalent to $exp^{1}(T^{i}[j])$ and $map^{1}(T^{i}[j])$, respectively. 

We use the term \emph{parsing} to refer to a procedure that breaks $T^{i}[1..n^i]$ into a sequence of (possibly overlapping) substrings. The strings labelling the substring are the \emph{phrases} of the parsing. On the other hand, we use the term \emph{dictionary} to refer to an abstract associative data structure where the keys are strings and the associated values are integers.

Let $\Sigma=[1..\sigma]$ be an alphabet of $\sigma$ symbols and let $\mathcal{T}=\{T_1,\ldots,T_{k}\}$ be a collection of $k$ strings over the alphabet $\Sigma \setminus \{1\}$. The input for our algorithm is the sequence $T=T_{1}\texttt{\$}T_{2} \ldots T_{k}\texttt{\$} \in \Sigma^{*}$ of total length $n=|T|$ representing the concatenation of $\mathcal{T}$. The symbol $\texttt{\$}$ is a sentinel that we use as a boundary between consecutive strings in $T$. We set $\texttt{\$}=1$ to the smallest symbol in $\Sigma$.

\subsection{Overview of Our Algorithm}\label{sec:grlbwt_ov}

We call our algorithm for computing the BCR BWT of $\mathcal{T}$ $\mathsf{grlBWT}$. This method relies on the ideas developed by Nong et al.~in the \textsf{SA-IS} algorithm (Section~\ref{ssec:iss}) but includes elements of grammar and run-length compression (Section~\ref{ssec:comp} to reduce the space usage of the temporary data that $\mathsf{grlBWT}$ maintains in memory. Overall, this idea allows us to decrease working memory and computing time.

Similar to \textsf{SA-IS}, our method \textsf{grlBWT} is recursive. We start by parsing $T^{1}$ using a mechanism that relies on the symbol types of Section~\ref{ssec:iss} and stores the resulting parsing phrases in a set $\mathcal{F}^{1}$. Then, we use the set $\mathcal{S}^{1}$ with the strings labelling the suffixes of $\mathcal{F}^{1}$ to partition $S\!A^{1}$ such that each block $S\!A^{1}[s_x..e_x]$ encodes the suffixes of $T^{1}$ prefixed by the same string $S_x \in \mathcal{S}^{1}$. An important observation is that if $S_x$ is always preceded by the same symbol $c$ in the parsing phrases, we can compute the associated substring $BWT^{1}_{bcr}[s_x..e_x]=c^\ell$, with $\ell=e_x-s_x+1$, directly from $\mathcal{F}^{1}$. We use this idea to create a sparse version of $BWT_{bcr}$ that only contains symbols for blocks that meet our observation. We leave the other areas of $BWT_{bcr}$ empty for the moment. We refer to the sparsely populated version of $BWT_{bcr}$ as the \emph{preliminary} BCR BWT of $T^{1}$ ($pBWT^{1}$). To fill the \emph{unsolved} areas of $pBWT^{1}$, we create another string $T^{2}$ by replacing the substrings in the parsing of $T^{1}$ with their associated LMS orders in $\mathcal{F}^{1}$, and then apply the same procedure recursively over $T^{2}$. We keep recursing until we reach a base-case level $h$ where the input $T^{h}$ has $k$ symbols (i.e., the number of strings in $\mathcal{T}$). At this point, the BCR BWT of $T^{h}$ ($BWT^{h}_{bcr}$) is $T^{h}$ itself (we explain this idea in Section~\ref{ssec:indpha}). %This condition indicates that we have all the necessary information to fill the unsolved areas in $pBWT_{bcr}$.
We refer to the process of recursing from level $1$ to level $h$ as the \emph{parsing phase} of \textsf{grlBWT}.

\begin{algorithm}[!t]
\caption{Overview of \textsf{grlBWT}}
\label{algo:ovgrlbwt}{}
\begin{algorithmic}[1]
\small 
\Require $T[1..n]$, $k$ \Comment{$T$ is the concatenation of the $k$ string in $\mathcal{T}$}
    \State $i \gets 1$
    \State $T^{1} \gets T$
    \While{$length(T^{i})\neq k$}\label{code:pf_start}\Comment{parsing phase}
        \State Build $\mathcal{F}^{i}$ by parsing $T^{i}$ \label{code:lms_par}
        \State Produce $pBWT^{i}$ from $\mathcal{F}^{i}$ and satellite data\label{code:get_pbwt} 
        \State Encode $\mathcal{F}^{i}$ using grammar compression\label{code:gram_comp_set}
        \State Build $T^{i+1}$ using $\mathcal{F}^{i}$ and $T^{i}$
        \State Store $pBWT^{i}$ and $\mathcal{F}^{i}$ on disk
        \State $i \gets i+1$
    \EndWhile\label{code:pf_end}
    
    \State $BWT^{i}_{bcr} \gets T^{i}[1..k]$\label{code:ind_phase_start} \Comment{induction phase}
    \State $i \gets i-1$
    \While{$i>0$}\label{code:if_start}
        \State Load $\mathcal{F}^{i}$ from disk to main memory
        \State Build $P^{i}$ using $BWT^{i+1}_{bcr}$ and $\mathcal{F}^{i}$
        \State Merge $P^{i}$ and $pBWT^{i}$ to produce $BWT^{i}_{bcr}$
        \State $i \gets i-1$
    \EndWhile \label{code:if_end}
    \State \textbf{return} $BWT^{1}_{bcr}$ \Comment{the BCR BWT of $\mathcal{T}$}
    
\end{algorithmic}
\end{algorithm}

The parsing phase produced the string $BWT^{h}_{bcr}$, the preliminary BWTs $pBWT^{1}, pBWT^{2},\ldots, pBWT^{h-1}$, and the sets $\mathcal{F}^{1},\mathcal{F}^{2},\ldots, \mathcal{F}^{h-1}$. Now we need to go back in the recursions to complete the execution of \textsf{grlBWT}. When we return to level $i<h$, we use $BWT^{i+1}_{bcr}$ and $\mathcal{F}^{i}$ to induce the symbols in the unsolved blocks of $pBWT^{i}$, thus producing $BWT_{bcr}^{i}$ (the BCR BWT of $T^{i}$). We keep $BWT_{bcr}^{i+1}$ and $pBWT^{i}$ in run-length-compressed format to reduce space usage and speed up the computation of $BWT^{i}_{bcr}$. Our induction procedure reads a position $BWT^{i+1}_{bcr}[j]$ and uses its information to insert symbols in multiple unsolved blocks of $pBWT^{i}$. However, equal symbols of $BWT^{i+1}_{bcr}$ yield the same information and for the same unsolved blocks of $pBWT^{i}$. Thus, if we have a run $BWT^{i+1}_{bcr}[j..j+\ell-1]=o^{\ell}$ of $\ell$ consecutive copies of $o \in \Sigma^{i+1}$, we perform the induction from $o$ only once and copy the result in $pBWT^{i}$ $\ell$ times. When we reach the recursion level $i=1$ again, $BWT^{1}_{bcr}$ becomes the BCR BWT of $\mathcal{T}$. We refer to the process of returning from recursion $h$ to recursion $1$ as the \emph{induction phase} of \textsf{grlBWT}.

\paragraph{Practical Considerations} We implement \textsf{grlBWT} as a semi-external algorithm that executes the recursion levels as iterations in a loop. In every iteration $i$, we perform the steps of recursion level $i$, keeping the information of the other levels on disk. When we execute iteration $i$ in the parsing phase, we scan $T^{i}$ linearly from the disk to produce $\mathcal{F}^{i}$. The only elements that are always in main memory are $\mathcal{F}^{i}$ and a couple of satellite vectors that help us to build $pBWT^{i}$, which is also accessed linearly from disk during its construction. Then, we compute $T^{i+1}$ from $\mathcal{F}^{i}$ externally and store $\mathcal{F}^{i}$ on disk to use it later. We encode $\mathcal{F}^{i}$ with a scheme similar to grammar compression to reduce space usage and facilitate the remaining computations. When we return to level $i$ in the induction phase, we access $BWT^{i+1}_{bcr}$ and $pBWT^{i}$ linearly from the disk. The only elements in main memory are $\mathcal{F}^{i}$ and a run-length-compressed vector $P^{i}$ that keeps the symbols we need to introduce in the unsolved blocks of $pBWT^{i}$. We produce $BWT^{i}_{bcr}$ by merging $pBWT^{i}$ and $P^{i}$ in a semi-external way. 
Algorithm~\ref{algo:ovgrlbwt} presents a general overview of \textsf{grlBWT}. 

\subsection{The Parsing Phase}\label{ssec:parpha}

This section explains the steps we perform in one iteration of the parsing phase (Lines~\ref{code:pf_start}-\ref{code:pf_end} of Algorithm~\ref{algo:ovgrlbwt}). The inputs for the iteration are $T^{i}$ and a bit vector $B^{i}[1..\sigma^{i}]$ that indicates with $B^{i}[c]=1$ if symbol $c \in \Sigma$ expands to a string $exp(c) \in \Sigma^{*}$ suffixed by the sentinel \texttt{\$}. The output of the iteration is the preliminary BCR BWT of $T^{i}$ ($pBWT^{i}$) and a compressed version of $\mathcal{F}^{i}$. 

\subsubsection{LMS parsing}

We start the iteration by producing the set $\mathcal{F}^{i}$ with the distinct parsing phrases of $T^{i}$ (Line~\ref{code:lms_par} of Algorithm~\ref{algo:ovgrlbwt}). Our mechanism to break $T^{i}$ uses the symbol types of Section~\ref{ssec:iss}, but includes some modifications to take into consideration that $T^{i}$ is a (compressed) collection rather than a single string.

\begin{definition}\label{def:lms_brk}
LMS breaks: given a string $T^{i}$ over the alphabet $\Sigma^{i}[1..\sigma^{i}]$, its sequence of LMS breaks is a set of strictly increasing integers $\mathcal{B}^{i} = \{1<j_1<j_2 < \ldots < j_z<n^i < n^{i}+1\}$ such that each $T^{i}[j_p]$, with $p \in [1..z]$, is either (i) an LMS-type position, (ii) a position where $map(T^{i}[j_p])=T^{1}[u..u']$ is suffixed by a sentinel $\texttt{\$}=T^{1}[u']$, or (iii) a position where $map(T^{i}[j_p])=T^{1}[u..u']$ is preceded by a sentinel $T^{1}[u-1]=\texttt{\$}$. Positions $1,n^{i}+1$, and those meeting condition (iii) are left-border breaks. On the other hand,  $n^{i}$ and the positions meeting condition (ii) are right-border breaks.  
\end{definition}

\begin{definition}\label{def:lms_par}
LMS parsing: the parsing of $T^{i}$ induced by consecutive breaks of $\mathcal{B}^{i}$. Let $(j_{p},j_{p+1}) \in \mathcal{B}^{i}$ be two consecutive breaks. They form the substring $T^{i}[j_{p}..j_{p+1}]$ in the LMS parsing if the following conditions hold: (i) $T^{i}[j_{p}]$ is a left-border break or an LMS-type position, and (ii) $T^{i}[j_{p+1}]$ is a right-border break or an LMS-type position. Additionally, $(j_{p},j_{p+1})$ forms the phrase $T^{i}[j_{p}..j_{p+1}-1]$ instead of $T^{i}[j_{p}..j_{p+1}]$ iff $j_{p+1}=j_{p}+1$ and both $T^{i}[j_{p}]$ and $T^{i}[j_{p+1}]$ are left-border breaks.
If $(j_{p},j_{p+1})$ does not meet any of the conditions above, it does not produce a substring in the LMS parsing.
\end{definition}

Intuitively, our definition of LMS parsing prevents the formation of phrases in $T^{i}$ that cover multiple strings of $\mathcal{T}$. Specifically, the parsing will never yield a substring $T^{i}[j..j']$ such that $map(T^{i}[j])=T^{1}[u..u']$ falls within the boundaries of a string $T_{x} \in \mathcal{T}$ and $map(T^{i}[j'])=T^{1}[o..o']$ falls within the boundaries of another string $T_y \in \mathcal{T}$. Figure~\ref{fig:lms_par} depicts an example of LMS parsing.

\begin{figure}
\resizebox{\textwidth}{!}{%
\begin{tikzpicture}[>=stealth,thick,baseline]

\matrix (m2) [matrix of nodes, ampersand replacement=\|,
              align=left,
              nodes={inner sep=2mm, anchor=base},
              row 2/.style={nodes={gray, font=\footnotesize\bfseries}},
              row 3/.style={nodes={gray, font=\footnotesize}} ] at (0,0) { 
     $T^{1}$ \| $=$ \| \texttt{g} \| \texttt{t} \| \texttt{a} \| \texttt{c} \| \texttt{c} \| \texttt{\$} \| \texttt{g} \| \texttt{t} \| \texttt{a} \|  \texttt{a} \| \texttt{t} \| \texttt{a} \| \texttt{g} \| \texttt{t} \| \texttt{a} \| \texttt{c} \| \texttt{c} \| \texttt{\$} \\
             \|     \|  $1$ \| $2$ \| $3$ \| $4$ \| $5$ \| $6$ \| $7$ \| $8$ \| $9$ \| $10$ \| $11$ \| $12$ \| $13$ \| $14$ \| $15$ \| $16$ \| $17$ \| $18$ \\
             \|     \|  $S$ \| $L$ \| $S*$ \| $L$ \| $L$ \|  \| $S$ \| $L$ \| $S*$ \| $S$ \| $L$ \| $S*$ \| $S$ \| $L$ \| $S*$ \| $L$ \| $L$ \|  \\
};
\draw[gray] ($(m2-1-3.base west)+(0,-2mm)$) -- ($(m2-1-5.base east)+(0,-2mm)$);
\draw[gray] ($(m2-1-5.base west)+(0,5mm)$) -- ($(m2-1-8.base east)+(0,5mm)$);
\draw[gray] ($(m2-1-9.base west)+(0,-2mm)$) -- ($(m2-1-11.base east)+(0,-2mm)$);
\draw[gray] ($(m2-1-11.base west)+(0,5mm)$) -- ($(m2-1-14.base east)+(0,5mm)$);
\draw[gray] ($(m2-1-14.base west)+(0,-2mm)$) -- ($(m2-1-17.base east)+(0,-2mm)$);
\draw[gray] ($(m2-1-17.base west)+(0,5mm)$) -- ($(m2-1-20.base east)+(0,5mm)$);

\path[pattern=north west lines, pattern color=gray!60] ($(m2-2-3.north west)+(1.3mm,-1.3mm)$) rectangle ($(m2-2-3.south east)-(2.5mm,-1.3mm)$);
\path[pattern=north west lines, pattern color=gray!60] ($(m2-2-5.north west)+(1.3mm,-1.3mm)$) rectangle ($(m2-2-5.south east)-(2.5mm,-1.3mm)$);
\path[pattern=north west lines, pattern color=gray!60] ($(m2-2-8.north west)+(1.3mm,-1.3mm)$) rectangle ($(m2-2-8.south east)-(2.5mm,-1.3mm)$);
\path[pattern=north west lines, pattern color=gray!60] ($(m2-2-9.north west)+(1.3mm,-1.3mm)$) rectangle ($(m2-2-9.south east)-(2.5mm,-1.3mm)$);
\path[pattern=north west lines, pattern color=gray!60] ($(m2-2-11.north west)+(1.3mm,-1.3mm)$) rectangle ($(m2-2-11.south east)-(2.5mm,-1.3mm)$);
\path[pattern=north west lines, pattern color=gray!60] ($(m2-2-14.north west)+(1.3mm,-1.3mm)$) rectangle ($(m2-2-14.south east)-(2.5mm,-1.3mm)$);
\path[pattern=north west lines, pattern color=gray!60] ($(m2-2-17.north west)+(1.3mm,-1.3mm)$) rectangle ($(m2-2-17.south east)-(2.5mm,-1.3mm)$);
\path[pattern=north west lines, pattern color=gray!60] ($(m2-2-20.north west)+(1.3mm,-1.3mm)$) rectangle ($(m2-2-20.south east)-(2.5mm,-1.3mm)$);

\matrix (m3) [matrix of nodes, ampersand replacement=\|,
              align=center,
              nodes={inner sep=2mm, anchor=base},
              row 2/.style={nodes={gray, font=\footnotesize\bfseries}},
              row 3/.style={nodes={gray, font=\footnotesize}}, above= 1mm of m2 ] { 
     $T^{2}$ \| $=$ \| 4   \| 2 \| 4   \| 1    \| 3   \| 2 \\
             \|     \| $1$\| $2$ \| $3$   \| $4$    \| $5$   \| $6$ \\
             \|     \| $L$ \|   \| $L$ \| $S*$ \| $L$ \|   \\
};
\draw[gray] ($(m3-1-3.base west)+(0,-2mm)$) -- ($(m3-1-4.base east)+(0,-2mm)$);
\draw[gray] ($(m3-1-5.base west)+(0,5mm)$) -- ($(m3-1-6.base east)+(0,5mm)$);
\draw[gray] ($(m3-1-6.base west)+(0,-2mm)$) -- ($(m3-1-8.base east)+(0,-2mm)$);

\path[pattern=north west lines, pattern color=gray!60] ($(m3-2-3.north west)+(1.3mm,-1.3mm)$) rectangle ($(m3-2-3.south east)-(1.3mm,-1.3mm)$);
\path[pattern=north west lines, pattern color=gray!60] ($(m3-2-4.north west)+(1.3mm,-1.3mm)$) rectangle ($(m3-2-4.south east)-(2.5mm,-1.3mm)$);
\path[pattern=north west lines, pattern color=gray!60] ($(m3-2-5.north west)+(1.3mm,-1.3mm)$) rectangle ($(m3-2-5.south east)-(2.5mm,-1.3mm)$);
\path[pattern=north west lines, pattern color=gray!60] ($(m3-2-6.north west)+(1.3mm,-1.3mm)$) rectangle ($(m3-2-6.south east)-(2.5mm,-1.3mm)$);
\path[pattern=north west lines, pattern color=gray!60] ($(m3-2-8.north west)+(1.3mm,-1.3mm)$) rectangle ($(m3-2-8.south east)-(2.5mm,-1.3mm)$);

\matrix (m4) [matrix of nodes, ampersand replacement=\|,
              align=center,
              nodes={inner sep=2mm, anchor=base},
              row 2/.style={nodes={gray, font=\footnotesize\bfseries}},
              row 3/.style={nodes={gray, font=\footnotesize}}, above= 1mm of m3 ] { 
     $T^{3}$ \| $=$ \| 3   \| 2  \| 1   \\
             \|     \| $1$\| $2$ \| $3$ \\
             \|     \|    \| $L$ \|    \\
};
\draw[gray] ($(m4-1-3.base west)+(0,-2mm)$) -- ($(m4-1-3.base east)+(0,-2mm)$);
\draw[gray] ($(m4-1-4.base west)+(0,5mm)$) -- ($(m4-1-5.base east)+(0,5mm)$);

\path[pattern=north west lines, pattern color=gray!60] ($(m4-2-3.north west)+(1.3mm,-1.3mm)$) rectangle ($(m4-2-3.south east)-(1.3mm,-1.3mm)$);
\path[pattern=north west lines, pattern color=gray!60] ($(m4-2-4.north west)+(1.3mm,-1.3mm)$) rectangle ($(m4-2-4.south east)-(2.5mm,-1.3mm)$);
\path[pattern=north west lines, pattern color=gray!60] ($(m4-2-5.north west)+(1.3mm,-1.3mm)$) rectangle ($(m4-2-5.south east)-(2.5mm,-1.3mm)$);

\matrix (m5) [matrix of nodes, ampersand replacement=\|,
              align=center,
              nodes={inner sep=2mm, anchor=base},
              row 2/.style={nodes={gray, font=\footnotesize\bfseries}},
              row 3/.style={nodes={gray, font=\footnotesize}}, above= 1mm of m4 ] { 
     $T^{4}$ \| $=$ \| 2   \| 1  \\
             \|     \| $1$\| $2$ \\
};

\path[pattern=north west lines, pattern color=gray!60] ($(m5-2-3.north west)+(1.3mm,-1.3mm)$) rectangle ($(m5-2-3.south east)-(1.3mm,-1.3mm)$);
\path[pattern=north west lines, pattern color=gray!60] ($(m5-2-4.north west)+(1.3mm,-1.3mm)$) rectangle ($(m5-2-4.south east)-(2.5mm,-1.3mm)$);
\end{tikzpicture}
}
\caption{Successive rounds of LMS parsing for the text $T^{1}=\texttt{gtacc\$gtaatagtacc\$}$. The symbols $L$, $S$ and $S*$ represent the L-type, S-type and LMS-type positions of $T^{i}$, respectively. The dashed boxes indicate the breaks in $\mathcal{B}^{i}$ (Definition~\ref{def:lms_brk}). Positions $1$ and $7$ of $T^{1}$ are left-border breaks, while positions $6$ and $18$ are right-border breaks. Similarly, positions $1$ and $3$ of $T^{2}$ are left-border breaks, and positions $2$ and $6$ are right-border breaks. The horizontal lines indicate the substrings induced by $\mathcal{B}^{i}$ as described in Definition~\ref{def:lms_par}. Each $T^{i}[j]$ is the LMS order  in $\mathcal{F}^{i}$ of the $jth$ substring in the LMS parsing of $T^{i-1}$. The parsing rounds end when the number of symbols in $T^{i}$ equals the number of strings in $\mathcal{T}$.} 
\label{fig:lms_par}
\end{figure}

Our procedure to build $\mathcal{F}^{i}$ works as follows:
we initialize a dictionary $D^{i}$ where each key $F \in \mathcal{F}^{i}$ in $D^{i}$ will be associated with an integer that indicates the number of times $F$ occurs as a phrase in the LMS parsing of $T^{i}$. We fill the dictionary by accessing the breaks in $\mathcal{B}^{i}$ in one right-to-left scan of $T^{i}$ that enumerates the LMS-type positions. We also use $B^{i}$ during the scan to detect right and left border breaks. For each pair of consecutive integers $(j_{p},j_{p+1}) \in \mathcal{B}^{i}$ forming a valid substring in the LMS parsing (see Definition~\ref{def:lms_par}), we record an occurrence of the phrase $F=T^{i}[j_{p}..j_{p+1}]$ in $D^{i}$ (or $F=T^{i}[j_{p}..j_{p+1}-1]$ if both positions are left-border breaks). If $F$ does not exist as a key, we create a new key-value entry $(F, 1)$ in $D^{i}$. On the other hand, if $F$ already exists in the keys, we just increase its associated value by one.

\subsubsection{The Generalized Suffix Array}

As explained in the overview of \textsf{grlBWT} (Section~\ref{sec:grlbwt_ov}), every recursion level $i$ constructs the vector $BWT^{i}_{bcr}$ of $T^{i}$ (besides performing LMS parsing). This step requires us to define the generalized suffix array \sa{} of $T^{i}$ to fit the description of the BCR BWT.

\begin{definition}\label{def:gsa}
Generalized suffix array of $T^{i}$. Consider a partition over $T^{i}[1..n^{i}]$ such that each block $T^{i}[j..j']$ is a substring where (i) $T^{i}[j]$ is a left-border break in $\mathcal{B}^{i}$, (ii) $T^{i}[j']$ is a right-border break in $\mathcal{B}^{i}$, and (iii) $exp(T^{i}[j..j'])=T_{x}\texttt{\$}$ expands to a string $T_{x} \in \mathcal{T}$.~$\same[1..n^{i}]$ is a vector that enumerates the suffixes of the blocks $T^{i}[j..j']$ in lexicographical order. When two suffixes $T^{i}[u..j']=T^{i}[u'..p']$ from different blocks $T^{i}[j..j']$ and $T^{i}[p..p']$ spell the same sequence, the ties are broken by $min(u, u')$. 
\end{definition}

We remark that our algorithm never constructs \sa, but it uses its definition to produce $BWT^{i}_{bcr}$. Figure~\ref{fig:prBWT_examp} shows an example of \sa. 

\subsubsection{Constructing the Preliminary BWT}\label{sec:c_preBWT}

The next step in iteration $i$ is to construct the preliminary BWT of $T^{i}$ from $D^{i}$ (Line~\ref{code:get_pbwt} of Algorithm~\ref{algo:ovgrlbwt}). Keep in mind that $D^{i}$ is an encoding for $\mathcal{F}^{i}$ that includes the frequencies in $T^{i}$ of the LMS parsing phrases.~We use the following observations to carry out the construction:

\begin{lemma}\label{lem:lem1}
Let $S_x[1..n_x]$ and $S_y[1..n_y]$ be two different strings over the alphabet $\Sigma^{i}$. Assume both occur as suffixes in one or more phrases of $\mathcal{F}^{i}$ and meet one of the following conditions: (i) $n_x>1$ or (ii) $n_x=1$ and the sequence $exp(S_x)$ is  suffixed by the sentinel  $\texttt{\$}$ (the same applies for $S_y$). Let $O_x$ be the set of positions in $T^{i}$ where $S_x$ occurs as a suffix of a parsing phrase. Specifically, each $j \in O_x$ is a position such that $S_x=T^{i}[j..j+n_x-1]$ and $T^{i}[j-j'..j+n_x-1]$, with $j'\geq 0$, is a substring in the LMS parsing. Let us define an equivalent set $O_y$ for $S_y$. If $S_x \prec_{LMS} S_y$ (see Section~\ref{ssec:iss}), then all the suffixes of $T^{i}$ starting at positions in $O_x$ are lexicographically smaller than those starting at positions in $O_y$. 
\end{lemma}

\begin{proof}
When $S_x$ is not a prefix of $S_y$ (and vice versa), the demonstration of this lemma is trivial: we compare the sequences of $S_x$ and $S_y$ from left to right until we find a mismatching position $u$ (i.e., $S_x[u]\neq S_y[u]$). If $S_x[u] < S_y[u]$, we know that all the suffixes in $O_x$ are lexicographically smaller than those in $O_y$ because this is how lexicographical sorting works. In the other scenario, when one string is a prefix of the other, we can not use this mechanism as we will not find a mismatching position. However, this scenario is not possible when $S_x$ or $S_y$ has length $1$ and expands to a string in $\Sigma^{*}$ suffixed by $\texttt{\$}$. For instance, if $S_x$ were a prefix of $S_y$, it would imply that $exp(S_y)$ contains a sentinel in a position that is not the end, but Definition~\ref{def:lms_par} (LMS parsing) prevents that situation from happening.  

When both $S_x$ and $S_y$ have length $>1$ and one is a prefix of the other, we resort to the symbol types of Section~\ref{ssec:iss}. We assume for this proof that $S_y$ is a prefix of $S_x$, but the other way is equivalent. We know that $S_y[n_y]$ and $S_x[n_y]$ have different types. On the one hand, $S_y[n_y]$ is LMS-type because $S_y$ is a suffix of an LMS substring, and $n_y$ is the last position of $S_y$. On the other hand, $S_x[n_y]$ is L-type because if it were S-type, then it would also be LMS-type, and thus $S_x[1..n_y]$ would be an occurrence for $S_y$. This observation is due to $S_y[n_y-1]=S_x[n_y-1]$ being L-type. Given the types of $S_y[n_y]$ and $S_x[n_y]$, the occurrences of $S_y$ in $O_y$ are always followed in $T^{i}$ by symbols that are greater than $S_x[n_y+1]$, meaning that the suffixes of $T^{i}$ starting at positions in $O_y$ are lexicographically greater than the suffixes starting at positions in $O_x$. This observation does not hold when $S_x$ or $S_y$ have length one: $S_y[n_y]$ equals $S_x[1]$, and both are LMS-type, so there is not enough information to decide the lexicographical order of the suffixes in $O_x$ and $O_y$. Figure~\ref{fig:lem1exp} shows an example of this lemma.
\end{proof}

\begin{figure}[t]
\centering
\resizebox{0.4\textwidth}{!}{%
\begin{tikzpicture}[>=stealth,thick,baseline]
    \matrix (m) [matrix of nodes, ampersand replacement=\|, align=left]  at (0,0) { 
     \|  \|  $\texttt{a}$ \| $\texttt{c}$  \| $\texttt{t}$\tiny\textsubscript{$L$} \| $\texttt{g}$\tiny\textsubscript{$L$}\| \texttt{g} \| \texttt{a}\tiny\textsubscript{$S*$}\| \texttt{c}\| \dots  \\
     \|  \|  $\texttt{a}$ \| $\texttt{c}$  \| $\texttt{t}$\tiny\textsubscript{$L$} \| $\texttt{g}$\tiny\textsubscript{$L$}\| \texttt{g} \| \texttt{a}\tiny\textsubscript{$S*$}\| \texttt{c}\| \dots  \\
     \|  \|  $\texttt{a}$ \| $\texttt{c}$  \| $\texttt{t}$\tiny\textsubscript{$L$} \| $\texttt{g}$\tiny\textsubscript{$L$}\| \texttt{g} \| \texttt{a}\tiny\textsubscript{$S*$}\| \texttt{c}\| \dots  \\
     \|  \|  $\texttt{a}$ \| $\texttt{c}$  \| $\texttt{t}$\tiny\textsubscript{$L$} \| $\texttt{g}$\tiny\textsubscript{$S*$} \| \texttt{t} \| \dots  \\
     \|  \|  $\texttt{a}$ \| $\texttt{c}$  \| $\texttt{t}$\tiny\textsubscript{$L$} \| $\texttt{g}$\tiny\textsubscript{$S*$} \| \texttt{t} \| \dots  \\
     \|  \|  $\texttt{a}$ \| $\texttt{c}$  \| $\texttt{t}$\tiny\textsubscript{$L$} \| $\texttt{g}$\tiny\textsubscript{$S*$} \| \texttt{t} \| \dots  \\
};
\draw (m-1-3.north west) rectangle (m-3-8.south east);
\draw (m-4-3.north west) rectangle (m-6-6.south east);

\node[anchor=north] at (m-6-6.south) {$S_y[n_y]$};
\node[anchor=south] at (m-1-6.north) {$S_x[n_y]$};
\end{tikzpicture}
}
\caption{Example of Lemma~\ref{lem:lem1}. The figure shows a subset of suffixes for some string $T^{1}$ sorted in lexicographical order. The upper box is the set $O_x$ of suffixes in $T^{i}$ prefixed by the string $S_x=\texttt{actgga}$ that occurs as a suffix in some LMS parsing phrases of $T^{1}$. Equivalently, the lower box is the set $O_y$ of suffixes prefixed by $S_y=\texttt{actg}$, which also occurs as a suffix in some LMS substrings of $T^{1}$. The L-type and LMS-type positions of $T^{1}$ are marked with the symbols $L$ and $S*$ in the figure, respectively. Because $S_x \prec_{LMS} S_y$, all the suffixes in $O_y$ are lexicographically greater than those in $O_x$.}
\label{fig:lem1exp}
\end{figure}

The consequence of Lemma~\ref{lem:lem1} is that the suffixes of length $>1$ in $\mathcal{F}^{i}$ induce a partition over \sa{} (Definition~\ref{def:gsa}):

\begin{lemma}\label{lem:lem2}
Let $\mathcal{S}^{i}=\{S_1,S_2,\ldots,S_{n'}\}$ be the set of strings where each element $S_x[1..n_x]$ occurs as a suffix in the phrases of $\mathcal{F}^{i}$ and either (i) has length $n_x>1$ or (ii) has length $n_x=1$ but $exp(S_x)$ is suffixed by $\texttt{\$}$. Additionally, let $\mathcal{O}=\{O_{1}, O_{2},\ldots,O_{n'}\}$ be the set of occurrences in $T^{i}$ for the strings in $\mathcal{S}$. For every $S_{x} \in \mathcal{S}^{i}$, its associated set $O_{x} \in \mathcal{O}$ stores each position $j$ such that $T^{i}[j..j+n_x-1]$ is an occurrence of $S_{x}$ and $T^{i}[j-j'..j+n_x-1]$, with $j' \geq 0$, is a phrase in the LMS parsing. It holds that $\mathcal{O}$ is a partition over \sa{} as the lexicographical sorting places the elements of each $O_{x} \in \mathcal{O}$ in a consecutive range of \sa. 
\end{lemma}

\begin{proof}
We demonstrate the lemma by showing that the lexicographical sorting does not interleave suffixes of $T^{i}$ in \sa{} that belong to different sets of $\mathcal{O}$. Assume the string $S_{x} \in \mathcal{S}^{i}$, associated with the set $O_{x} \in \mathcal{O}$, is a prefix in another string $S_{y} \in \mathcal{S}^{i}$, which in turn is associated with the set $O_{y} \in \mathcal{O}$. Even though we do not know the symbols that occur to the right of $S_{x}$ in its occurrences of $O_{x}$, we do know that both $S_{x}$ and $S_{y}$ are suffixes of substrings in the LMS parsing, and by Lemma~\ref{lem:lem1}, we know that all the suffixes of $T^{i}$ in $O_{x}$ are lexicographically greater than the suffixes in $O_{y}$. Hence, the interleaving of suffixes in \sa{} from different sets of $\mathcal{O}$ is not possible, even if $\mathcal{S}^{i}$ is not a prefix-free set. 
\end{proof}

Lemma~\ref{lem:lem2} gives us a simple way to construct the preliminary BWT of $T^{i}$ from $\mathcal{F}^{i}$. Our explanation requires projecting the partition of \sa{} induced by $\mathcal{S}^{i}$ to $BWT^{i}_{bcr}$ (the BCR BWT of $\mathcal{T}^{i}$) such that, if $\same[s_x..e_x]$ is the block formed by $S_{x} \in \mathcal{S}^{i}$, then $BWT^{i}_{bcr}[s_x..e_x]$ is the projected block. There are three cases to consider:

\begin{lemma}\label{lem:lem3}
Let $S_x \in \mathcal{S}^{i}$ be the string with LMS order $o$ among all the elements of $\mathcal{S}^{i}$. If $S_x$ is left-maximal in $\mathcal{F}^{i}$, then the $oth$ block $BWT^{i}_{bcr}[s_x..e_x]$ in the projected partition contains more than one distinct symbol. Therefore, $\mathcal{F}^{i}$ does not have enough information to compute the sequence of symbols in $BWT^{i}_{bcr}[s_x..e_x]$.  
\end{lemma}

\begin{proof}
Let $F$ and $F'$ be two phrases of $\mathcal{F}^{i}$ where $S_x$ occurs as a suffix. Assume the left symbol of $S_x$ in $F$ is $c \in \Sigma^{i}$ and the left symbol in $F'$ is $c' \in \Sigma^{i}$. In this scenario, the relative order of $c$ and $c'$ is not decided by $S_x$, but for the sequences that occur to the right of $F$ and $F'$ in $T^{i}$. However, those sequences are not accessible directly from $\mathcal{F}^{i}$. Hence, it is not possible to decide the order of $c$ and $c'$ in $BWT^{i}_{bcr}[s_x..e_x]$.  
\end{proof}

\begin{lemma}\label{lem:lem4}
Consider the string $S_x \in \mathcal{S}^{i}$ of Lemma~\ref{lem:lem3}. When $S_{x}$ only occurs as a non-proper suffix in a phrase $S_x=F \in \mathcal{F}^{i}$, it is not possible to complete the sequence of symbols in $BWT^{i}_{bcr}[s_x..e_x]$.
\end{lemma}

\begin{proof}
The symbols that occur to the left of $S_x$ in $T^{i}$ are in the substrings of the LMS parsing that precede $F$ in $T^{i}$. However, it is not possible to know from $\mathcal{F}^{i}$ the sequences of those substrings.
\end{proof}

\begin{figure}[!t]
\centering
\resizebox{0.7\textwidth}{!}{%
\begin{tikzpicture}[>=stealth,thick,baseline]

\matrix (m4) [matrix of nodes, ampersand replacement=\|] at (0,0)  { 
     $\mathcal{F}^{1}$ \| $=$ \|  $\{\texttt{gta}$, \| \texttt{acc\$}, \| \texttt{aata}, \| $\texttt{agta}\}$ \\
};

\matrix (m5) [matrix of nodes, ampersand replacement=\|, below= 0.1cm of m4]  { 
     $\mathcal{S}^{1}$ \| $=$ \|  $\{\texttt{\$}$, \| \texttt{aata}\| \texttt{acc\$}, \| \texttt{agta}, \| \texttt{ata}, \| \texttt{c\$}, \| \texttt{cc\$}, \| \texttt{gta}, \| $\texttt{ta}\}$ \\
};

\matrix (m3) [matrix of nodes, ampersand replacement=\|, below= 2.2cm of m5,
              row 2/.style={nodes=draw, white},
              ]  { 
     $GS\!A^{1}$      \| $=$ \|  6 \| 18 \| 9 \| 3 \| 15\| 12\| 10\| 5\| 17 \| 4 \| 16 \| 7 \| 1 \| 13 \| 8 \| 2 \| 14 \| 11 \\
                     \|  a  \|    \|    \|   \|   \|   \|   \|   \|  \|    \|   \|    \|   \|   \|    \|   \|   \|    \|    \\
     $BWT^{1}_{bcr}$ \| $=$ \|  \texttt{c} \| \texttt{c}  \| \texttt{t} \| \texttt{t} \| \texttt{t} \| \texttt{t} \| \texttt{a} \| \texttt{c} \| \texttt{c}  \| \texttt{a} \| \texttt{a}  \| \texttt{\$}\| \texttt{\$}\| \texttt{a}  \| \texttt{g} \| \texttt{g} \| \texttt{g}  \| \texttt{a}  \\
                     \||[white]|  a  \|    \|    \|   \|   \|   \|   \|   \|  \|    \|   \|    \|   \|   \|    \|   \|   \|    \|    \\
     $pBWT^{1}$ \| $=$ \|  \texttt{c} \| \texttt{c}  \| \texttt{*} \| \texttt{*} \| \texttt{*} \| \texttt{*} \| \texttt{a} \| \texttt{c} \| \texttt{c}  \| \texttt{a} \| \texttt{a}  \| \texttt{\#}\| \texttt{\#}\| \texttt{\#}  \| \texttt{\#} \| \texttt{\#} \| \texttt{\#}  \| \texttt{\#}  \\
};

\node[anchor = west, rotate = 90, gray] at (m3-1-3.north) {\texttt{\$}};
\node[anchor = west, rotate = 90, gray] at (m3-1-4.north) {\texttt{\$}};
\node[anchor = west, rotate = 90, gray] at (m3-1-5.north) {\texttt{aata}\textcolor{black}{\texttt{gta}} $\ldots$};
\node[anchor = west, rotate = 90, gray] at (m3-1-6.north) {\texttt{acc\$}};
\node[anchor = west, rotate = 90, gray] at (m3-1-7.north) {\texttt{acc\$}};
\node[anchor = west, rotate = 90, gray] at (m3-1-8.north) {\texttt{agta}\textcolor{black}{\texttt{cc\$}}};
\node[anchor = west, rotate = 90, gray] at (m3-1-9.north) {\texttt{ata} \textcolor{black}{\texttt{gta}} $\ldots$};
\node[anchor = west, rotate = 90, gray] at (m3-1-10.north) {\texttt{c\$}};
\node[anchor = west, rotate = 90, gray] at (m3-1-11.north) {\texttt{c\$}};
\node[anchor = west, rotate = 90, gray] at (m3-1-12.north) {\texttt{cc\$}};
\node[anchor = west, rotate = 90, gray] at (m3-1-13.north) {\texttt{cc\$}};
\node[anchor = west, rotate = 90, gray] at (m3-1-14.north) {\texttt{gta}\textcolor{black}{\texttt{ata}} $\ldots$ };
\node[anchor = west, rotate = 90, gray] at (m3-1-15.north) {\texttt{gta}\textcolor{black}{\texttt{cc\$}}};
\node[anchor = west, rotate = 90, gray] at (m3-1-16.north) {\texttt{gta}\textcolor{black}{\texttt{cc\$}}};
\node[anchor = west, rotate = 90, gray] at (m3-1-17.north) {\texttt{ta}\textcolor{black}{\texttt{ata}} $\ldots$};
\node[anchor = west, rotate = 90, gray] at (m3-1-18.north) {\texttt{ta}\textcolor{black}{\texttt{cc\$}}};
\node[anchor = west, rotate = 90, gray] at (m3-1-19.north) {\texttt{ta}\textcolor{black}{\texttt{cc\$}}};
\node[anchor = west, rotate = 90, gray] at (m3-1-20.north) {\texttt{ta}\textcolor{black}{\texttt{gta}} $\ldots$};

\draw (m3-1-3.north west) rectangle (m3-1-4.south east);
\draw (m3-1-4.north east) rectangle (m3-1-5.south east);
\draw (m3-1-5.north east) rectangle (m3-1-7.south east);
\draw (m3-1-7.north east) rectangle (m3-1-8.south east);
\draw (m3-1-8.north east) rectangle (m3-1-9.south east);
\draw (m3-1-9.north east) rectangle (m3-1-11.south east);
\draw (m3-1-11.north east) rectangle (m3-1-13.south east);
\draw (m3-1-13.north east) rectangle (m3-1-16.south east);
\draw (m3-1-16.north east) rectangle (m3-1-20.south east);

\draw (m3-1-3.north west)+(0,-1.1cm) rectangle ([yshift=-1.1cm]m3-1-4.south east);
\draw[pattern=north west lines, pattern color=gray!60] (m3-1-4.north east) + (0,-1.1cm) rectangle ([yshift=-1.1cm] m3-1-5.south east);
\draw[pattern=north west lines, pattern color=gray!60] (m3-1-5.north east) + (0,-1.1cm) rectangle ([yshift=-1.1cm] m3-1-7.south east);
\draw[pattern=north west lines, pattern color=gray!60] (m3-1-7.north east) + (0,-1.1cm) rectangle ([yshift=-1.1cm] m3-1-8.south east);
\draw (m3-1-8.north east) + (0,-1.1cm) rectangle ([yshift=-1.1cm] m3-1-9.south east);
\draw (m3-1-9.north east) + (0,-1.1cm) rectangle ([yshift=-1.1cm] m3-1-11.south east);
\draw (m3-1-11.north east) + (0,-1.1cm) rectangle ([yshift=-1.1cm] m3-1-13.south east);
\draw[pattern=north west lines, pattern color=gray!60] (m3-1-13.north east) + (0,-1.1cm) rectangle ([yshift=-1.1cm] m3-1-16.south east);
\draw[pattern=north west lines, pattern color=gray!60] (m3-1-16.north east) + (0,-1.1cm) rectangle ([yshift=-1.1cm] m3-1-20.south east);
\end{tikzpicture}
}
\caption{Preliminary BCR BWT of $\mathcal{T}=\{\texttt{gtacc},\ \texttt{gtaatagtacc}\}$ of Figure~\ref{fig:lms_par}. \sa{} is the generalized suffix array of $\mathcal{T}$ built on top of $T^1=\texttt{gtacc\$gtaatagtacc\$}$. The vertical strings are the suffixes sorted in lexicographical order. The boxes in \sa{} indicate the partition induced by $\mathcal{S}^{i}$ (strings labelled in grey on top of \sa{}). The boxes in $BWT_{bcr}^{1}$ are the projected blocks in \sa's partition, with the dashed boxes being the blocks we can not solve using $\mathcal{F}^{1}$. The first three dashed boxes (from left to right) of $BWT^{i}_{bcr}$ are represented with $\texttt{*}$ in $pBWT^{1}$ because their blocks meet Lemma~\ref{lem:lem4}. In contrast, the last two dashed boxes are represented with \texttt{\#}, indicating they meet Lemma~\ref{lem:lem3}. The final run-length-compressed vector is $pBWT^{1}=(\texttt{c},2),(\texttt{*},4), (\texttt{a}, 1), (\texttt{c},2), (\texttt{a}, 2), (\texttt{\#},7)$.}
\label{fig:prBWT_examp}
\end{figure}

We now describe the information of the preliminary BWT that we can extract from $\mathcal{F}^{i}$: 

\begin{lemma}\label{lem:rl}
Let $S_x \in \mathcal{S}^{i}$ be the string of Lemma~\ref{lem:lem3}. Additionally, let $O_x \in \mathcal{O}$ be the set of occurrences of $S_x$ in $T^{i}$ as described in Lemma~\ref{lem:lem2}. If all the suffixes of $T^{i}$ in $O_x$ are preceded by the same symbol $c \in \Sigma^{i}$ (i.e., $S_x$ is not left-maximal), $BWT^{i}_{bcr}[s_x..e_x]=c^\ell$ is an equal-symbol run of length $\ell = |O_x|$.
\end{lemma}

\begin{proof}
By Lemma~\ref{lem:lem2}, we know that $S_x$ prefixes the suffixes of $T^i$ in $O_x$, and that they form a consecutive range $\same[s_x..e_x]$. Additionally, the symbols that occur to the left of the suffixes in $\same[s_x..e_x]$ are those for the projected block $BWT^{i}_{bcr}[s_x..e_x]$. However, we still have not resolved the relative order of the suffixes in $\same[s_x..e_x]$, so (in theory) we do not know how to rearrange the symbols in $BWT^{i}_{bcr}[s_x..e_x]$. The suffixes of $T^{i}$ in $O_x$ are preceded by the same symbol $c$, so it is not necessary to further sort $\same[s_x..e_x]$ because the outcome for the projected block will always be an equal-symbol run $c^{\ell}$ of length $\ell=|O_x|=e_x-s_x+1$.  
\end{proof}

Now we have all the necessary elements to describe the preliminary BCR BWT formally:

\begin{definition}
   Preliminary BCR BWT of $\mathcal{T}^{i}$: a vector $pBWT^{i}[1..n^{i}]$ over the alphabet $\Sigma^{i} \cup \{\texttt{\#}, \texttt{*}\}$, where $\texttt{\#}$ and $\texttt{*}$ are special symbols out of $\Sigma^{i}$ denoting \emph{unsolved} areas of $BWT^{i}_{bcr}$ for which we do not know the order of the symbols. Let $S_x \in \mathcal{S}^{i}$ be the string inducing the projected block $BWT^{i}_{bcr}[s_x..e_x]$, with $\ell=e_x-s_x+1$. If $BWT^{i}_{bcr}[s_x..e_x]$ meets Lemma~\ref{lem:lem3}, then $pBWT^{i}[s_x..e_x]=\texttt{\#}^{\ell}$. On the other hand, if $BWT^{i}_{bcr}[s_x..e_x]$ meets Lemma~\ref{lem:lem4}, then $pBWT^{i}[s_x..e_x]=\texttt{*}^{\ell}$. Finally, if $BWT^{i}_{bcr}[s_x..e_x]$ meets Lemma~\ref{lem:rl} and is always preceded by symbol $c \in \Sigma^{i}$ in $\mathcal{F}^{i}$, then $pBWT^{i}[s_x..e_x]=c^{\ell}$. The areas of $pBWT^{i}$ storing \texttt{\#} or \texttt{*} symbols are its unsolved blocks.
\end{definition}

We use two special symbols $\texttt{\$}, \texttt{\#} \notin \Sigma^{i}$ because the mechanism to fill the unsolved blocks of $pBWT^{i}$ that meet Lemma~\ref{lem:lem3} during the induction phase is different from the one to fill the unsolved blocks meeting Lemma~\ref{lem:lem4}. The difference will become clear in Section~\ref{ssec:indpha}. Additionally, we keep $pBWT^{i}$ in run-length-compressed format to reduce space usage. Despite our encoding, we will keep using the notation $pBWT^{i}[s_x..e_x]$ to refer to an uncompressed area of $pBWT^{i}$. Figure~\ref{fig:prBWT_examp} shows an example of the construction of $pBWT^{i}$ using Lemmas~\ref{lem:lem3},\ref{lem:lem4}, and \ref{lem:rl}. 

\paragraph{Succinct Dictionary Encoding} \label{par:suc_enc_dict} The construction of $pBWT^{i}$ starts by changing the representation of $D^{i}$ to a more convenient data structure. First, we concatenate all the keys of $D^{i}$ in one single vector $R^{i}[1..||\mathcal{F}^{i}||]$ without changing their relative order. We mark the boundaries of consecutive phrases in $R^{i}$ with a bit vector $L^{i}[1..||\mathcal{F}^{i}||]$ in which we set $L^{i}[j]=1$ if $R^{i}[j]$ is the first symbol of a phrase and set $L^{i}[j]=0$ otherwise. We store the values of $D^{i}$ in another integer vector $N^{i}[1..|\mathcal{F}^{i}|]$. We maintain the relative order so that the value $N^{i}[o]$ is associated in $D^{i}$ with $oth$ phrase we inserted into $R^{i}$. We also augment $L^{i}$ with a data structure that supports $\mathsf{rank}_{1}$ queries~\cite{ok07prac} to map each symbol $R^{i}[j]$ to its corresponding phrase in $D^{i}$. Thus, given the phrase $F=R^{i}[j..j'] \in \mathcal{F}^{i}$, we can obtain its associated value in $D^{i}$ as $N^{i}[\mathsf{rank}_{1}(L, p)]$, with $p \in [j..j']$. 

\paragraph{Generalized Suffix Array of the Parsing Set}\label{par:iss_dict} We will use $(R^{i}, L^{i})$ to compute an array $\gsame[1..|R^{i}|]$ that enumerates the suffixes of $\mathcal{F}^{i}$ in LMS order. This vector serves a double purpose as we will use it to get the unsolved blocks of $pBWT^{i}$ and the LMS orders of the phrases in $\mathcal{F}^{i}$. The values in \gsa{} are text positions in $R^{i}$, which we sort as follows: consider two substrings $R^{i}[j..j']$ and $R^{i}[p..p']$ encoding suffixes of $\mathcal{F}^{i}$ that are consecutive in \gsa. That is, $\gsame[u]=j$, $\gsame[u+1]=p$, $L^{i}[j'+1]=1$, and $L^{i}[p'+1]=1$. If $R^{i}[j..j'] \neq R^{i}[p..p']$, then $R^{i}[j..j'] \prec_{LMS} R^{i}[p..p']$. On the other hand, if $R^{i}[j..j'] = R^{i}[p..p']$, it implies $j<p$. We compute \gsa{} using a modified version of the ISS method described in Section~\ref{ssec:iss}. We first create an empty array $\gsame[1..|R^{i}|]$, which we divide into $\Sigma^{i}$ buckets. Then, we scan $R^{i}$ from left to right, and for each symbol $R^{i}[j]=c$ with $L^{i}[j+1]=1$ (i.e., the rightmost symbol of a phrase), we insert $j$ in the rightmost empty position in the bucket $c$ of $\gsame$. Then, in one left-to-right scan of \gsa{}, we perform the first step of ISS: for every $\gsame[u]$ such that $R^{i}[\gsame[u]-1]=c$ is L-type, we insert $\gsame[u]-1$ in the leftmost empty cell in the bucket $c$ of \gsa. In the next ISS step, we perform a right-to-left scan of \gsa.~This time, for every $\gsame[u]$ such that $R^{i}[\gsame[u]-1]=c$ is S-type, we insert $\gsame[u]-1$ in the rightmost empty cell in the bucket $c$ of \gsa.~In both ISS scans, if $L^{i}[\gsame[u]]=1$ (i.e., $R^{i}[\gsame[u]]$ is the leftmost symbol of a phrase), we skip the position as it does not induce the order of any suffix in $\mathcal{F}^{i}$. 

\begin{algorithm}[!htp]
\caption{Computing the preliminary BWT of $T^{i}$}
\label{algo:prebwt}{}
\begin{algorithmic}[1]
\small 
\Require $\gsame, D^{i}=(R^{i}, L^{i}, N^{i}), B^{i}$

    \State $\ell, b, f \gets 0$
    \State $s_x,e_x,o \gets 1$
    \State $c', c \gets \varepsilon$
    \State $pBWT^{i}, M, S\!A_{s} \gets \emptyset$ 

    \For{$u=1$ to $|\gsame|$}\label{algo:inf:spar} 
        \If{$u>1$ and $\gsame[u]$ is the start of a new block in $\gsame$}\label{code:newrange}
            \State $j \gets \gsame[s_x]$, $e_x \gets u-1$ \Comment{process block $\gsame[s_x..e_x]$ for $S_x$}
            \If{$L^{i}[j+1] \neq 1$ \textbf{or} $B^{i}[R^{i}[j]]=1$} \Comment{$S_x$ belongs to $\mathcal{S}^{i}$}
                \If{$b>1$ \textbf{or} $f=1$} \Comment{$S_x$ produces an unsolved block in $pBWT^{i}$}
                    \If{$b>1$}\label{code:lmif} \Comment{$S_x$ is a left-maximal suffix in $\mathcal{F}^{i}$}
                        \State $c \gets \texttt{\#}$\label{code:pbwtdumm2}
                        \State $j' \gets$\ leftmost index $j' \geq j$ with $L[j'+1]=1 $ or $j'=|R^{i}|$  
                        \State $S_x \gets R^{i}[j..j']$
                        \State $M \gets M \cup (S_x, \sigma^{i} + o)$
                        \State mark occurrences of $S_x$ in $R^{i}$ as left-maximal
                    \Else \Comment{$S_x$ is a non-proper suffix in $\mathcal{F}^{i}$}
                        \State $c \gets \texttt{*}$%\label{code:pbwtdumm1}
                    \EndIf
                    \State $S\!A_{s} \gets S\!A_{s} \cup \{j\}$ \Comment{sample one occurrence of $S_x$}
                    \State $o \gets o+1$
                \EndIf

                \If{$c$ equals the symbol in the rightmost run of $pBWT^{i}$}
                    \State increase the length of that run by $\ell$
                \Else
                    \State $pBWT^{i} \gets pBWT^{i} \cup \{(c, \ell)\}$ \Comment{append new run}
                \EndIf
            \EndIf
            \State $\ell, b,f \gets 0, s_x \gets u, c' \gets \varepsilon$ \Comment{initialize new block's information}
        \EndIf\label{code:newrange2}

        \If{$L[\gsame[u]]=1$}\label{code:comp1} \Comment{$S_x$ is a phrase in $\mathcal{F}^{i}$}
            \State $c \gets \texttt{@}, f\gets f+1$%\label{code:pbwtdumm2}
        \Else
            \State $c \gets R^{i}[\gsame[u]-1]$
        \EndIf
        \State $b \gets b + c \neq c'$, $\ell \gets \ell + N^{i}[\mathsf{rank}(L^{i}, j)]$
        \State $c' \gets c$\label{code:comp2}
    \EndFor
    \State process last block of $\gsame$ as in lines \ref{code:newrange}-\ref{code:newrange2}
    \State store $pBWT^{i}$ on disk
    \State \textbf{return} $(S\!A_{s}, M)$
\end{algorithmic}
\end{algorithm}

\paragraph{Computing $pBWT^{i}$ in Compressed Space} Algorithm~\ref{algo:prebwt} produces $pBWT^{i}$ from \gsa{}, $D^{i}=(R^{i}, L^{i}, N^{i})$, and $B^{i}$. We consider the partition of \gsa{} where every block $\gsame[s_x..e_x]$ encodes the suffixes that come from different phrases of $\mathcal{F}^{i}$ but spell the same sequence $S_x \in \Sigma^{i}$\footnote{In practice, we compute every distinct block $S\!A_{\mathcal{F}^{i}}[s_x..e_x]$ during the construction of \gsa. We reserve the least significant bit in the cells of \gsa to mark every $S\!A_{\mathcal{F}^{i}}[s_x]$. We flag these positions during the execution of our modified version of ISS (Paragraph~\ref{par:iss_dict}).}. For $\gsame[s_x..e_x]$, we compute the number $b$ of left-context breaks, the number $f$ of non-proper suffixes, and the accumulative frequency $\ell$ of the phrases of $\mathcal{F}^{i}$ enclosing the occurrences of $S_x$ in $\gsame[s_x..e_x]$. A left-context break between consecutive positions $\gsame[u-1]$ and $\gsame[u]$, with $u \in [s_x+1,e_x]$, occurs when $R^{i}[\gsame[u-1]-1]$ differs from $R^{i}[\gsame[u]-1]$. For technical convenience, we compare the left context of $\gsame[s_x]$ against the empty symbol $\varepsilon \neq \Sigma^{i}$, so every block in \gsa{} has at least one left-context break. On the other hand, a suffix $\gsame[u]$, with $u \in [s_x,e_x]$, is non-proper when $R^{i}[\gsame[u]]$ is the leftmost symbol of a phrase. In this case, we assign a special symbol $\texttt{@} \notin \Sigma^{i}$ as its left context. Lines~\ref{code:comp1}--\ref{code:comp2} describe the process of computing the values $(b, f, \ell)$ for $\gsame[s_x..e_x]$ on the fly as we scan $\gsame$.~Once we reach the start of a new block in \gsa, we process the information of the previous $\gsame[s_x..e_x]$, but only if its associated sequence $S_x$ belongs to $\mathcal{S}^{i}$. This condition holds when $|S_x|>1$, or $|S_x|=1$ and $exp(S_x[1])$ is suffixed by \texttt{\$}. If $b>1$ or $f=1$, $S_x$ induces an unsolved block in $pBWT^{i}$. Consequently, we append $\ell$ copies of a special symbol to $pBWT^{i}$. The value of this symbol depends on $f$ and $b$. If $b=1$ and $f=1$, $S_x$ is always a non-proper suffix in $\mathcal{F}^{i}$ (Lemma~\ref{lem:lem4}), so we append $\texttt{*}^{\ell}$. On the other hand, if $S_x$ is left-maximal ($b>1$), we append $\texttt{\#}^{\ell}$ instead (Lemma~\ref{lem:lem3}). When the symbol is $\texttt{\#}$, we also store $S_x$ associated with an integer in a dictionary $M$. If $\gsame[s_x..e_x]$ is the $oth$ block producing a special symbol in $pBWT^{i}$, the integer for $S_x$ is $\sigma^{i} + o$. After deciding the special symbol we insert in $pBWT^{i}$, we append $\gsame[s_x]$ into another array $S\!A_s$ representing a sampled version of \gsa. On the other hand, when $f=0$ and $b=1$, $S_x$ always occurs as a proper suffix in $\mathcal{F}^{i}$, and it is always preceded by $c \in \Sigma^{i}$ (Lemma~\ref{lem:rl}). Thus, we append $c^{\ell}$ to $pBWT^{i}$. Lines~\ref{code:newrange}--\ref{code:newrange2} describe the processing of $\gsame[s_x..e_x]$. Once we finish the scan of \gsa{}, we store $pBWT^{i}$ into the disk and return $S\!A_{s}$ and $M$. Figure~\ref{fig:prBWT_comp} shows an example of the computation of $pBWT^{i}$ in compressed space.

\begin{example}
Construction of $pBWT^{i}$ in Figure~\ref{fig:prBWT_comp}. The first run in $pBWT^{i}$ is $\texttt{c}^{2}$ because the suffix $R^{i}[\gsame[1]=7] = \texttt{\$}$ of the first block belongs to $R^{1}[4..7]=\texttt{acc\$}$, that has frequency $N^{1}[\mathsf{rank}_1(L^{1}, 7)=2]=2$. In contrast, the second run is $\texttt{*}^{4}$ because the suffixes $8$, $4$ and $12$ of the next three \gsa{} blocks map to the full phrases $R^{i}[8..11]=\texttt{aata}$, $R^{i}[4..7]=\texttt{acc\$}$, and $R^{i}[12..15]=\texttt{agta}$, respectively. The run has length $4$, and not $3$, because $\texttt{acc\$}$ has frequency $N^{1}[\mathsf{rank}_{1}(L^{1}, 4)=2]=2$. We add $(\texttt{gta}, 9)$ into $M$ because $\texttt{gta}$ is left-maximal in $\mathcal{F}^{i}$, and its associated block $\gsame[9..10]$ is the $4th$ block producing special symbols in $pBWT^{i}$. Now, assuming $\sigma^{1}=5$, we have that the integer value we store in $M$ is $4+5=9$. A similar situation occurs with $(\texttt{ta}, 10)$.
\end{example}

\paragraph{Output Data Structures in the Construction of $pBWT^{i}$} In addition to the preliminary BCR BWT of $T^{i}$, Algorithm~\ref{algo:prebwt} produces two auxiliary data structures: $S\!A_{s}$ and $M$. We will use these extra elements to assist in the compression of $\mathcal{F}^{i}$ in the next iteration step. $S\!A_{s}$ is a sampled version of \gsa{} storing one occurrence in $R^{i}$ for each string $S_x \in \mathcal{S}^{i}$ whose corresponding block $pBWT^{i}[s_x..e_x]$ is unsolved. Thus, if there are $n'\geq |\mathcal{F}^{i}|$ blocks in \gsa{} that have an unsolved projected block in $pBWT^{i}$, then $S\!A_{s}[1..n']$ has $n'$ sampled positions. Notice that because $S\!A_{s}$ maintains the relative order of the elements in \gsa{}, the strings encoded by $S\!A_{s}$ are sorted in LMS order. On the other hand, the dictionary $M$ stores as a key each string $S_x \in \mathcal{S}^{i}$ that is left-maximal in $\mathcal{F}^{i}$. The associated value of each key is the rank of the unsolved block that $S_x$ produced in $pBWT^{i}$. For instance, if the area $\gsame[s_x..e_x]$ for $S_x$ is the $oth$ block producing an unsolved segment in $pBWT^{i}$, the integer in $M$ for $S_x$ is $o+\sigma^{i}$. We add $\sigma^{i}$ just for technical convenience.  

\begin{figure}[!t]
\centering
\resizebox{0.7\textwidth}{!}{%
\begin{tikzpicture}[>=stealth,thick,baseline]

\matrix (m4) [matrix of nodes, ampersand replacement=\|,
     row 1/.style={nodes={gray, font=\footnotesize\bfseries}}] at (0,0)  { 
      \| \|  $1$ \| $2$ \| $3$ \| $4$ \| $5$ \| $6$ \| $7$ \| $8$ \| $9$ \| $10$ \| $11$ \| $12$ \| $13$ \| $14$ \| $15$ \\
     $R^{1}$ \| $=$ \|  \texttt{g} \| \texttt{t} \| \texttt{a} \| \texttt{a} \| \texttt{c} \| \texttt{c} \| \texttt{\$} \| \texttt{a} \| \texttt{a} \| \texttt{t} \| \texttt{a} \| \texttt{a} \| \texttt{g} \| \texttt{t} \| \texttt{a} \\
     $L^{1}$ \| $=$ \|  1 \| 0 \| 0 \| 1 \| 0 \| 0 \| 0 \| 1 \| 0 \| 0 \| 0 \| 1 \| 0 \| 0 \| 0 \\
     $N^{1}$ \| $=$ \|  2 \| \| \| 2 \| \| \| \| 1 \| \| \| \| 1 \\
};

\matrix (m3) [matrix of nodes, ampersand replacement=\|, below= 1.5cm of m4,
              row 2/.style={nodes=draw, white},
              column 1/.style={nodes={anchor=base east}}]  { 
     $\gsame$ \| $=$ \|  7 \| 8 \| 4 \| 12 \| 9\| 3 \| 11 \| 15 \| 6 \| 5 \| 1 \| 13 \| 2 \| 10 \| 14  \\
                   \|  a  \|    \|    \|   \|   \|   \|   \|   \|  \|    \|   \|    \|   \|   \|    \|   \|   \|    \|    \\
     $S\!A_{s}$    \| $=$ \|   \| $8_{\textcolor{gray!70}{1}}$ \| $4_{\textcolor{gray!70}{2}}$ \| $12_{\textcolor{gray!70}{3}}$ \| \|  \|  \|  \|  \|  \| $1_{\textcolor{gray!70}{4}}$ \|  \| $2_{\textcolor{gray!70}{5}}$ \|  \|  \\
     $pBWT^{i}$    \| $=$ \| $\texttt{c}^{2}$  \| $\texttt{*}^{4}$ \|  \| \|  $\texttt{a}^{1}$ \| \|  \|  \| $\texttt{c}^{2}$ \| $\texttt{a}^{2}$ \| $\texttt{\#}^{7}$ \|  \| \|  \|  \\
                     \||[white]|  a  \|    \|    \|   \|   \|   \|   \|   \|  \|    \|   \|    \|   \|   \|    \|   \|   \|    \|    \\
     $M$    \| $=$ \|   \|  \|  \| \|   \| \|  \|  \|  \|  \| |[white]| a \|  \||[white]|a \|  \|  \\
};
\node[] at (m3-6-13) {$(\texttt{gta},9)$};
\node[] at (m3-6-15) {$(\texttt{ta},10)$};

\node[anchor = west, rotate = 90, gray] (s1) at (m3-1-3.north) {\texttt{\$}};
\node[anchor = west, rotate = 90, gray] (s2) at (m3-1-4.north) {\texttt{aata}};
\node[anchor = west, rotate = 90, gray] (s3) at (m3-1-5.north) {\texttt{acc\$}};
\node[anchor = west, rotate = 90, gray] (s4) at (m3-1-6.north) {\texttt{agta}};
\node[anchor = west, rotate = 90, gray] (s5) at (m3-1-7.north) {\texttt{ata}};
\node[anchor = west, rotate = 90, gray] (s6) at (m3-1-8.north) {\texttt{a}};
\node[anchor = west, rotate = 90, gray]      at (m3-1-9.north) {\texttt{a}};
\node[anchor = west, rotate = 90, gray] (s7) at (m3-1-10.north) {\texttt{a}};
\node[anchor = west, rotate = 90, gray] at (m3-1-11.north) {\texttt{c\$}};
\node[anchor = west, rotate = 90, gray] at (m3-1-12.north) {\texttt{cc\$}};
\node[anchor = west, rotate = 90, gray] at (m3-1-13.north) {\texttt{gta}};
\node[anchor = west, rotate = 90, gray] at (m3-1-14.north) {\texttt{gta}};
\node[anchor = west, rotate = 90, gray] at (m3-1-15.north) {\texttt{ta}};
\node[anchor = west, rotate = 90, gray] at (m3-1-16.north) {\texttt{ta}};
\node[anchor = west, rotate = 90, gray] at (m3-1-17.north) {\texttt{ta}};

\draw (m3-1-3.north west) rectangle (m3-1-3.south east);
\draw (m3-1-3.north east) rectangle (m3-1-4.south east);
\draw (m3-1-4.north east) rectangle (m3-1-5.south east);
\draw (m3-1-5.north east) rectangle (m3-1-6.south east);
\draw (m3-1-6.north east) rectangle (m3-1-7.south east);
\draw (m3-1-7.north east) rectangle (m3-1-10.south east);
\draw (m3-1-10.north east) rectangle (m3-1-11.south east);
\draw (m3-1-11.north east) rectangle (m3-1-12.south east);
\draw (m3-1-12.north east) rectangle (m3-1-14.south east);
\draw (m3-1-14.north east) rectangle (m3-1-17.south east);
\draw[pattern=north west lines, pattern color=gray!60] (m3-1-7.north east)  rectangle (m3-1-10.south east);
\end{tikzpicture}
}
\caption{Construction of $pBWT^{1}$ in compressed space using the succinct dictionary $D^{1}=(R^{1}, L^{1}, N^{1})$ of the LMS parsing of Figure~\ref{fig:lms_par}. The boxes in \gsa{} are the blocks of equal suffixes in $\mathcal{F}^{i}$. We skip block $6$ (dashed box) because its suffixes overlap other phrases in $\mathcal{F}^{i}$, and hence, are redundant to build $pBWT^{i}$. }
\label{fig:prBWT_comp}
\end{figure}

\subsubsection{Grammar Compression}\label{ssec:gc}

Once we finish constructing $pBWT^{i}$, the next step in the parsing iteration $i$ is to store $\mathcal{F}^{i}$ in a compact form to use it later during the induction phase (Line~\ref{code:gram_comp_set} of Algorithm~\ref{algo:ovgrlbwt}). We first explain why we need the parsing set during the induction phase and then describe the format we choose to encode it. 

Broadly speaking, the induction process consists of scanning $BWT^{i+1}_{bcr}$ left to right, mapping every symbol $BWT^{i+1}_{bcr}[j] \in \Sigma^{i+1}$ back to the phrase $F[1..n_f] \in \mathcal{F}^{i}$ from which it originated, and then checking which of the proper suffixes of $F$ produced unsolved blocks in $pBWT^{i}$ (see Lemmas~\ref{lem:lem3} and~\ref{lem:lem4}). Assume the suffix $F[u..n_f] = S_x \in \mathcal{S}^{i}$ produced the unsolved block $pBWT^{i}[s_x..e_x]$, then we insert $F[u-1]$ into $pBWT^{i}[s_x..e_x]$. 

The process described above requires a mechanism to map left-maximal suffixes in $\mathcal{F}^{i}$ back to the unsolved block they produce in $pBWT^{i}$. We implement this feature by encoding $\mathcal{F}^{i}$ with a representation that is similar to grammar compression (Section~\ref{ssec:comp}). 

We start by modifying the set to make it suitable for our encoding. First, we insert each string $S_x \in \mathcal{S}^{i} \setminus \mathcal{F}^{i}$ associated with the unsolved block $pBWT^{i}[s_x..e_x]$ into $\mathcal{F}^{i}$. These strings are those meeting Lemma~\ref{lem:lem3} and that only appear as a proper suffix in $\mathcal{F}^{i}$, not as a full phrase. Then, we sort the strings in $\mathcal{F}^{i}$ in $\prec_{LMS}$ order. We refer to this new version of $\mathcal{F}^{i}$ as the \emph{expanded} parsing set \exppset{}.

\begin{definition}
   Expanded parsing set: a string set $\exppsetme \subset \mathcal{S}^{i}$ storing each element $S_x \in \mathcal{S}^{i}$ whose associated block $pBWT^{i}[s_x..e_x]$ either meets Lemma~\ref{lem:lem3} or Lemma~\ref{lem:lem4}. Additionally, the strings in \exppset{} are sorted in LMS order.  
\end{definition}

\begin{figure}[!t]
\centering
\resizebox{0.6\textwidth}{!}{%
\begin{tikzpicture}[>=stealth,thick,baseline]

\matrix (m4) [matrix of nodes, ampersand replacement=\|,
     ] at (0,0)  { 
     $\exppsetme$ \| $=$ \|  $\{\texttt{a\underline{a}\textcolor{gray!70}{ta}}$, \| \texttt{acc\$}, \| \texttt{\underline{a}\textcolor{gray!70}{gta}}, \| \texttt{\underline{g}}\textcolor{gray!70}{\texttt{ta}}, \| $\textcolor{gray!70}{\texttt{ta}}\}$ \\
     $G^{1}$ \| $=$ \|  $\texttt{a}10$ \| $\texttt{@\$}$ \| $\texttt{a}9$ \| $\texttt{g}10$ \|  $\texttt{@t}$  \\
};
\end{tikzpicture}
}
\caption{Example of our grammar-like encoding for the set \exppset{} constructed from the parsing set $\mathcal{F}^{1}$ of Figure~\ref{fig:prBWT_examp}. Each suffix in grey is the longest left-maximal suffix of its phrase. The underlined symbol is the left context of that suffix. }
\label{fig:comp_exppset}
\end{figure}

Notice the convenience of the expanded set: if we need to access the string $S_x$ associated with the $oth$ unsolved block $pBWT^{i}[s_x..e_x]$, we visit the $oth$ string in \exppset. Now, to fill $pBWT^{i}[s_x..e_x]$, we also need to know the symbols in $\Sigma^{i}$ preceding $S_x$. We support this functionality by compressing \exppset. Specifically, we replace the suffix occurrences of $S_x$ in \exppset{} with its LMS order $o$ in \exppset{}. Thus, for instance, if we access a string $F$ in the compressed version of \exppset{} that is suffixed by $c{\cdot}o$, we know that $c$ is the left context of $S_x$, and that goes within the $oth$ unsolved block $pBWT^{i}[s_x..e_x]$. This grammar-like encoding is lossless because there is a phrase within \exppset{} for each symbol $o \notin \Sigma^{i}$ we insert. Therefore, if we need to access the nested left-maximal suffixes of $S_x$, we just visit the $oth$ string in \exppset{}. We can further compress \exppset{} by removing the suffixes that do not produce unsolved blocks in $pBWT^{i}$. The result is a lossy grammar-like encoding that gives us access to the left-maximal suffixes of \exppset{} and their left-context symbols. This feature is enough for us to run the induction phase. 

In practice, the sampled suffix array $S\!A_{s}$ that we produced with Algorithm~\ref{algo:prebwt} is an implicit representation of \exppset{}, so we do not have to compute it. On the other hand, the dictionary $M$ contains the information we need to compress the suffixes of \exppset{} that generate unsolved blocks in $pBWT^{i}$. 

\paragraph{Grammar-compressing the Expanded Parsing Set} Algorithm~\ref{algo:gramc} shows in detail the steps we perform during this procedure. We start by initializing a new vector $G^{i}[1..2|S\!A_{s}|]$ to store the compressed version of \exppset.~Then, we scan $S\!A_{s}$ from left to right, and for every position $j=S\!A_{s}[u]$ we visit, we access its associated phrase $F[1..n_f]=R^{i}[j..j'] \in \mathcal{F}^{i}_{exp}$, with $j'\geq j$ being the leftmost index with $L^{i}[j'+1]=1$. First, if $j=j'$ and $L^{i}[j']=1$, then $F=R^{i}[j]$ has length one and expands to a full string  $exp(F[1])=T_{x}\texttt{\$}$, with $T_x \in \mathcal{T}$. This situation is a corner case generated by the LMS parsing, so we set $F=\texttt{@}R^{i}[j]$ (see Line~\ref{code:fullphrase}), where $\texttt{@}=|S\!A_{s}|+1$ is a special symbol that denotes an invalid element in the encoding. On the other hand, when $n_f>1$, we scan the phrase from left to right to find its longest proper suffix $S_x=F[u..n_f]=R^{i}[p..j']$, with $j>p\leq j'$, which exists in $M$ as a key. If such a key exists, we obtain its left context $c=F[u-1]=R^{i}[p-1]$ and replace $F$ with $c{\cdot}o'$, where $o'$ is the integer associated with $S_x$ in $M$ (Lines~\ref{code:ls1}--\ref{code:ls2}). If no proper suffix of $F$ exists in $M$ as a key, we replace $F$ with a sequence of length two. The left symbol will be $\texttt{@}$, while the right symbol depends on $F$' sequence. If $F[n_f] = R^{i}[j']$ expands to a string in $\Sigma^{*}$ suffixed by $\texttt{\$}$, we set the right symbol to $F[n_f]$, or set the right symbol to $F[n_f-1]=R^{i}[j'-1]$ otherwise~(Lines\ref{code:nls1}--\ref{code:nls2}). Once we update $F$, we append it to $G^{i}$. After we finish the scan of $S\!A_{s}$, we destroy $R^{i}$, $L^{i}$, $S\!A_s$, and $M$, and finally store $G^{i}$ on disk. Figure~\ref{algo:gramc} depicts an example of our grammar-like encoding.

\paragraph{Decompressing Strings of the Expanded Parsing Set} Consider again the phrase $F$, which we replaced with the string $c{\cdot}o'$ in $G^{i}$. The symbol $o' > \sigma^{i}$ is the integer we obtained when we performed a lookup operation of $S_x = F[u..n_f]$ in $M$ during the execution of Algorithm~\ref{algo:gramc}, while $c=F[u-1]$. Further, the value $o' = o + \sigma^{i}$ is the LMS order $o$ of $S_x$ in \exppset{} plus $\sigma^{i}$ (see Algorithm~\ref{algo:prebwt}). The existence of $S_x$ in $M$'s keys implies that $S_x$ is a left-maximal suffix in $\mathcal{F}^{i}$, which in turn implies that $S_x$ is a full phrase in \exppset. Additionally, the left-maximal condition of $S_x$ implies that at least two suffix occurrences of $S_x$ in $\mathcal{F}^{i}$ were preceded by different symbols. This is why $S_x$ produces the $oth$ unsolved block $pBWT^{i}[s_x..e_x]$. Now, recall that the phrases of \exppset{} are stored in $\prec_{LMS}$ order. Therefore, if we want to access the area where $S_x$ lies, we have to set $o=o'-\sigma^{i}$ and go to $G^{i}[2o-1..2o]$. This substring does not encode the full sequence of $S_x$, but its longest left-maximal suffix $G^{i}[2o]$ (which is also a left-maximal suffix of $F$) along with the left-context symbol for that suffix ($G^{i}[2o-1] \in \Sigma^{i}$). Recursively, the longest left-maximal symbol of $S_x$ is not a sequence either but a pointer to another position of $G^{i}$. We access this nested left-maximal suffix by setting $o=G^{i}[2o] - \sigma^{i}$ and updating the area $G^{i}[2o-1..2o]$. We continue applying this idea until we reach a range $G^{i}[2o-1..2o]$ where $c=G^{i}[2o] \leq \sigma^{i}$, which implies that we reached the last suffix of $F$. In most of the cases, $c$ is not $F[n_f]$, but $F[n_f-1]$. The reason is that the LMS substrings overlap by one symbol in $T^{i}$, so $F[n_f]$ is redundant as it also appears as a prefix in another phrase. The only exception to this rule is when $F$ expands to a suffix of a string in $\mathcal{T}$. In that case, $G^{i}[2o]$ is indeed $F[n_f]$ because $F$ does not overlap the prefixes of other LMS parsing phrases. 

\begin{example}
Spelling the left-context symbols of the left-maximal suffix \texttt{agta} (Figure~\ref{fig:comp_exppset}). The string \texttt{agta} has LMS order $3$ in \exppset{} so we visit the $3th$ phrase in $G^{1}[2{\times}3-1..2{\times}3]=G^{1}[5..6]=\texttt{a}9$, and access the left symbol \texttt{a}. Then, we compute the next string using the right symbol $9$ as $4=9-\sigma^{i}=9-5$, visit the $4th$ string $G^{i}[2{\times}4-1..2{\times}4]=G^{1}[7..8]=\texttt{g}10\ (\texttt{gta})$ and output $\texttt{g}$. We repeat the same process, computing the string $5=10-5$ and visiting $G^{1}[9..10]=\texttt{@t}\ (\texttt{t\textcolor{gray!70}{a}})$. However, this time we have $\texttt{t} \leq \sigma^{1}$, which means we reached the rightmost suffix of $\texttt{agta}$ that does match the prefix of other phrases in the LMS parsing of $T^{1}$ (see the phrases ending with \texttt{ta} in the string $T^{i}$ of Figure~\ref{fig:lms_par}). This situation does not occur, for instance, with \texttt{acc\$}, whose rightmost suffix in $G^{1}$ is indeed \texttt{\$} because $\texttt{acc\$}$ can not overlap.
\end{example}

\begin{algorithm}[!t]
\caption{Grammar compressing \exppset{}}

\label{algo:gramc}{}
\begin{algorithmic}[1]
\Require $R^{i}, L^{i}, B^{i}, S\!A_{s}, M$

    \State $G^{i} \gets \emptyset$ \Comment{grammar-compressed \exppset}

    \For{$u=1$ to $|S\!A_{s}|$}  \Comment{visit each $F=R^{i}[j..j'] \in \exppsetme$}

        \State $j \gets S\!A_{s}[u]$ 
        \State $j' \gets$\ leftmost index $j'\geq j$ with $L[j'+1]=1$ or $j'=|R^{i}|$ 
        \State $o' \gets \varepsilon, c \gets \texttt{@}$

        \If {$L^{i}[j']=1$ and $B^{i}[R^{i}[j']]=1$}\label{code:fullphrase} \Comment{$exp(R^{i}[1]=F[1])=T_x\texttt{\$}$}
            \State $o' \gets R^{i}[j]$
        \Else

            \State $p \gets j+1$
            \While{ $p\leq j'$ \textbf{and} $R^{i}[p..j']$ is not left-maximal}\label{code:ls1}
                \State $p \gets p+1$
            \EndWhile

            \If{$S_x = R^{i}[p..j']$ is left-maximal}
                \State $c \gets R^{i}[p-1]$
                \State $o' \gets$\ value associated to $S_x$ in $M$\label{code:ls2}
            \Else\label{code:nls1} \Comment{$F$ does not have left-maximal suffixes}

                \If{$B^{i}[R^{i}[j']]=1$} \Comment{$F$ expands to a suffix in $\mathcal{T}$}
                    \State $o' \gets R^{i}[j']$ 
                 \Else \Comment{$F$ is a regular parsing phrase with overlap}
                    \State $o' \gets R^{i}[j'-1]$ 
                \EndIf
            \EndIf\label{code:nls2}
        \EndIf
        \State $F \gets c{\cdot}o'$
        \State $G^{i} \gets G^{i} \cup F$

    \EndFor\label{algo:inf:epar}
    \State destroy $R^{i}, L^{i}, B^{i}, S\!A_{s}$ and $M$
    \State store $G^{i}$ on disk 
\end{algorithmic}
\end{algorithm}

\subsection{Creating the String for the Next Iteration of Parsing}

The final step of iteration $i$ during the parsing phase is to create the text $T^{i+1}$. For this purpose, we produce a new dictionary $D^{i}$ containing the strings in $\mathcal{F}^{i}$ (i.e., the parsing phrases) as keys. The value associated with each key is its LMS order in \exppset.  We construct $T^{i+1}$ by running LMS parsing over $T^{i}$ again to replace the phrases with their associated values in $D^{i}$. If $T^{i+1}$ has length $k$ (the number of strings in $\mathcal{T}$), we stop the parsing phase as all the strings in $\mathcal{T}$ are now compressed to one symbol. 

The only caveat with this construction is that the symbols in the alphabet $\Sigma^{i+1}$ of $T^{i+1}$ are not consecutive if $|\exppsetme|>|\mathcal{F}^{i}|$, but this feature does not change the correctness of our method. Specifically, if $T^{i+1}[j]<T^{i+1}[j']$, it still hold that $exp(T^{i+1}[j]) \prec_{LMS} exp(T^{i+1}[j'])$. We will use the fact that the symbols in $T^{i+1}$ are the LMS order in \exppset{} and not in $\mathcal{F}^{i}$ during the induction phase.   

\subsection{The Induction Phase}\label{ssec:indpha}

The induction phase starts with the computation of $BWT^{h}_{bcr}$, the BCR BWT for the text $T^{h}$ of the last iteration $h$ of the parsing phase (Line~\ref{code:ind_phase_start} of Algorithm~\ref{algo:ovgrlbwt}). This step is trivial as each symbol in $T^{h}$ encodes a full string of $\mathcal{T}$ (see the ending condition of the parsing phase). Hence, the left context of every symbol is the symbol itself. The BCR BWT maintains the relative order of the strings in $\mathcal{T}$ (see Section~\ref{ssec:bwt}), so $BWT^{h}_{bcr}$ is $T^{h}$ itself.

Before explaining our induction procedure, we describe some important properties of $BWT^{i+1}_{bcr}$. As a quick reminder, \exppset{} is the expanded parsing set encoding the strings in $\mathcal{F}^{i}$ (see Lemma~\ref{lem:lem2}) plus the sequences that are left-maximal suffixes in $\mathcal{F}^{i}$. The strings in \exppset{} are precisely those inducing unsolved blocks in $pBWT^{i}$. 

\begin{lemma}\label{lem:ind1}
Let $BWT^{i+1}_{bcr}[j]$ and $BWT^{i+1}_{bcr}[j']$ be two symbols in $\Sigma^{i+1}$, with $j<j'$, whose corresponding phrases in $\mathcal{F}^{i}$ are $F[1..n_f]$ and $F'[1..n_{f'}]$, respectively. Additionally, let the proper suffixes $F[u..n_f]=F'[u'..n_{f'}]=S_x \in \exppsetme$ be left-maximal in $\mathcal{F}^{i}$. Now consider the substrings $map^{i}(T^{i+1}[GS\!A^{i+1}[j]-1]) = T^{i}[p..p+n_f-1]$ and $map^{i}(T^{i+1}[GS\!A^{i+1}[j']-1])=T^{i}[p'..p'+n_{f'}-1]$ with the occurrences of $F$ and $F'$ that formed the symbols $BWT^{i+1}_{bcr}[j]$ and $BWT^{i+1}_{bcr}[j']$ in $T^{i+1}$, respectively. The suffix $T^{i}[p+u-1..n^{i}]$ prefixed by $S_x = F[u..n_f]$ precedes in \sa{} the suffix $T^{i}[p'+u'-1..n^{i}]$ prefixed by $S_x = F[u'..n_{f'}]$.
\end{lemma}

\begin{proof}
As the prefixes $T^{i}[p+u-1..n_f]=F[u..n_f]=S_x$ and $T^{i}[p'+u'-1..n_{f'}-1]=F'[u'..n_{f'}]=S_x$ are equal, the relative order of $T^{i}[p+u-1..n^{i}]$ and $T^{i}[p'+u'-1..n^{i}]$ is decided by the right contexts in $T^{i+1}$ of the occurrences $BWT^{i+1}_{bcr}[j]$ and $BWT^{i+1}_{bcr}[j']$. By induction, we know that $BWT^{i+1}_{bcr}$ is complete, and as $BWT^{i+1}_{bcr}[j]$ precedes $BWT^{i+1}_{bcr}[j']$ in the BWT, the right context of $T^{i}[p+u-1..n_f]$ is lexicographically smaller than the right context of $T^{i}[p+u'-1..n_{f'}]$.
\end{proof}

We generalize Lemma~\ref{lem:ind1} to compute the block $pBWT^{i}[s_x..e_x]=\texttt{\#}^{\ell}$ that meets Lemma~\ref{lem:lem3} and whose string $S_x \in \mathcal{S}^{i} \setminus \mathcal{F}^{i}$ \emph{only} occurs as a proper suffix in $\mathcal{F}^{i}$.

\begin{lemma}\label{lem:ind2}
Let $J=\{j_1, j_2,\ldots,j_\ell\}$ be a set of strictly increasing positions of $BWT^{i+1}_{bcr}$. Every $BWT^{i+1}_{bcr}[j_{b}]$, with $j_{b} \in J$, is a symbol $o \in \Sigma^{i+1}$ assigned to a phrase $F[1..n_f] \in \mathcal{F}^{i}$ where $S_x=F[u..n_{f}] \in \exppsetme$ occurs as a proper suffix. The symbols of $BWT^{i+1}_{bcr}$ referenced by $J$ are not necessarily equal, and hence, their associated phrases in $\mathcal{F}^{i}$ are not necessarily the same. However, these phrases of $\mathcal{F}^{i}$ are all suffixed by $S_x$. Assume we scan $J$ from left to right, and for every $j_{b}$, we extract the symbol $F[u-1] \in \Sigma^{i}$ that precedes $S_x$ and append it to a vector $L_{S_x}$. The resulting sequence for $L_{S_{x}} \in \Sigma^{i*}$ matches the unsolved block $pBWT^{i}_{bcr}[s_x..e_x]=\texttt{\#}^{\ell}$ generated by $S_x$.
\end{lemma}

\begin{proof}
Lemma~\ref{lem:ind1} tells us that the suffix of $T^{i}$ prefixed by the occurrence of $S_x$ encoded by $BWT^{i+1}[j_{b}]$ precedes the suffix of $T^{i}$ prefixed by the occurrence encoded by $BWT^{i+1}[j_{b+1}]$. This property holds for every $j_{b}$, with $b \in [1,\ell-1]$. Hence, the suffixes of $T^{i}$ prefixed by $S_x$ are already sorted in $J$. On the other hand, Lemma~\ref{lem:lem2} tells us that all the occurrences of $S_x$ as a suffix of a parsing phrase appear consecutively in $\same[s_x..e_x]$. Thus, by taking the left-context symbols of $S_x$'s occurrences encoded by $J$, we obtain $pBWT^{i}_{bcr}[s_x..e_x]$.
\end{proof}

\begin{example}
Filling the unsolved block $pBWT^{i}[15..18]=\texttt{\#}^{4}$ of Figure~\ref{fig:prBWT_examp} with Lemma~\ref{lem:ind2}.~Consider the BCR BWT $BWT^{2}_{bcr}=4\ 4\ 3\ 1\ 2\ 2$ for the text $T^{2}=4\ 2\ 4\ 1\ 3\ 2$ of Figure~\ref{fig:lms_par}.~The expanded parsing set \exppset{} from which the symbols of $BWT^{2}_{bcr}$ were generated is shown in Figure~\ref{fig:comp_exppset}. Additionally, consider the string $\texttt{ta} \in \mathcal{S}^{1}$ generating the block $GS\!A^{1}[15..18]$ in the partition induced by $\mathcal{S}^{1}$ (see Figure~\ref{fig:prBWT_examp}). As $S\!A^{1}[15..18]$ meets Lemma~\ref{lem:lem3}, the projected block $pBWT^{i}[15..18]=\texttt{\#}^{4}$ is unsolved. Lemma~\ref{lem:ind2} tells us that we can solve $pBWT^{1}[15..18]$ provided we know $BWT^{i+1}_{bcr}$ and the phrases of $\mathcal{F}^{1}$ where $\texttt{ta}$ occurs as a suffix. The prefix $BWT^{2}_{bcr}[1..4]= 4\ 4\ 3\ 1$ contains all the symbols in $\Sigma^{i+1}$ whose associated phrases are suffixed by \texttt{ta}. In particular, if we replace $4\ 4\ 3\ 1$ with their phrases in \exppset{}, we obtain \texttt{gta}, \texttt{gta}, \texttt{agta}, \texttt{aata}. We apply Lemma~\ref{lem:ind2} by taking the left context of \texttt{ta} in those strings without changing the relative order and thus obtain $L_{\texttt{ta}} = pBWT^{1}[15..18]=\texttt{g g g a}$, which is precisely the substring $BWT_{bcr}[15..18]$ of Figure~\ref{fig:prBWT_examp}.
\end{example}

When $S_x$ does not occur as a nested proper suffix in $\mathcal{F}^{i}$ (Lemma~\ref{lem:lem4}), there is no left-context symbol we can extract from the parsing set, so Lemma~\ref{lem:ind2} does not work. Nevertheless, we can use $BWT^{i}_{bcr}$ in other ways to complete the unsolved block $pBWT^{i}[s_x..e_x]=\texttt{*}^\ell$ that $S_x$ generated.

\begin{lemma}\label{lem:ind3}
Let $pBWT^{i}[s_x..e_x]=\texttt{*}^\ell$ be an unsolved block induced by a string $S_x \in \mathcal{S}^{i}$ that meets Lemma~\ref{lem:lem4}. Additionally, let $o'$ and $o \in \Sigma^{i+1}$ be the LMS orders of $S_x$ in $\mathcal{F}^{i}$ and \exppset, respectively. Let $\same[u..u']$ be the bucket $o' \leq o$ storing the suffixes of $T^{i}$ prefixed by $o$. The element in $pBWT[s_x+j-1]$ is the rightmost symbol in $exp^{i}(BWT^{i+1}_{bcr}[u+j-1])$, with $u+j-1 \leq u'$.
\end{lemma}

\begin{proof}
$S_x$ is the LMS parsing phrase to which we assign the symbol $o \in \Sigma^{i+1}$ as a replacement for the text $T^{i+1}$, $o$ being the LMS order of $S_x$ in \exppset{}. Recall that this construction produces the alphabet $\Sigma^{i+1}$ of $T^{i+1}$ to be non-contiguous. Therefore, the bucket $\same[u..u']$ number $o'\leq o$ is the one storing the suffixes of $T^{i+1}$ prefixed by $o$. We know that $S_x$ does not occur as a nested proper suffix within $\mathcal{F}$ (Lemma~\ref{lem:lem4}), meaning that \emph{all} its left-context symbols (those we insert in $pBWT^{i}[s_x..e_x]$) are captured\footnote{That is, the symbols in $pBWT^{i}[s_x..e_x]$ occur within the phrases of $\mathcal{F}^{i}$ that map to symbols preceding $o$ in $T^{i+1}$.} by the symbols that precede $o$ in $T^{i+1}$. By induction, $BWT^{i+1}_{bcr}[u..u']$ already has these preceding symbols in BWT order. Thus, the remaining step is to decompress those symbols, take the rightmost element of their phrases, and place them in $pBWT^{i}[s_x..e_x]$ without changing their relative order. 
\end{proof}

\begin{example}
Filling the unsolved block $pBWT^{i}[4..5]=\texttt{*}^{2}$ produced by  $\texttt{acc\$} \in \mathcal{F}^{1}$ of Figure~\ref{fig:prBWT_examp} using Lemma~\ref{lem:ind3}.  The phrase \texttt{acc\$} has LMS order $o=2 \in \Sigma^{2}$ in \exppset. Further, $\same[2..3]$ is the bucket number $o'=2$ and has the suffixes of $T^{i+1}$ prefixed by $o=2$ (see Figure~\ref{fig:ind_examp}). The corresponding range in the BWT has the sequence $BWT^{2}_{bcr}[2..3]=4\ 3$. Decompressing their phrases gives us $exp(4)=\texttt{gt}$ and $exp(3)=\texttt{agt}$ (recall that $exp$ removes the last element in the overlapping LMS phrases of $T^{i}$). If we take their rightmost symbols, we produce  $pBWT^{1}[4..5]= \texttt{t\ t}$, which matches $BWT_{bcr}[4..5]$ in Figure~\ref{fig:prBWT_examp}. 
\end{example}

Finally, we cover the case when $S_x$ occurs as a phrase $F=S_x \in \mathcal{F}^{i}$ but also as a nested proper suffix $S_x=F'[p..n_{f'}]$ in another parsing phrase $F'[1..n_{n'}] \in \mathcal{F}^{i}$. We solve its block $pBWT^{i}[s_x..e_x]=\texttt{\#}^{\ell}$ using a hybrid strategy that combines Lemmas~\ref{lem:ind2} and Lemma~\ref{lem:ind3}.

\begin{lemma}\label{lem:ind4}
Let $J$ be the set of Lemma~\ref{lem:ind2} with the occurrences in $BWT^{i+1}$ of the phrases suffixed by $S_x$. This time, these suffixes could be proper or non-proper. Assume we scan $J$ from left to right to fill $L_{S_x}$, but with a small change: if $exp^{i}(BWT^{i}[j_b])=F$, with $j_b \in J$, we append the special symbol \texttt{*} in $L_{S_{x}}$. After the scan, we replace the occurrences of \texttt{*} in $L_{S_x}$ with Lemma~\ref{lem:ind3}. The replacement of the $jth$ special symbol in $L_{S_{x}}$ is the rightmost element of $exp^{i}(BWT^{i+1}_{bcr}[u+j-1])$, where $\same[u..u']$ is the bucket with the suffixes of $T^{i+1}$ prefixed by the symbol $o \in \Sigma^{i+1}$ assigned to $F=S_{x}$.  
\end{lemma}

Our lossy grammar-like representation of \exppset{} (i.e., the vector $G^{i}$ of Section~\ref{ssec:gc}) has all the information we need to fill the unsolved blocks as described in Lemmas~\ref{lem:ind2}, \ref{lem:ind3}, and \ref{lem:ind4}. The only extra information we need to complete $pBWT^{i}$, and that it is not in $G^{i}$, is the order in which we have to rearrange the left-context symbols of $S_x$ within $pBWT^{i}_{bcr}[s_x..e_x]$. Fortunately, we can induce this information from $BWT^{i+1}_{bcr}$.  

It is also worth mentioning that the special symbols \texttt{*} and \texttt{\#} we introduced in $pBWT^{i}$ indicate what lemma we should use to fill the unsolved blocks. We will use this fact in the next section.   

\subsection{The Induction Algorithm}\label{sec:indbwtalgo}

We now describe the steps we perform during iteration $i<h$ in the induction phase (loop in lines~\ref{code:if_start}-\ref{code:if_end} of Algorithm~\ref{algo:ovgrlbwt}). Notice that the value of $i$ decrements with each round as we simulate the return from a recursion that is equivalent to that of \textsf{SA-IS}. In this case, we assume we receive as input for the iteration the string (i) $BWT^{i+1}_{bcr}$, (ii) the vector $G^{i}$ with the lossy grammar-like encoding of \exppset{}, and (iii) $pBWT^{i}$. We also assume a bit vector $V^{i}[1..\sigma^{i+1}]$, with $\sigma^{i+1}=|\exppsetme|$, indicates with $V^{i}[o]=1$ if the $oth$ string $S_x \in \exppsetme$ in LMS order is left-maximal in $\mathcal{F}^{i}$. The output of the iteration is $BWT^{i}_{bcr}$, the BCR BWT of string $T^{i}$. The details of the procedure are shown in Algorithm~\ref{algo:indBWT}.

\paragraph{Data Structures and Encoding} A central piece of our induction algorithm is an in-memory vector $P^{i}$ storing the left-context symbols of the strings $S_x \in \exppsetme$ occurring as left-maximal suffixes in $\mathcal{F}^{i}$ (Lemmas~\ref{lem:ind2} and \ref{lem:ind4}). We divide $P^{i}$ into $\mathsf{rank}_{1}(V^{i}, \sigma^{i+1})$ buckets. Thus, for a left-maximal suffix $S_x \in \exppsetme$ whose LMS order in \exppset{} is $o \in \Sigma^{i+1}$, we store its left-context symbols in the bucket $b=\mathsf{rank}_{i}(V^{i}, o)$. Put differently, the bucket $b$ of $P^{i}$ is an encoding for the sequence $L_{S_x}=pBWT^{i}[s_x..e_x]$ of the previous section. We keep $P^{i}$ in run-length-compressed format to reduce working memory and CPU time. We estimate the (run-length-compressed) area of every bucket within $P^{i}$ before the induction starts. To simplify our explanations, we will assume we already know this information and will describe how to compute it later (see Paragraph \ref{par:get_p_size}). We use the notation $P^{i}[b]$ to denote the complete run-length-compressed area of $P^{i}$ storing the symbols of bucket $b$.~We also define a process called \emph{RLC append}: let $L$ be a run-length-compressed string and let $(c, \ell)$ be an equal-symbol run we need to append into $L$. If $c$ matches the symbol of the rightmost run in $L$, we increase the length of that run by $\ell$ or append $(c, \ell)$ as a new run otherwise.

\begin{figure}[!t]
\centering
\resizebox{0.7\textwidth}{!}{%
\begin{tikzpicture}[>=stealth,thick,baseline]

\matrix (m4) [matrix of nodes, ampersand replacement=\|,
     column 1/.style={nodes={gray!70, font=\footnotesize}, anchor=base east},
     column 5/.style={nodes={gray!70, font=\footnotesize}, anchor=base east},
     column 7/.style={nodes={gray!70}, anchor=base west},
     column 9/.style={nodes={gray!70, font=\footnotesize}, anchor=base west}
     ] at (0,0)  { 
                 \| $G^{1}$                                \| $V^{1}$ \| |[white]|aaaa \|                                         \| $BWT^{2}_{bcr}$ \|              \| $P^{1}$\\
 $\texttt{aata}$ \| $\texttt{a}5_{\textcolor{gray!70}{1}}$ \| $0$     \|               \| $\texttt{g}5$, \texttt{t}               \| $4$               \| $1\ 3\ 2$    \| $(\texttt{*},2)$ \| \texttt{gta} \\
 $\texttt{acc\$}$\| $\texttt{@\$}_{\textcolor{gray!70}{2}}$\| $0$     \|               \| $\texttt{g}5$, \texttt{t}               \| $4$               \| $2$          \| $(\texttt{a},1)$ \| \\
 $\texttt{agta}$ \| $\texttt{a}4_{\textcolor{gray!70}{3}}$ \| $0$     \|            \| $\texttt{a}4$, $\texttt{g}5$, \texttt{t}\| $3$               \| $2$          \| $(\texttt{g},3)$ \| \texttt{ta} \\
 $\texttt{gta}$  \| $\texttt{g}5_{\textcolor{gray!70}{4}}$ \| $1$     \|           \| $\texttt{a}5$, \texttt{t}               \| $1$               \| $3\ 2$       \| $(\texttt{a},1)$ \| \\
 $\texttt{ta}$   \| $\texttt{@t}_{\textcolor{gray!70}{5}}$ \| $1$     \|           \| \texttt{\$}                             \| $2$               \| $4\ 1\ 3\ 2$ \|                  \\
                 \|                                        \|         \|           \| \texttt{\$}                             \| $2$               \| $4\ 2$       \|                  \\
};
\draw[dashed, gray!70] (m4-3-8.south west) -- (m4-3-8.south east);
\node[anchor = east] at (m4-1-2.west) {(A)};
\node[anchor = east] at (m4-1-6.west) {(B)};
\end{tikzpicture}
}

\caption{Construction of $P^{1}$ during the induction phase. (A) The grammar-like encoding $G^i$ of \exppset{} in Figure~\ref{fig:comp_exppset}. We subtracted $\sigma^{1}=5$ to the right symbols to simplify the example. The vector $V^{1}$ indicates which phrases of \exppset{} are left-maximal suffixes in $\mathcal{F}^{1}$. (B) The vector $BWT^{2}_{bcr}$ for the string $T^{2}$ of Figure~\ref{fig:lms_par} and the vector $P^{1}$. The dashed line in $P^{1}$ marks the boundary between its buckets. The sequence to the left of each $BWT^{2}_{bcr}[j]=o$ stores the elements we decompress from $G^{1}$ starting from symbol $o$, while the sequence to the right of $BWT^{2}_{bcr}[j]$ is its following suffix in $T^{2}$ sorted as in $GS\!A^{2}$. }
\label{fig:ind_examp}
\end{figure}

\paragraph{The Induction Process} This step consists in computing the unsolved blocks $pBWT^{i}[s_x..e_x]=\texttt{\#}^{\ell}$ that meet Lemmas~\ref{lem:ind2} and \ref{lem:ind4}. However, instead of inserting the result right away into $pBWT^{i}$, we will insert it into $P^{i}$. We visit the equal-symbol runs of $BWT^{i+1}_{bcr}$ from left to right. When we reach the $jth$ run $(o,\ell)$, with $o \in \Sigma^{i+1}$, we check if its corresponding\footnote{The string from which we obtain the symbol $o$ during iteration $i$ of the parsing phase} parsing phrase $F[1..n_f] \in \mathcal{F}^{i}$ exists as a suffix in other phrases of $\mathcal{F}^{i}$. If that is the case, we RLC append $\ell$ copies of the special symbol $\texttt{*} \notin \Sigma^{i}$ to the bucket $b=\mathsf{rank}(V^{i}, o)$ of $P^{i}$ (see Lines~\ref{code:dummy1}--\ref{code:dummy2}). Inserting $\texttt{*}$ into $P^{i}$ is equivalent to handling a block of $pBWT^{i}$ that meets Lemma~\ref{lem:ind4}.
%The purpose of \texttt{*} is to flag the rightmost $\ell$ symbols in $P^{i}[b]$ as unknown elements that we will solve in a later stage of the iteration. The value of \texttt{*} indicates that the place where we will obtain the unknown $\ell$ symbols is the area $BWT^{i+1}_{bcr}[u..u']$ for the corresponding bucket $GS\!A^{i+1}[u..u']$ with the suffixes of $T^{i+1}$ prefixed by $o$.
The next step is to decompress the left-maximal suffixes of $F$ from $G^{i}$ (see Section~\ref{ssec:gc}). This process begins by accessing $G^{i}[2o-1..2o]$. If $o'=G^{i}[2o]>\sigma^{i+1}$, then $o'$ encodes a string $F[u..n_f]=S_x \in \mathcal{F}_{exp}$, with $u>1$, whose sequence is a left-maximal suffix in $\mathcal{F}^{i}$ (see Lemma~\ref{lem:lem3}). We encoded $o$ as $o' = o+\sigma^{i}$ in $G^{i}$ to differentiate it from the symbols in $\Sigma^{i}$. On the other hand, the left-context symbol of $S_x$ is $G^{i}[2o-1] \in \Sigma^{i}$. With this information, we apply Lemma~\ref{lem:ind2} by RLC appending $\ell$ copies of $G^{i}[2o-1]$ to the bucket $\mathsf{rank}_{1}(V^{i}, G^{i}[2o]-\sigma^{i})$ of $P^{i}$. Then, we move to the next left-maximal suffix of $F$ by setting $o=G^{i}[2o]-\sigma^{i}$ and updating the range $G^{i}[2o-1..2o]$. The decompression of $F$ stops when $G^{i}[2o]\leq\sigma^{i}$, which means we reached the last symbol of $F$. For the moment, we do not know for which phrase of $\mathcal{F}^{i}$ $G^{i}[2o]$ is its left context. Hence, we update the value of the $jth$ run in $BWT^{i+1}$ to $G^{i}[2o] \in \Sigma^{i}$ and leave this run on hold to process it later when we solve the blocks of Lemma~\ref{lem:ind3}. %We describe the visit of the left-maximal suffixes of $F$ in Lines~\ref{code:lms1}--\ref{code:lms2}. After finishing the scan of $BWT^{i+1}_{bcr}$, its symbols are now over the alphabet $\Sigma^{i}$. These values are the ones we have to insert in the special entries \texttt{*} of $P^{i}$.

\begin{example}
Construction of the vector $P^{i}$ in Figure~\ref{fig:ind_examp}. Consider the run $BWT^{2}_{bcr}[1..2]=4^{2}$. Its phrase $F=\texttt{gta}$ (the $4th$ string of \exppset{} in LMS order) is left-maximal in $\mathcal{F}^{i}$ as $V^{1}[4]=1$. Hence, we RLC append $(\texttt{*},2)$ to the bucket $\mathsf{rank}(V^{1},4)=1$ of $P^{1}$. Then, the decompression of $4$ from $G^{1}$ produces $\texttt{g}5$, which means we RLC append $(\texttt{g},2)$ into the bucket $\mathsf{rank}_{1}(V^{1}, 5)=2$ of $P^{1}$. Finally, the last decompressed element $\texttt{t}$ indicates we reached the rightmost non-overlapping suffix of $F$. Therefore, we replace $BWT^{2}[1..2]=4^{2}$ by $BWT^{2}[1..2]=\texttt{t}^{2}$. Notice we decompressed the symbol $4 \in \Sigma^{2}$ in $BWT^{2}_{bcr}[1..2]$ only once but copied the information twice (i.e., the length of the run) to each bucket in $P^{1}$.
\end{example}

\paragraph{Merging the Induced Symbols} The last step in iteration $i$ is to compute the blocks $pBWT^{i}[s_x..e_x]=\texttt{*}^{\ell}$ of Lemma~\ref{lem:ind3} and solve the unfinished blocks $pBWT^{i}[s_x..e_x]=\texttt{\#}^{\ell}$ that meet Lemma~\ref{lem:ind4} (buckets in $P^{i}$ containing \texttt{*} symbols). We carry out this process by merging $pBWT^{i}$, $BWT^{i+1}_{bcr}$, and $P^{i}$. This procedure is, in practice, a merge of three sorted vectors as the symbols are already sorted by their right contexts in $T^{i}$. We use the special symbols $\texttt{*},\texttt{\#}$ in $pBWT^{i}$ and $P^{i}$ to change the active vector in the merge. The change is equivalent to switching the lemma we use to build $BWT^{i}_{bcr}$. For instance, if we see $\texttt{\#}$ in $pBWT^{i}$, we go to $P^{i}$ as this vector contains the symbols sorted with Lemma~\ref{lem:ind2}. Further, if we see \texttt{*} in $P^{i}$, we reached an unfinished block of Lemma~\ref{lem:ind4}, so we need to visit $BWT^{i+1}_{bcr}$. Finally, if we see $\texttt{*}$ in $pBWT^{i}$, we go to $BWT^{i+1}_{bcr}$ as this vector has the symbols sorted with Lemma~\ref{lem:ind3}. The merge algorithm is as follows: we scan $pBWT^{i}$ and RLC append its entries to $BWT^{i}_{bcr}$ as long as the symbols we see in $pBWT^{i}$ are not $\texttt{*}$ or $\texttt{\#}$ (Line~\ref{code:pbwtapp}). If we reach a run $\texttt{*}^{\ell}$, we RLC append the next $\ell$ symbols of $BWT^{i+1}_{bcr}$ into $BWT^{i}_{bcr}$. On the other hand, if we reach $\texttt{\#}^{\ell}$, we RLC append the next $\ell$ symbols from $P^{i}$ instead.~Additionally, as we consume the $\ell$ elements from $P^{i}$, we might reach a run $(\texttt{*}, \ell')$. When this happens, we change the active list again %. Recall that we flagged these positions in $P^{i}$ during the previous scan of $BWT^{i+1}_{bcr}$ (Line~\ref{code:flagp}).
%Now we fill them by
and RLC append the next $\ell'$ symbols of $BWT^{i+1}_{bcr}$ into $BWT^{i}_{bcr}$. If $\ell'>\ell$, we recover $\ell$ symbols of $BWT^{i+1}$ instead and decrease the active run of $BWT^{i+1}_{bcr}$ by $\ell$. Once we consume the $\ell'$ symbols from $BWT^{i+1}_{bcr}$, we return to $P^{i}$ to consume what it remains from the $\ell$ symbols. Equivalently, once we consume the $\ell$ symbols from $P^{i}$ (and possibly $BWT^{i+1}_{bcr}$), we return to $pBWT^{i}$. Lines~\ref{code:consumep1}--\ref{code:consumep2} show how we process the $\ell$ symbols of $P^{i}$. %Moving symbols from $BWT^{i+1}_{bcr}$ to $BWT^{i}_{bcr}$ (Lines~\ref{code:consumebwt1} and \ref{code:consumebwt2}) works almost in the same way. The only difference is that we do not have to deal with special symbols. 

\begin{figure}[!t]
\centering
\resizebox{0.8\textwidth}{!}{%
\begin{tikzpicture}[>=stealth,thick,baseline]
\matrix (m5) [matrix of nodes, ampersand replacement=\|,
              column 1/.style={nodes={anchor=base east}}] at (0,0) { 
           $P^{1}$          \| $=$           \| $(\texttt{*},2)$ \| $(\texttt{a}, 1)$ \| $(\texttt{g}, 3)$ \| $(\texttt{a}, 1)$ \\
           $BWT^{2}_{bcr}$\| $=$           \| $(\texttt{t},4)$ \| $(\texttt{\$}, 2)$ \\
           $pBWT^{1}$       \| $=$           \| $(\texttt{c},2)$ \| $(\texttt{*},4)$ \| $(\texttt{a},1)$ \| $(\texttt{c}, 2)$ \| $(\texttt{a},2)$ \| $(\texttt{\#},7)$ \\
                            \| |[white]| aaa \|  \|  \|  \|  \|  \|  \\
           $BWT_{bcr}$      \| $=$           \| $(\texttt{c}, 2)$ \| $(\texttt{t},4)$ \| $(\texttt{a},1)$ \| $(\texttt{c},2)$ \| $(\texttt{a},2)$ \| $(\texttt{\$},2)$ \| $(\texttt{a},1)$ \| $(\texttt{g},3)$ \| $(\texttt{a},1)$ \\
};
\end{tikzpicture}
}
\caption{Merge of $pBWT^{1}$, $BWT^{i+1}$ and $pBWT^{i}$ to produce $BWT_{bcr}$ for the string $T^{1}=\texttt{gtacc\$gtaatagtacc\$}$ of Figure~\ref{fig:lms_par}. We constructed $pBWT^{1}$ in Figure~\ref{fig:prBWT_comp}, and vectors $BWT^{i+1}_{bcr}$ and $P^{1}$ in Figure~\ref{fig:ind_examp}. }
\label{fig:bcr_bwt_merge}
\end{figure}

\begin{example}
Producing $BWT_{bcr}$ in Figure~\ref{fig:bcr_bwt_merge}. The merge starts by appending $pBWT^{1}[1..2] = (\texttt{c},2)$ into $BWT_{bcr}$. The next run $pBWT^{i}[3..6]=(\texttt{*},4)$ is flagged with the special symbol \texttt{*}. Therefore, we change the active vector in the merge to $BWT^{2}_{bcr}$, and RLC append $BWT^{2}_{bcr}[1..4]=(\texttt{t},4)$ into $BWT_{bcr}$. Now the active merge position of $BWT^{2}_{bcr}$ becomes $BWT^{2}_{bcr}[5]$. We switch back to $pBWT^{1}$ to RLC append $pBWT^{1}[7..11] = (a,1) (c,2) (a,2)$ into $BWT_{bcr}$ as they have symbols in $\Sigma$. However, the next run $pBWT^{1}[12..18] = (\texttt{\#}, 7)$ is flagged with the special symbol $\texttt{\#}$, indicating we need to switch the active vector in the merge again and extract the next seven symbols of $BWT_{bcr}$ from $P^{1}[1..7]$. Still, the first run $P^{1}[1..2]=(\texttt{*}, 2)$ is, in turn, flagged with $\texttt{*}$, which means we have to go to $BWT^{2}_{bcr}[5..7]=(\texttt{\$},2)$ ($5$ being the active merge position of $BWT^{2}_{bcr}$) and RLC append $(\texttt{\$}, 2)$ into $BWT^{1}$. Finally, we return to $P^{1}[3..7]$ to obtain the remaining $7-2=5$ symbols. Notice that the resulting $BWT_{bcr}$ matches the BCR BWT of $T^{1}$ we presented in Figure~\ref{fig:prBWT_examp}.
\end{example}

\paragraph{Precomputing the Number of Buckets}\label{par:get_p_size} The last aspect we address for the construction of $BWT^{i}_{bcr}$ is how to compute the area of every bucket within $P^{i}$. We initialize a vector $P'[1,\mathsf{rank}(V^{i}, \sigma^{i+1})]$ of pairs. Every pair $P'[b] = (x, y)$ is a temporal variable to count the number of runs in the bucket $b$ of $P^{i}$. We obtain the values for $P'$ with one decompression of $BWT^{i+1}_{bcr}$ from left to right. The procedure is almost equal to what we show in Lines~\ref{code:popp1}--\ref{code:popp2} of Algorithm~\ref{algo:indBWT}, although we do not modify $BWT^{i+1}_{bcr}$. Recall that this procedure yields a sequence $(c_1, b_1), \ldots, (c_x, b_x)$, where $b_j$ is a bucket we visit in $P^{i}$ and $c_j$ is the left-context symbol for the phrase associated with the bucket $b_j$. We compare the symbol $c_j$ with the left element in $P'[b_j]$: if they are equal, we do nothing. If, on the other hand, they differ, we increase the right element in $P'[b_j]$ by one and set the left element to $c_{j}$. Once we scan $BWT^{i+1}$, the right element of every pair $P'[b]$ will contain the maximum number of runs we could see in $P^{i}[b]$. These values can decrease once we recover the symbols for the unsolved areas (flagged with \texttt{*}) of $P^{i}$.

\begin{algorithm}[!htp]
\caption{Induction of $BWT^{i}_{bcr}$}
\label{algo:indBWT}
\begin{algorithmic}[1]
\small 

\Require $pBWT^{i}$, $G^{i}$, $BWT^{i+1}_{bcr}$, and $V^{i}$

\State $P^{i} \gets$ estimate size of $P^{i}$'s buckets

\For{$j=1$ to number of runs in $BWT^{i+1}_{bcr}$}\label{code:popp1} \Comment{populate $P^{i}$'s buckets}
	\State $(o, \ell) \gets $ $jth$ run in $BWT^{i+1}_{bcr}$ \Comment{symbol $o \in \Sigma^{i+1}$ assigned to $F \in \mathcal{F}^{i}$}
	\If{$V^{i}[o]=1$}\label{code:dummy1}\Comment{$F$ occurs as a left-maximal suffix in $\mathcal{F}^{i}$}
		\State RLC append $\ell$ copies of $\texttt{*}$ to the bucket $\mathsf{rank}_{1}(V^{i}, o)$ of $P^{i}$ \label{code:flagp} 
	\EndIf\label{code:dummy2}

	\While{ $G{i}[2o] > \sigma^{i}$}\label{code:lms1} \Comment{visit left-maximal suffixes of $F$}
		\State RLC append $\ell$ copies of $G^{i}[2o-1]$ to $P^{i}[\mathsf{rank}_{1}(V^{i}, G^{i}[2o]-\sigma^{i})]$
		\State $o \gets G^{i}[2o]-\sigma^{i}$
	\EndWhile\label{code:lms2}
	\State set $G^{i}[2o]$ as the symbol for the $jth$ run of $BWT^{i+1}_{bcr}$ 
\EndFor\label{code:popp2}

\For{$j=1$ to number of runs in $pBWT^{i}$}\label{code:assmbwt1} \Comment{assemble $BWT^{i}_{bcr}$}
	\State $(o, \ell) \gets $ $jth$ run in $pBWT^{i}$
	\If{$o = \texttt{*}$} \label{code:activebwt1}
		\State RLC append next $\ell$ symbols of $BWT^{i+1}_{bcr}$ into $BWT^{i}_{bcr}$\label{code:consumebwt1}
	\ElsIf{$o = \texttt{\#}$} \label{code:activebwt2}
		\While{$\ell>0$}\label{code:consumep1}
			\State $(o', \ell') \gets $ active run in $P^{i}$  
			\If{$\ell'>\ell$}
				\State decrease length of active run in $P^{i}$ by $\ell'-\ell$
				\State $\ell' \gets \ell$ 
			\Else
				\State move to the next run in $P^{i}$
			\EndIf

			\If{$o'=\texttt{*}$}
				\State RLC append the next $\ell'$ symbols of $BWT^{i+1}_{bcr}$ into $BWT^{i}_{bcr}$\label{code:consumebwt2}
			\Else
				\State RLC append $\ell'$ copies of $o'$ into $BWT^{i}_{bcr}$ 
			\EndIf

			\State $\ell \gets \ell - \ell'$
		\EndWhile\label{code:consumep2}
	\Else
		\State RLC append $\ell$ copies of $o$ to $BWT^{i}_{bcr}$\label{code:pbwtapp} 
	\EndIf
\EndFor\label{code:assmbwt2} 
\State \textbf{return} $BWT^{i}_{bcr}$
\end{algorithmic}
\end{algorithm}

\subsection{The Complexity of Our Method}

%We show that the construction of the BCR BWT remains linear, even though we perform compression during the intermediate steps.

We now present the theoretical bounds of \textsf{grlBWT}.

\begin{theorem}
Let $\mathcal{T}=\{T_{1},T_{2},\ldots, T_{k}\}$ be a collection with $k=|\mathcal{T}|$ strings over the alphabet $\Sigma$ and let $T=T_1\texttt{\$}{\cdots}T_{k}\texttt{\$}$ be a string of $n=|T|$ symbols over the alphabet $\Sigma \cup \texttt{\$}$ storing the concatenation of $\mathcal{T}$. Additionally, let $m$ be the length of the longest string in $\mathcal{T}$. The algorithm $\mathsf{grlBWT}$ constructs the BCR BWT of $\mathcal{T}$ in $O(n + k\log m)$ expected time and requires $O((n+k\log m)\log n)$ bits of working space.  
\end{theorem}

\begin{proof}
Our algorithm \textsf{grlBWT} is an adaptation of the algorithm \textsf{SA-IS} of Nong et al.~\cite{n2009li}. The authors showed that this method takes linear time and uses $O(n \log n)$ bits of space. The relevant observation is that the length $n^{i+1}$ of the string $T^{i+1}$ is at most $\frac{n^{i}}{2}$, $n^{i}$ being the length of the previous $T^{i}$. Thus, the cumulative lengths of strings $T^{1}, T^{2},\ldots, T^{h}$, with $h=O(\log n)$, that \textsf{SA-IS} produces is no more than $2n$. On the other hand, the amount of work in every level $i$ is proportional to $n^{i}$, so the overall cost of the algorithm is linear on the input text. We adapt this argument to our method to probe our theoretical bounds.  

Running the LMS parsing reduces each substring $T^{i}[j..j']$ of length $j'-j+1>1$ representing a string $exp(T^{i}[j..j'])=T_{x}\texttt{\$}$ of $\mathcal{T}$ to another substring $T^{i+1}[p..p']$ that is at most half the length of $T^{i}[j..j']$. However, it might happen that the resulting substring $T^{i+1}[p..p']$  has length  $p'-p+1=1$. We say that $T^{i+1}[p]$ is uncompressible because it already covers a full string of $\mathcal{T}$, but we can not reduce its size with a new round of LMS parsing. If there are substrings of $T^{i+1}$ encoding elements of $\mathcal{T}$ that are still compressible, then \textsf{grlBWT} will incur in another recursion $i+2$ and  carry $T^{i+1}[p]$ to $T^{i+2}$ as another symbol. This feature implies that the length $n^{i+2}$ of $T^{i+2}$ is not guaranteed to be at most $\frac{n^{i+1}}{2}$ like in \textsf{SA-IS}.  Instead, $n^{i+2}$ is upper bounded by $O(\frac{n^{i+1}}{2} + k)$ as there are less than $k$ uncompressible symbols $T^{i+1}[p]$ we can carry from $T^{i+1}$ to $T^{i+2}$. Our method finishes the recursions when all the strings of $\mathcal{T}$ encoded by $T^{i}$ are uncompressible, meaning that we can not produce more than $\log m$ recursions. Thus, the cumulative lengths of our strings $T^{1}, T^{2}, \ldots, T^{h}$, with $h=O(\log m)$, is no more than $\sum_{i=0}^{\log m} \frac{n}{2^{i}} + k \leq 2n + k\log m$.   

We now analyse the amount of work we perform in every recursion level $i$. Creating the dictionary $D^{i}$ from $T^{i}$ using a standard hash table takes $O(n^{i})$ expected time and requires $O(n^{i}\log n^{i})$ bits of space. The construction of  \gsa{} runs in $O(n^{i})$ time and space as we use ISS to build it, and the number of symbols in the keys of $D^{i}$ is never greater than $n^{i}$. The extra steps of the parsing iteration only require a constant number of linear scans over \gsa.  During the induction phase, we only perform linear scans over $BWT^{i+1}_{bcr}$ and $pBWT^{i}$. We still have the cost of accessing the left-maximal suffixes of $\mathcal{F}^{i}$ when we scan $BWT^{i+1}_{bcr}$. However, our simple grammar-like representation $G^{i}$ (Section~\ref{ssec:gc}) supports random access in $O(1)$ time to the symbols, and the number of left-maximal suffixes we visit during the scan of $BWT^{i+1}_{bcr}$ is no more than $n^{i}$. Thus, the cost of every recursion level $i$ is $O(n^{i})$ time and $O(n^{i}\log n^{i})$ bits of space. Summing up the $h=O(\log p)$ recursions levels, the total cost of \textsf{grlBWT} is $O(n+k\log m)$ time and $O((n+k\log m)\log n)$ bits of space.  
\end{proof}

While the only general bound in terms of $n$ is $O(n+k\log m) \subseteq O(n\log n)$, this bound is reached only in degenerate cases (e.g., one string of length $n/2$ and $n/2$ strings of length $1$ or $2$). In typical cases, where $m = O(n/k)$, it holds $O(n+k\log m) = O(n)$. This holds even if the largest string in $\mathcal{T}$ is significantly larger than the average, for example $m = O((n/k)^c)$ for a constant $c$. In practical applications, the worst case is probably that of $k$ short reads whose length $\ell$ is a few hundreds, and then still $k\log m = (n/\ell)\log\ell$ is an order of magnitude less than $n$.

\section{Refining Our Algorithm}\label{sec:ref_algo}

In this section, we explain how to produce smaller temporary data structures during the parsing of $T^{i}$. As before, we are interested in a lightweight compression scheme that improves the overall performance of $\mathsf{grlBWT}$, especially when $\mathcal{T}$ is not that repetitive. The challenge is to find a balance between the extra compression we gain for $\mathcal{F}^{i}$ and $T^{i+1}$ and the computational resources we use to achieve it. A bad choice could yield good space reductions at the cost of making \textsf{grlBWT} slower. On top of that, we also need to be careful not to compromise the induction phase of \textsf{grlBWT} with the changes we introduce.

Our strategy consists of merging consecutive substrings in the LMS parsing of $T^{i}$ into one \emph{super} phrase whenever the merge does not affect the construction of $BWT^{i}_{bcr}$. We formalize this idea as follows: 

\begin{definition}\label{def:super_lms}
A super phrase is a substring $T^{i}[j_1.. j_{z}]$ with the following properties: (i) it spans a group of consecutive substrings $F_1 = T^{i}[j_{1}..j_{2}], F_2=T^{i}[j_2..j_3],\ldots, F_z= T^{i}[j_{z-1}..j_{z}]$ in the LMS parsing whose associated phrases $F_1, F_2, \ldots, F_z \in \mathcal{F}^{i}$ always appear together in $T^{i}$ (also as parsing phrases), and in the same order. Thus, for each $oth$ occurrence $F_{p} = T^{i}[j_{p,o}..j_{p+1,o}]$ of $F_{p}$, it always holds that the phrase following that occurrence is $F_{p+1}=T^{i}[j_{p+1, o}..j_{p+2,o}]$, the $oth$ occurrence of $F_{p+1}$. Additionally, (ii)  any of the phrases $F_{1}, F_2, \ldots, F_{z-1}$ contain a proper suffix that is left-maximal in $\mathcal{F}^{i}$.
\end{definition}

If $F=T^{i}[j_1..j_z]$ is a super phrase, we store $F$ in $\mathcal{F}^{i}$ instead of the internal phrases $F_1, F_2,\ldots, F_z$ individually. We refer to $F$ as a \emph{super} phrase. We now explain why super phrases do not affect the induction of $BWT^{i}_{bcr}$.

\begin{lemma}
Let $F=T^{i}[j_1..j_z] \in \mathcal{F}^{i}$ be a super phrase of $T^{i}$. The structure of $F$ does not affect the construction of $BWT^{i}_{bcr}$ during the induction phase.
\end{lemma}

\begin{proof}
We first prove property (i) of Definition~\ref{def:super_lms} by contradiction. Assume that the LMS parsing phrase $F_p=T^{i}[j_p.. j_{p+1}]$ within the super phrase $F=T^{i}[j_1..j_z]$ matches the sequence of another LMS parsing phrase $F_p=T^{i}[l..l']$, with $(l,l') \notin [j_1..j_z]$. Additionally, suppose the occurrence $F_p=T^{i}[l..l']$ does not meet the conditions to be encapsulated within a super phrase, so both $F_p$ and $F$ become members of $\mathcal{F}^{i}$. In \gsa~(i.e., the generalized suffix array of $\mathcal{F}^{i}$), there will be a value that points to the suffix of $F$ prefixed by $F_p=T^{i}[j_p.. j_{p+1}]$ and another value pointing to the full phrase $F_p=T^{i}[l..l']$. Both suffix array values encode suffixes of $\mathcal{F}^{i}$ labelled $F_{p}$, and hence, our algorithm will induce the lexicographical order of the suffixes of $T^{i}$ prefixed by $F_{p}$ from $BWT^{i+1}_{bcr}$ (see Lemmas~\ref{lem:ind1} and \ref{lem:ind2}). Now, for the induction to happen, we need two symbols in $T^{i+1}$, one encoding the occurrence $F_{p}=T^{i}[j_p.. j_{p+1}]$ and another encoding the occurrence $F_{p}=T^{i}[l..l']$. However, $T^{i}[j_{p}.. j_{p+1}]$ will not have a symbol in $T^{i+1}$ as it is fully encapsulated within $F$, meaning that we will not be able to perform the induction. This situation does not happen if $F_{p}$ always occurs in $T^{i}$ as a substring of $F$ as we have enough context within the super phrase the solve the BWT range associated with $F_{p}$.

The proof for property (ii) of Definition~\ref{def:super_lms} is similar: assume this time that the set of LMS parsing phrases $F_{1}, F_{2},\ldots, F_{z}$ encapsulated by the super phrase $F$ meet property (i). However, one of them (say $F_{u}$, with $u<z$) has a proper suffix $S_x$ that is left-maximal. That is, $S_x$ also occurs as a proper suffix in, at least, one other phrase $Y \in \mathcal{F}^{i}$, and the symbols preceding the occurrences of $S_x$ in $F_{u}$ and $Y$ are different. During parsing iteration $i$, we do not have enough information in $\mathcal{F}^{i}$ to order the suffixes of $T^{i}$ prefixed by $S_x$. We have the right context for the occurrence of $S_x$ under $F_{u}$ as $F_{u}$ is a substring in $F$ (not a suffix), but we do not have the right context for the occurrence under $Y$. Our method solves this problem in the next parsing iteration $i+1$ when comparing the symbols in $T^{i+1}$ assigned to $F_{u}$ and $Y$. However, with the introduction of super phrases, we have symbols in $T^{i+1}$ for $F$ and $Y$, but not for $F_{u}$, and the right context of $F$ is the right context of $F_{z}$, not of $F_{u}$. This situation leads to an error during the induction of $BWT^{i}_{bcr}$ when solving the occurrences of $S_x$. 
\end{proof}

The introduction of a super phrase $F=F_1,F_2,\ldots,F_z$ removes $z-1$ regular phrases from $\mathcal{F}^{i}$ and decreases the number of symbols in $\mathcal{F}^{i}$ by $z-1$. In our regular parsing scheme, the last element of every $F_{u}$, with $u<z$, was a copy of the first element of $F_{u+1}$, because regular LMS parsing phrases overlap, but now that $F_u$ and $F_{u+1}$ belong to the same super phrase, that copy is not in the parsing set. Also, notice that the length of $T^{i+1}$ also decreased by $z-1$ symbols.

We restricted our definition of super phrases to avoid redundancy in $\mathcal{F}^{i}$. One could, for instance, allow an LMS parsing phrase $F_{u}$ to occur in different super phrases as a substring as long as all the occurrences of $F_{u}$ in $T^{i}$ were covered by a super phrase. This relaxation could increase the number of super phrases without affecting the induction of $BWT^{i}_{bcr}$. However, it could also create multiple copies of $F_{u}$ within $\mathcal{F}^{i}$, which is undesirable. 

Introducing super phrases requires a small change in the algorithm that constructs \gsa. In the original version we described in Section~\ref{sec:c_preBWT}, we first insert the last symbols of every phrase $F \in \mathcal{F}^{i}$ at the end of its corresponding bucket in \gsa. Subsequently, we perform one scan of \gsa~to insert L-type symbols and another scan to insert the S-type symbols. Now, with the introduction of super phrases, we proceed as follows: before the suffix induction, we visit every phrase $F \in \mathcal{F}^{i}$ and put its last symbol at the end of its corresponding bucket in \gsa. Additionally, if $F$ is a super phrase, we scan it to find all its internal LMS-type symbols and insert their positions in $\mathcal{F}^{i}$ at the end of their corresponding buckets in \gsa. Finally, we can proceed with the induction of the L-type and S-type symbols as usual. 

\subsection{Practical Considerations of Super Phrases}\label{sec:prac_super_phrases}

When parsing $T^{i}$, we do not know beforehand if a specific group of consecutive LMS parsing substrings will produce a super phrase. The naive approach to check that information would be to construct a suffix tree of $T^{i}$, but this solution is impractical. An alternative idea would be building a regular version of $\mathcal{F}^{i}$ first, computing some satellite information about the phrases directly from the dictionary, and then performing an extra scan of $T^{i}$ to gather the super phrases, hoping that there are enough of them so that the overhead of the extra scan produces a much smaller version of $\mathcal{F}^{i}$ and $T^{i}$. Still, there is a third alternative that might not capture all the super phrases, but it is for free. During every parsing round $i-1$, we mark which symbols in the alphabet of $T^{i}$ are unique. Then, when we scan $T^{i}$ in the $ith$ parsing round, we proceed as follows: every time we reach an LMS-type symbol $T^{i}[j]$ (i.e., a break in the parsing), we check if $T^{i}[j-1]$ is unique. If so, we skip the break and continue extending the current phrase to the left. We keep applying this procedure until we reach a break where the condition does not hold.

It is easy to see that if the symbol $T^{i}[j-1]$ is unique in $T^{i}$, then its enclosing LMS parsing phrase $F=T^{i}[j'..j]$ spells a sequence that is unique in $T^{i}$. On the other hand, $T^{i}[j-1]$ is the last symbol we consider in $F$ to get its left-maximal suffixes (recall that the last symbol in $F$ is redundant), and because $T[j-1]$ is unique, $F$ does not have left-maximal suffixes. These observations allow us to append $T[j'..j]$ directly to a super phrase without further preprocessing. 

\section{BCR BWT with Different String Order}\label{sec:optbwt}

One can reorder the strings of $\mathcal{T}$ to reduce the number of runs in its BCR BWT, thus improving the compression. This technique is useful for applications where the order of the strings is irrelevant. In this regard, Bentley et al.~\cite{ben20ont} showed a linear-time procedure to obtain the smallest BCR BWT (in terms of the number of runs) one can get by permuting the order of the strings in $\mathcal{T}$. This method, referred to here as \textsf{CAO}, requires as input the BCR BWT $BWT_{bcr}$ of $\mathcal{T}$ and the partition $\mathcal{A}$ over $BWT_{bcr}$ that induces equal suffixes of $\mathcal{T}$. More specifically, a block $BWT_{bcr}[s..e]$ belongs to $\mathcal{A}$ if $BWT_{bcr}[s.. e]$ stores the left-context symbols for suffixes of $\mathcal{T}$ that spell the same sequence. The output of \textsf{CAO} is the optimal BCR BWT $BWT_{bcr}$ for $\mathcal{T}$, referred to here as $BWT_{opt}$ (Section~\ref{ssec:bwt} describes \textsf{CAO} in more detail). 

In this section, we address the problem of efficiently building  $\mathcal{A}$ so that we can apply \textsf{CAO} to transform $BWT_{bcr}$ into $BWT_{opt}$. Our strategy involves inducing the tuples of $\mathcal{A}$ during the execution of \textsf{grlBWT}. Nevertheless, we compute a sampled version of $\mathcal{A}$ (denoted $\mathcal{A}')$ that only considers BWT blocks (tuples) with more than one distinct symbol. The important observation is that there is no point in keeping in main memory tuples of $\mathcal{A}$ with one symbol during the execution of \textsf{CAO}, as they are already sorted. For instance, the vector $\mathcal{A}$ of Figure~\ref{fig:bct_bwt_examp} becomes $\mathcal{A}' = \texttt{.}\ \texttt{.}\ \texttt{.}\ [(\texttt{a},1), (\texttt{c},1)]\ \texttt{.}\ \texttt{.}\ [(\texttt{c},1),(\texttt{a},2)]\ \texttt{.}$, where the dots indicate the tuples we removed. %An extra bit vector marks which tuples of $\mathcal{A}$ are also in $\mathcal{A}'$. 

The execution of \textsf{CAO} requires an encoding that regards each block $BWT_{bcr}[s..e] \in \mathcal{A}$ as a tuple of up to $\sigma$ pairs where each element $(c, \ell)$ stores a symbol $c$ occurring in $BWT_{bcr}[s..e]$ and its frequency $\ell$ in $BWT_{bcr}[s..e]$. In practice, if we know the boundaries in $BWT_{bcr}$ for the blocks in $\mathcal{A}$, we can compute the \textsf{CAO} encoding on the fly using the run-length compressed version of $BWT_{bcr}$. However, \textsf{grlBWT} produces the ranges for the blocks in $\mathcal{A}'$, so we need to do a small change to \textsf{CAO}: if for a block $BWT_{bcr}[s..e] \in \mathcal{A}'$ we need to visit its preceding block (respectively, the following block) to check for adjacent matches, and this block is not in $\mathcal{A}'$, we use the symbol $BWT_{bcr}[s-1]$ to do the check instead (respectively, $BWT_{bcr}[e+1]$). 

Before explaining our idea, we will briefly redefine some key concepts for clarity. The parsing set $\mathcal{F}^{i}$ stores the distinct phrases in the LMS parsing of $T^{i}$. The suffixes in $\mathcal{F}^{i}$ form a string set $\mathcal{S}^{i}$ that induces a partition over \sa, the generalized suffix array of $T^{i}$. In this partition, every block $\same[s_x..e_x]$ stores the suffixes of $T^{i}$ prefixed by $S_x \in \mathcal{S}^{i}$, the $xth$ string of $\mathcal{S}^{i}$ in LMS order. We say that $BWT^{i}_{bcr}[s_x..e_x]$ is associated with $S_x$ because it stores its left-context symbols in $T^{i}$. Additionally, $\gsame$ is the generalized suffix array of $\mathcal{F}^{i}$. We also consider a partition over $\gsame$, where every block $\gsame[j..j']$ encodes the different suffixes of $\mathcal{F}^{i}$ spelling the same sequence $S_x \in \mathcal{S}^{i}$. We use $\mathcal{S}^{i}$ to define an expanded parsing set \exppset{} that contains the elements of $\mathcal{F}^{i}$ plus the strings in $\mathcal{S}^{i} \setminus \mathcal{F}^{i}$ that only occur as left-maximal suffixes in $\mathcal{F}^{i}$. The elements in the expanded parsing set are kept in LMS order. 

\subsection{Partitioning the BCR BWT}

We will produce the sequence $A^{i} = (s_1, e_1), \ldots, (s_z, e_z)$ with the blocks in $BWT^{i}_{bcr}$ induced by the substrings of $T^{i}$ that expand to equal suffixes of $\mathcal{T}$. Specifically, each $jth$ pair $(s_x, e_x) \in A^{i}$ is the $jth$ block $BWT^{i}_{bcr}[s_x..e_x]$ whose string $S_x \in \mathcal{S}^{i}$ expands to a string $exp(S_x) \in \Sigma^{*}$ suffixed by \texttt{\$} (i.e., $exp(S_x)$ is a suffix in $\mathcal{T}$). As before, we will construct a preliminary version of $A^{i}$, called $pA^{i}$, during the parsing iteration $i$, and then we will combine $A^{i+1}$ and $pA^{i}$ to produce $A^{i}$ during iteration $i$ of induction. The final list $A^{1}$ stores the ranges of $\mathcal{A}'$ that we will use to run \textsf{CAO}. 

\paragraph{Encoding} We will update Algorithms~\ref{algo:prebwt} and \ref{algo:indBWT} to implement the ideas in the paragraph above. Still, modifying these methods requires a bit of extra work as they operate over run-length-compressed data, and the ranges in $A^{i}$ are indexes in the plain version of $BWT^{i}_{bcr}$. This difference in the encoding means that a range in $A^{i}$ might not exist explicitly in the run-length-compressed version of $BWT^{i}_{bcr}$. From now on, we consider the arrays involved in the construction of $A^{i}$ to be in plain format unless we state otherwise. For example, when we say that $u$ is the \emph{uncompressed} position of $BWT^{i+1}_{bcr}$, we mean that $u$ is an index within the plain version of $BWT^{i+1}_{bcr}$. The same idea applies to the other vectors. This plain encoding is only logical and intended to simplify our explanations. We omit the details on how to implement the construction of $A^{i}$ using run-length-compressed data structures. 

\subsubsection{Parsing Phase}

Our algorithm to construct $pA^{i}$ is a modification of Algorithm~\ref{algo:prebwt}. During the execution of Line~\ref{code:newrange}, when we start to consume a new range in the partition of \gsa{}, we check if the previous range $\gsame[s_x..e_x]$ is associated with a phrase $S_x \in \mathcal{S}^{i}$ such that $exp(S_x)$ is a suffix in $\mathcal{T}$. %If that is the case, we append its left-context symbols to $pBWT^{i}$ without changing the order in which they appear in $\gsame[s_x..e_x]$.
If this condition holds, and $S_x$ is a left-maximal suffix in $\mathcal{F}^{i}$ (Line~\ref{code:lmif}) that does not occur as a full phrase, we append a new pair into $pA^{i}$: let $s_x$ be the uncompressed position in $pBWT^{i}$ where we inserted the leftmost symbol $\texttt{\#}$ in $pBWT^{i}$ associated with $S_x$, and let $e_x$ be the uncompressed position in $pBWT^{i}$ where we inserted the rightmost copy of $\texttt{\#}$. We append the pair $(s_x, e_x)$ into $pA^{i}$.

The left-maximal condition of $S_x$ implies that its range $(s_x, e_x) \in pA^{i}$ will contain more than one distinct symbol in $\Sigma^{i}$. This is the kind of entries \textsf{CAO} needs to sort to produce $BWT_{opt}$. On the other hand, by not storing $(s_x, e_x)$ when $S_x$ is a full phrase in $\mathcal{F}^{i}$, we avoid redundancy: assume $S_x$ is assigned the symbol $o \in \Sigma^{i+1}$ in the parsing iteration $i$. The pair $(s_x, e_x) \in pA^{i}$ of $S_x$ will be equivalent to the pair $(s_{x'}, e_{x'}) \in pA^{i+1}$ obtained from the string $S'_x=o \in \mathcal{S}^{i+1}$.

%Another important observation is that if $S_x$ is always a proper left-maximal suffix in $\mathcal{F}^{i}$ and $exp(S_x)$ is a suffix in $\mathcal{T}$, the unsolved block in $pBWT^{i}$ associated with $S_x$ will not ``unsolved'', and $S_x$ will not be a phrase in \exppset{}. This change is on purpose and related to how \textsf{CAO} works. In the original version of Algorithm~\ref{algo:prebwt}, $S_x$ produces an unsolved block in $pBWT^{i}$ because there is not enough information to sort its left-context symbols. Thus, \textsf{grlBWT} inserts $e_x-s_x+1$ copies of \texttt{\#} in $pBWT^{i}[s_x..e_x]$ as a flag to indicate that this region will be solved during the induction phase. The concept of induction implies that the way in which we store the left-context symbols of $S_x$ within $BWT^{i}_{bcr}[s_x..e_x]$ affects the order in which we insert the left-context symbols of $exp(S_x)$ in its associated range of $BWT^{1}_{bcr}$. However, because $\textsf{grlBWT}$ is now constructing $BWT_{opt}$, things work differently. It is \textsf{CAO} who decides the order in $BWT_{opt}$ of the left-context symbols of $exp(S_x)$ as its occurrences in $\mathcal{T}$ do not have a real right context. This ordering does not depend on $BWT^{i}_{bcr}[s_x..e_x]$ (i.e., it is not induced); it only uses information in $\mathcal{A}$ (or $\mathcal{A}'$ in our case). As a consequence, there is no use in sorting $BWT^{i}_{bcr}[s_x..e_x]$ during the induction phase. We insert the symbols to this range as we access them from $\gsame[s_x..e_x]$, and then \textsf{CAO} will do the rest.

Notice that $pA^{i}$ stores each block $BWT^{i}_{bcr}[s_x..e_x]$ that later will become the range in $BWT^{1}_{bcr}$ associated with $exp(S_x)$. $\mathsf{CAO}$ will use this block to produce $BWT_{opt}$. %This is why we store $(s_x, e_x)$ in $pA^{i}$. We will explain how to use this information in the next section. 

\subsubsection{Induction Phase}

We now explain how to modify Algorithm~\ref{algo:indBWT} (the induction iteration) to compute $A^{i}$. Our method has three main steps. First, it produces a new sequence $A^{i}_{P}$ with the pairs of $A^{i}$ that we will induce from $P^{i}$ (i.e., the vector of Algorithm~\ref{algo:indBWT} storing the sorted left-context symbols of the left-maximal suffixes in $\mathcal{F}^{i}$). Then, it updates the pairs in $A^{i+1}$ and $A^{i}_{P}$ so they reference positions in $BWT^{i}_{bcr}$. Finally, it merges $A^{i+1}$, $A^{i}_{P}$, and $pA^{i}$ into one single sequence $A^{i}$.

We start with the construction of $A^{i}_{P}$ during the execution of Lines~\ref{code:popp1}-\ref{code:popp2}, when we populate $P^{i}$ in one scan of $BWT^{i+1}_{bcr}$. We logically divide $A^{i}_{P}$ into buckets so that each bucket $c$ stores the pairs induced from the bucket $c$ of $P^{i}$. As we previously did with $P^{i}$, we use the notation $A^{i}_{P}[c]$ to refer to the area within $A^{i}_{P}$ that contains the ranges of bucket $c$.

Before the scan of $BWT^{i+1}_{bcr}$, we set $A^{i+1}[1]=(s, e)$ as the active pair we use to fill $A^{i}_{P}$. We also initialize a set $H$ that will store the different buckets of $P$ we visit during the scan of $BWT^{i+1}_{bcr}[s..e]$. Subsequently, when we visit the symbols of $BWT^{i+1}_{bcr}$, we proceed as follows. Assume we are in the uncompressed position $u$ of $BWT^{i+1}_{bcr}$, where $BWT^{i+1}_{bcr}[u] \in \Sigma^{i+1}$ is the symbol assigned to $F \in \mathcal{F}^{i}$. Also, assume that $u=s$ matches the left element of $(s, e)$. If $F$ occurs as a proper suffix in $\mathcal{F}^{i}$, we compute its bucket $b$ in $P^{i}$, record $b$ in $H$, and append a new pair $(u',u')$ in $A^{i}_{P}[b]$, where $u'$ is the uncompressed position within $P^{i}[b]$ where we store the left-context of $F$ (Lines~\ref{code:dummy1}-\ref{code:dummy2}). Then, when we visit every left-maximal suffix $S_x \in \mathcal{S}^{i}$ of $F$ (Lines~\ref{code:lms1}-\ref{code:lms2}), we apply the same procedure: we compute the bucket $b$ in $P^{i}$ associated with $S_x$, record $b$ in $H$, and append a new pair $(u', u')$ in $A^{i}_{P}[b]$, where $u'$ is the uncompressed position within $P^{i}[b]$ where we store the left-context of $S_x$. Later in the scan of $BWT^{i+1}_{bcr}$, when $s < u \leq  e$, we proceed slightly differently: for every new bucket $b$ of $P^{i}$ we visit during the decompression of $F$, we check first if $b$ exists in $H$. If that is the case, we increase the right value in the rightmost pair of $A^{i}_{P}[b]$ by one. On the other hand, if $b$ is not in $H$, we record $b$ in $H$ and append a new pair $(u',u')$ in the bucket $A^{i}_{P}[b]$. Additionally, when $u=e$, we flush the content in $H$ and set the next pair in $A^{i+1}$ as the active element we will use from now on to fill $A^{i}_{P}$. %Every time we insert a new pair $(s, e)$ in $A^{i}_{P}[b]$, we first check if the rightmost pair in that bucket spans an area within $P^{i}[u]$ that spells an equal-symbol run. If that is the case, we remove that pair from $A^{i}_{P}[b]$ before inserting $(s, e)$.

Once we finish the traversal of $BWT^{i+1}_{bcr}$, we transform the pairs in $A^{i}_P$ to absolute values. Concretely, a pair $(s, e)$ in the bucket $A^{i}_{P}[b]$ becomes $(s+s', e+s')$, where $s'$ is the cumulative number of symbols in the buckets $b'<b$ of $P^{i}$.

The next step in the induction iteration is to update the pairs in $A^{i+1}$ and $A^{i}_{P}$ so they reference ranges within $BWT^{i}_{bcr}$. We carry out this process during the merge of $BWT^{i+1}_{bcr}$, $pBWT^{i}$, and $P^{i}$ (Lines~\ref{code:assmbwt1}-\ref{code:assmbwt2}). Recall that we change the active list of the merge depending on the special symbols we access in the vectors. Similarly, here we change the active list we are updating from $A^{i+1}$ to $A^{i}_{P}$ (or vice-versa) depending on whether the active list in the merge is $BWT^{i+1}_{bcr}$ or $P^{i}$, respectively.

Assume that, at a given point of the loop in Lines \ref{code:assmbwt1}-\ref{code:assmbwt2}, $BWT^{i+1}_{bcr}$ becomes the active list (Line~\ref{code:activebwt1}), and that the next pair to update in $A^{i+1}$ is $A^{i+1}[u_a] = (s, e)$. We first check if the current uncompressed position $BWT^{i+1}_{bcr}[u]$ (i.e., the one that we are consuming in the merge) falls within $(s, e)$. If $u=s$, we set $s = s'$, where $s'$ is the uncompressed position in $BWT^{i}_{bcr}$ where we store $BWT^{i+1}_{bcr}[u]$. Then, when $u$ equals $e$, we update $e$ accordingly, increase $u_a=u_a+1$, and set $A^{i+1}[u_a]$ as the next pair to update in $A^{i+1}$. Now assume the active symbol in the merge in the $xth$ uncompressed position of $P^{i}$ (Line~\ref{code:activebwt2}), and that the next pair we need to update in $A^{i}_{P}$ is $A^{i}_{P}[u_b]=(s, e)$. In this case, we check if $x = s$ and update the left value $s=s'$, where $s'$ is the uncompressed position in $BWT^{i}_{bcr}$ where we store the active symbol of $P^{i}$. On the other hand, if $x=e$, we update $e$, increase the index $u_b=u_b+1$, and set $A^{i}_{P}[u_b]$ as the next pair to update in $A^{i}_{P}$. 

After we finish updating the values, we merge $A^{i+1}$, $pA^{i}$, and $A^{i}_{P}$ in an orderly way to produce $A^{i}$. As mentioned, this step only requires a simultaneous scan of the vectors.
\section{Experiments}\label{sec:exp}

We implemented $\mathsf{grlBWT}$ as a $\texttt{C++}$ tool, also called \texttt{grlBWT}. This software uses the $\texttt{SDSL-lite}$ library~\cite{gbmp2014sea} to operate with bit vectors and rank data structures. This implementation includes the improvements we described in Section~\ref{sec:ref_algo} to reduce the size of the dictionary, but it does not contain the procedure to compute the optimal BWT (Section~\ref{sec:optbwt}). Our source code is available at \url{https://github.com/ddiazdom/grlBWT}.

\subsection{Implementation Details}\label{sec:imp_det}

Our algorithm constructs a dictionary of phrases in every text $T^{i}$ and then replaces the occurrences of those phrases with their corresponding symbols in $T^{i+1}$. These two steps can be challenging to implement in massive inputs as they are linear-time. Using a hash table is a simple alternative to computing the dictionary, but it can impose a considerable overhead (in terms of time and space) if the text is not so repetitive (the less repetitiveness, the bigger the dictionary). 

We implement a simple parallel hashing strategy to compute the dictionaries in a more efficient way. In every parsing round $i$ (Section~\ref{ssec:parpha}), we proceed as follows: we first set a buffer size $b$ and the number $p$ of parallel processes we will run. Both $b$ and $p$ are input parameters. Subsequently, we allocate $b/p$ bits (assume for simplicity $b$ is divisible by $p$) to store a semi-external hash table $H_{j}$, with $j \in [1,p]$, for every parallel process. It is semi-external in the sense that every time $H_{j}$ has to grow beyond $b/p$ bits (either because it exceeds the maximum load factor or because of the insertion of a new pair), it dumps all its content to disk and resets its state to empty. We implement $H_{j}$ using Robing Hood probing to operate at high load factors.

Once we initialize the hash tables $H_{i},\ldots,H_{p}$, we divide the input text $T^{i}[1..n_{i}]$ of the parsing round into $p$ chunks of $\lceil n_{i}/p \rceil$ symbols each. Every $jth$ parallel process will consume the $jth$ chuck of $T^{i}$ and store its phrases in its corresponding hash table $H_{j}$. When the parallel processes finish, we merge the dumps of the hash tables $H_{1},\ldots,H_{p}$ into one single hash table $H$ that contains all the phrases of $T^{i}$. 

The reason why this approach is efficient is simple: in the first parsing round, the dictionary is, in most cases, small compared to $T^{1}$, so it is likely that the hash tables $H_{1},\ldots,H_{p}$ contain near-identical copies of the same small string set and performed almost no data dumps. This observation means that the construction of the final hash table (the merge) is fast in practice.

The dictionary size increases considerably in the next parsing rounds, but so does the number of unique phrases (those with frequency one in $T^{i}$). Thus, the chances of  a phrase appearing in different hash tables decrease as the parsing rounds move forward.

\paragraph{Effects of Page Caching} We remark that our semi-external parsing strategy is efficient only if $T^{i}$ fits the page cache. That is, the free area of the main memory where the operative system's kernel keeps the recently-accessed file pages. Reading different areas of $T^{i}$ simultaneously from the disk is costly in standard hard drives as it requires the disk to spin back and forward to reach the sectors where the requested file pages reside. If we have never accessed $T^{i}$ before, the kernel will inevitably perform these expensive I/O operations. However, \texttt{grlBWT} needs three parallel scans of $T^{i}$, first when it produces it from $T^{i-1}$, then when it gets $D^{i}$, and finally when it builds $T^{i+1}$. Thus, if $T^{i}$ fits the page cache, we will perform the disk operations in the first scan, and the rest will mostly use the pages of $T^{i}$ in the cache. On the other hand, when $T^{i}$ does not fit the page cache, the number of page faults\footnote{When an process asks for a file page, but that page is not in the cache and the kernel has to access the disk to retrieve it.} increases considerably, but not just that, the operations become slower due to the non-linear disk access pattern of our method. This problem implies that the parallel processes will remain idle most of the time, waiting for the kernel to complete previous page requests, making thus the whole parsing phase slow. This phenomenon of constant paging and page faults is known as cache trashing. 

\subsection{Datasets}

We consider two classes of Genomics collections for our experiments: \emph{reads} and \emph{pangenomes}. The BWT plays a key role in processing this kind of data, but constricting the transform is challenging in practice as reads and pangenomes are usually massive. It is worth mentioning that \texttt{grlBWT} works with any kind of byte alphabet, not just DNA.

Reads are overlapping strings that represent random and redundant fragments of a genome. The level of redundancy depends on the length of the reads and the \emph{coverage}: the average number of times each position of the genome was sequenced. The more coverage the sequencing experiment has, the more repetitive the read collection is. The number of DNA samples also affects the repetitiveness. It is common in Genomics to concatenate reads of closely-related individuals into one dataset. These individuals are genetically almost identical, so their reads should yield nearly identical sequences. A common problem with reads, however, is that they are short and contain sequencing errors. These limitations make the repetitive patterns of the underlying genome more difficult to capture and compress. 

A genome (from a Bioinformatics point of view) is a string collection resulting from the assembly\footnote{Merging the reads by computing suffix-prefix overlaps} of a group of reads. A pangenome, on the other hand, is a collection that can contain several assembled genomes of closely-related individuals. A pangenome is massive and highly repetitive, but problems in the assembly process and sequencing errors can break the repetitive patterns, making the collection less compressible.

We now briefly describe our datasets. Table~\ref{tab:datasets} presents basic statistics about these files.

\begin{table}[t]
\centering
\begin{adjustbox}{max width=\textwidth}
\begin{tabular}{%
l
r
S[table-format=-1.0e-1, round-precision = 2, round-mode = figures, scientific-notation = true]
S[table-format=-1.0e-1, round-precision = 2, round-mode = figures, scientific-notation = true]
S[table-format=-1.0e-1, round-precision = 2, round-mode = figures, scientific-notation = true]
S[table-format=-1.0e-1, round-precision = 2, round-mode = figures, scientific-notation = true]
S[table-format=4.2, exponent-mode=fixed, round-mode=places,round-precision=2]
}
\toprule
{Dataset} & $\sigma$ & {Number of} & {Max. str.} & {Avg. str.} & {Number of}        & {$n/r$} \\
           &          & {strings}   & {length}    & {length}    & {symbols ($n$)}   &       \\
\midrule
\ill{1} & 5 &  84006956 & \unformat{3}{151} & \unformat{3}{151} & 12769057312 & 3.18 \\
\ill{2} & 5 & 160285798 & \unformat{3}{151} & \unformat{3}{151} & 24363441296 & 4.07 \\
\ill{3} & 5 & 235805550 & \unformat{3}{151} & \unformat{3}{151} & 35842443600 & 4.67 \\
\ill{4} & 5 & 305931740 & \unformat{3}{151} & \unformat{3}{151} & 46501624480 & 5.03 \\
\ill{5} & 5 & 377453488 & \unformat{3}{151} & \unformat{3}{151} & 57372930176 & 5.33 \\
\hifi & 5 & 6153050 & 50061 & 20029 & 123240580106 & 19.27\\
\midrule
\hg{05} & 16 & 334065 & 248956422 & 42715 & 14269998434 & 4.82 \\
\hg{10} & 16 & 759341 & 250522664 & 39025 & 29634170092 & 8.76 \\
\hg{15} & 16 & 835485 & 250522664 & 53918 & 45048695199 & 12.02\\
\hg{20} & 16 & 874235 & 250522664 & 68650 & 60017146889 & 15.67\\
\hg{25} & 16 & 899424 & 250522664 & 83447 & 75055723570 & 19.42\\
\hg{400} & 16 & 10012324 & 250749104 & 121094 & 1212440438146 & 224.4\\
\ecoli & 16 & 666225 & 5704396 & 28231 & 18808586125 & 140.018\\
%\midrule
%\texttt{ein} & 140 & \unformat{1}{1} & 467626545 & 467626545 & 467626545 & 1611.18\\
%\texttt{wld} & 90 & \unformat{1}{1} & 46968182 & 46968182 & 46968182 & 81.8993\\
\bottomrule
\end{tabular}
\end{adjustbox}
\caption{Datasets. The rightmost column shows the ratio between $n$ and the number of runs ($r$) in the BCR BWT obtained without changing the order of the strings.}
\label{tab:datasets}
\end{table}

\begin{itemize}
\item \emph{Illumina} (\ill{x}): five collections of Illumina reads\footnote{\url{https://www.illumina.com}} generated from different human genomes. The name of each collection has the format \ill{x}, where \texttt{x} indicates the number of individuals in the collection. For instance, the file \ill{5} encodes reads from five humans. We obtained the raw reads from the International Genome Project~\footnote{\url{https://www.internationalgenome.org/data-portal/data-collection/hgdp}}.  

\item \emph{PacBio HiFi} (\hifi): one read collection from one human genome sequenced at deep coverage (40x) using the PacBio HiFi technology\footnote{\url{https://www.pacb.com/technology/hifi-sequencing}}. HiFi reads are longer than Illumina reads; hence, they are more repetitive.

\item \emph{Human Pangenomes} (\hg{x}): 400 different human assemblies from NCBI\footnote{\url{https://www.ncbi.nlm.nih.gov/data-hub/genome/?taxon=9606}} grouped into different files to simulate six different pangenomes. Every pangenome file has the format \hg{x}, indicating that this collection contains \texttt{x} assemblies. The sources of this data are varied: some are different assemblies of the same genome, while others are assemblies of different genomes (different humans and/or cell lines). The quality of the assemblies is also heterogeneous. Some are high quality, while others are poor or intermediate reconstructions. The file \hg{400} encodes 400 assemblies, and it is the largest input in our experiments (1.2~TB). This file is the one that most closely resembles a real-life pangenome, given its massiveness and the variability in the strings it contains.

\item \emph{Escherichia Coli Pangenome} (\ecoli): 31,733 different assemblies of the Escherichia coli (E. coli) genome downloaded from NCBI\footnote{\url{https://www.ncbi.nlm.nih.gov/data-hub/genome/?taxon=562}}. As with humans, these assemblies come from different sources, and their qualities are variable. This file is highly repetitive but much smaller than our artificial pangenome \hg{400} as the E. coli genome is small. 
\end{itemize}

\subsection{Competitor Tools and Experimental Setup}

We compared the performance of $\texttt{grlBWT}$ against other tools that compute BWTs for string collections:

\begin{itemize}
    \item \texttt{ropebwt2}\footnote{\url{https://github.com/lh3/ropebwt2}}: a variation of the original BCR algorithm of Bauer et al.~\cite{bauer13lw} that uses rope data structures~\cite{boehm1995ropes}. This method is described in Heng Lee~\cite{li2014fast}.
    \item \texttt{pfp-eBWT}\footnote{\url{https://github.com/davidecenzato/PFP-eBWT}}: the eBWT algorithm of Boucher et al.~\cite{bou21com} that builds on PFP and ISS.
    \item \texttt{r-pfpbwt}\footnote{\url{https://github.com/marco-oliva/r-pfbwt}}: implementation of the method of Oliva et al.~\cite{ol23rec} that applies recursive rounds of PFP.
    \item \texttt{BCR\_LCP\_GSA}\footnote{\url{https://github.com/giovannarosone/BCR\_LCP\_GSA}} : the current implementation of the semi-external BCR algorithm~\cite{bauer13lw}.
    \item \texttt{egap}\footnote{\url{https://github.com/felipelouza/egap}}: a semi-external algorithm of Edigi et al.~\cite{egidi19ext} that builds the BCR BWT. 
    \item \texttt{gsufsort}\footnote{\url{https://github.com/felipelouza/gsufsort}}: an in-memory method proposed by Louza et al.~\cite{lou20gsuf} that computes the BCR BWT and (optionally) other data structures. 
\end{itemize}

We also considered the tool \texttt{bwt-lcp-em}~\cite{bon20com} for the experiments. Still, by default, it builds both the BWT and the LCP array, and there is no option to turn off the LCP array, so we discarded it. We compiled all the tools according to their authors' descriptions. For \texttt{grlBWT}, we used the compiler flags \texttt{-O3 -msse4.2 -funroll-loops -march=native}.

\paragraph{Experiments on Reads} We ran \texttt{grlBWT}, \texttt{ropebwt2}, \texttt{egap}, \texttt{gsufsort}, \texttt{BCR\_LCP\_GSA}, and \texttt{pfp-bwt} on Illumina data. We did not use \texttt{r-pfpbwt} as it is unsuitable for short reads. We limited the RAM usage of \texttt{egap} to three times the input size. For \texttt{BCR\_LCP\_GSA}, we turned off the construction of the data structures other than the BCR BWT and left the memory parameters by default. In the case of \texttt{gsufsort}, we used the flag \texttt{--bwt} to build only the BWT. For \texttt{ropebwt2}, we set the flag $\texttt{-L}$ to indicate that the data was in one-sequence-per-line format, and the flag $\texttt{-R}$ to avoid considering the DNA reverse strands in the BWT. We ran the experiments on the Illumina reads using one thread in all programs because not all support multi-threading. For this purpose, we set the extra flag \texttt{-P} to \texttt{ropebwt2} to indicate single-thread execution. Figure~\ref{fig:ill_exp} summarises the results of our experiment on Illumina data. We only tested \texttt{grlBWT} and \texttt{ropebwt2} with the \hifi{} dataset (HiFi reads) as \texttt{egap}, \texttt{BCR\_LCP\_GSA}, and \texttt{gsufsort} are unsuitable for long strings. We did not use \texttt{r-pfpbwt} with \hifi{} either because we assumed it would exceed our available resources. We based our conclusions on the results we obtained with \texttt{r-pfpbwt} on the human pangenomes. Both \texttt{grlBWT} and \texttt{ropebwt2} support multi-threading, so we used four threads in both.

\paragraph{Experiments on Small Human Pangenomes} We assessed the performance of \texttt{ropebwt2}, \texttt{grlBWT}, \texttt{pfp-ebwt}, and \texttt{r-pfpbwt} in the small human pangenomes (files \hg{5-25}). As with \hifi{}, we did not report experiments on the tools tailored for short strings\footnote{We tried to test them, but they either crushed or their resource consumption was too high to compare against the other tools.} By default, \texttt{ropebwt2} uses four working threads, so we set the same number of threads for \texttt{grlBWT}, \texttt{pfp-ebwt}, and \texttt{r-pfpbwt}. The tool \texttt{r-pfp} has three steps, each requiring a different set of parameters. In the first step (\texttt{pfp++}), we used \texttt{-w 10 -p 71}. In the second one (recursive \texttt{pfp++}), we used \texttt{-w 5 -p 11}. Finally, we used \texttt{rpfpbwt64} with the parameter \texttt{--bwt-only} to produce the BWT. The input parameters for \texttt{ropebwt2} were the same as with Illumina data, except for the flag \texttt{-P}. We ran \texttt{pfp-ebwt} with default parameters. We did not report experiments with \texttt{pfp-ebwt} and \ill{25} as their execution crashed. Our results on small human pangenomes are shown in Figure~\ref{fig:hga_exp}. 

\paragraph{Experiments on the E. coli Pangenome} The file \ecoli{} (i.e., the E. coli pangenome) is particularly repetitive and not so big (see Table~\ref{tab:datasets}), so we used it to assess the performance improvement one could obtain in the BWT construction under highly-repetitive scenarios. We also used it to evaluate the impact of our parallel method (Section~\ref{sec:imp_det}) as this file fits the page cache of our machine. Thus, we limited our experiments on \ecoli{} to the tools \texttt{ropebwt2}, \texttt{pfp-ebwt}, \texttt{r-pfpbwt}, and \texttt{grlBWT}. We included \texttt{ropebwt2} as a baseline because it is the non-repetition-aware tool (i.e., it does not exploit repetitions) with the best performance. We ran each software twice, one execution with one thread and the other with four threads. Figure~\ref{fig:ecoli_exp} presents the results on the E. coli pangenome. 

\paragraph{Experiments on the Big Human Pangenome} We evaluated the performance of \texttt{grlBWT} in the big human pangenome (\hg{400}). We measured the running time and memory consumption of every round of \texttt{grlBWT} to look for potential problems that are not evident in small and repetitive instances (like \ecoli). We used 10 threads and a buffer for the parallel hash tables (Section~\ref{sec:imp_det}) whose size in RAM is 10\% of the input (near 120 GBs). The selection of 10\% for the buffer was arbitrary, and it is an input parameter. We did not perform experiments on \hg{400} with the other competitor tools because of the high computational resources they would require. Figures~\ref{fig:hg400_res1},~\ref{fig:hg400_res2}, and \ref{fig:hg400_res3} show the running time and memory usage of \texttt{grlBWT} with \hg{400}. 

\paragraph{Experiments on Super Phrases} We assessed the impact of super phrases in the LMS parsing. First, we ran our current implementation of \texttt{grlBWT}, which computes super phrases as described in Section~\ref{sec:prac_super_phrases}, and then we ran our CPM'22 version\footnote{Pre-release v1.0.0-alpha in our GitHub repository.} of \texttt{grlBWT} , which does not include super phrases. We ran both implementations with the inputs \texttt{hg10} (repetitive) and \texttt{ill1} (not so repetitive), recording the number of phrases in each $\mathcal{F}^{i}$ as well as its number of symbols $||\mathcal{F}^{i}||$. The results are shown in Figure~\ref{fig:super_phrases}. 

\paragraph{Machine} We carried out the experiments on a machine with Debian 4.9, 736 GB of RAM, and processor Intel(R) Xeon(R) Silver @ 2.10GHz, with 32 cores.

\section{Results and Discussion}

\subsection{Illumina and HiFi Reads}

The fastest method in Illumina reads was \texttt{ropebwt2}, with a mean elapsed time of 4.14 hours. It is then followed by \texttt{grlBWT}, \texttt{gsufsort}, \texttt{BCR\_LCP\_GSA}, \texttt{pfp-bwt}, and \texttt{egap}, with mean elapsed times of 6.08, 9.43, 9.58, 13.08, and 27.30 hours, respectively (Figure~\ref{fig:ill_exp}B). Regarding the working space, the most efficient was \texttt{BCR\_LCP\_GSA}, with an average memory peak of 5.73 GB. It is then followed by \texttt{grlBWT} and \texttt{ropebwt2}, with average memory peaks of 23.95 and 26.64 GBs, respectively. In both cases, the memory consumption increases slowly with the input size (see Figure~\ref{fig:ill_exp}A). In contrast, \texttt{egap}, \texttt{gsufsort}, and \texttt{pfp-ebwt} are far more expensive, and their memory consumption grows fast. The tool \texttt{egap} uses 110.94 GBs on average. On the other hand, \texttt{pfp-ebwt} and \texttt{gsufsort} have similar average memory peaks: 331.98 and 372.68 GBs, respectively. We notice \texttt{grlBWT} is the tool with the second-best overall performance, only outperformed in time by \texttt{ropebwt2} and in space by \texttt{BCR\_LCP\_GSA}. We consider this result remarkable as \texttt{grlBWT} does not perform any optimization on short reads. Besides, the repetitive patterns (the main feature \texttt{grlBWT} uses to reduce CPU time and space consumption) are highly fragmented in short reads.

One possible explanation for why we outperformed even the in-memory tool \texttt{gsufsort} is because \texttt{grlBWT} is less likely to have cache misses as it operates over small data structures. In particular, \texttt{gsufsort} resembles \textsf{SA-IS}, so it runs ISS over each $T^{i}$. The problem is that the cache misses triggered by the induction of distant suffix array buckets affect the running time, and the longer $T^{i}$ is, the more cache misses we trigger. Our method, in contrast, performs ISS over the parsing set $\mathcal{F}^{i}$, which is considerably smaller than $T^{i}$, making cache misses far less likely. 

\begin{figure}[t]
\centering
\includegraphics[width=\textwidth]{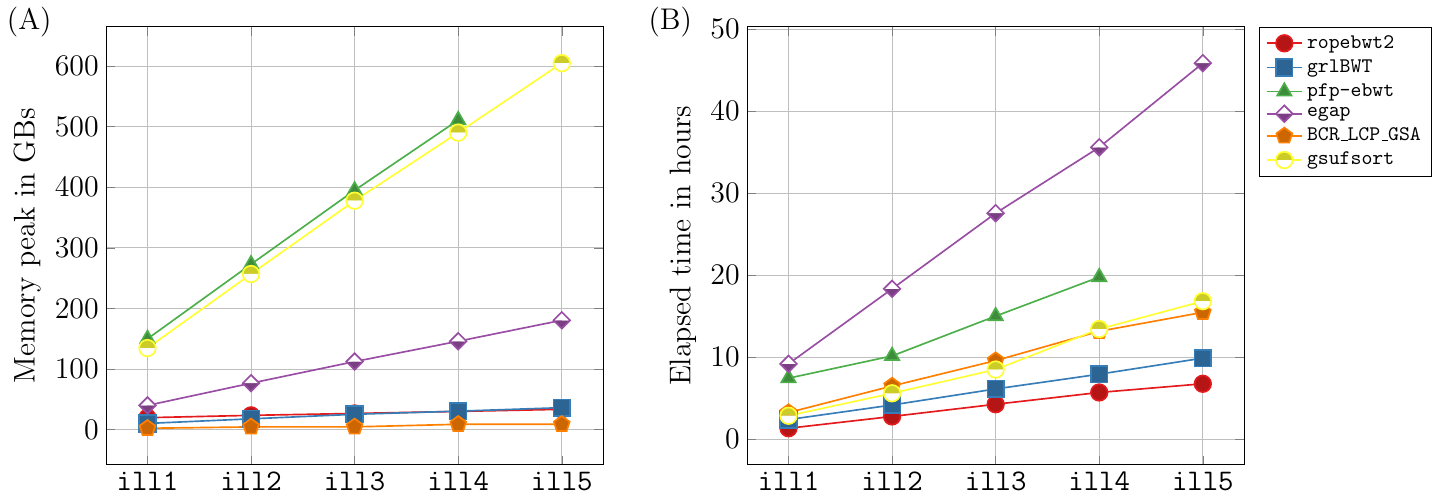}
\caption{Memory peak usage (GBs) and elapsed time (in hours) for the Illumina collections.}
\label{fig:ill_exp}
\end{figure}

Our results on HiFi reads (\hifi{}) differ from what we obtained with Illumina. The elapsed time for \texttt{grlBWT} was 8.48 hours, almost half the time spent by \texttt{ropebwt2} (16.02 hours). Still, they performed similarly in terms of memory peak: 25.15 GB of \texttt{grlBWT} versus 27.20 GB of \texttt{ropebwt2}.

We believe there are two reasons for our results on reads. The first reason is that \texttt{ropebwt2} is highly optimized for short reads, but not for long strings. The second reason is that HiFi reads are longer than Illumina reads, so \texttt{grlBWT} can capture repetitive patterns more efficiently. 

\subsection{Small Human Pangenomes}

Our tool \texttt{grlBWT} was the fastest software in small human pangenomes, with an average elapsed time of 3.94 hours versus 9.55, 20.95, and 11.26 hours for \texttt{pfp-ebwt}, \texttt{ropebwt2}, and \texttt{r-pfpbwt}, respectively. The time for \texttt{grlBWT}, \texttt{pfp-ebwt}, and \texttt{r-pfpbwt} grows smoothly with the input size, while the time for \texttt{ropeBWT} grows fast (see Figure~\ref{fig:hga_exp}B). These patterns of growth are because \texttt{grlBWT}, \texttt{pfp-ebwt}, \texttt{r-pfpbwt} exploit the text repetitions, while \texttt{ropeBWT2} does not. 

Regarding memory peak, \texttt{grlBWT} is also the most efficient tool, with a mean of 10.52 GB versus 18.05, 204.62, and 587.05 GBs for \texttt{ropebwt2}, \texttt{pfp-ebwt} and \texttt{r-pfpbwt}, respectively. Although \texttt{grlBWT} outperforms \texttt{ropebwt2} on average, their memory functions have the same pattern: both grow smoothly with the input size. In contrast, the memory consumption of \texttt{pfp-ebwt} and \texttt{r-pfpbwt} is considerable, although they still have a smooth pattern of growth (see Figure~\ref{fig:hga_exp}A). We did not expect \texttt{pfp-ebwt} and \texttt{r-pfpbwt} to have a high memory consumption as they also exploit the text repetitions. Still, we acknowledge this result could be because we did not choose suitable input parameters or because the implementations of \texttt{pfp-ebwt} and \texttt{r-pfpbwt} are still incomplete. It might also be possible that the parsing scheme of \texttt{grlBWT} is better than PFP at capturing text repetitions. 

\begin{figure}[t]
\centering
\includegraphics[width=\textwidth]{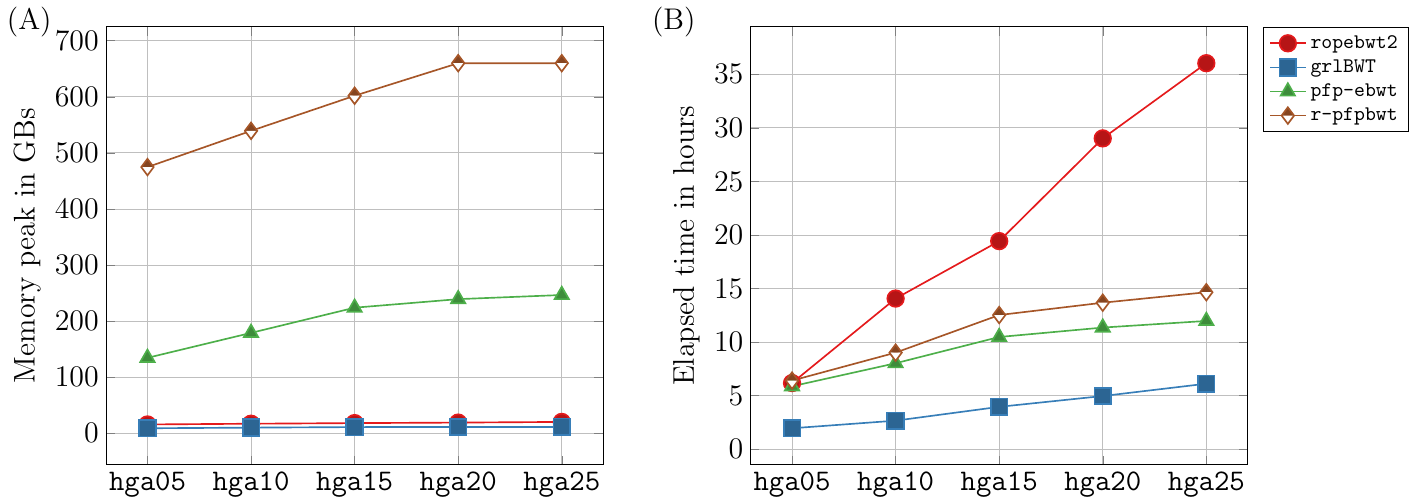}
\caption{Memory peak usage (GBs) and elapsed time (in hours) for the small human pangenomes.}
\label{fig:hga_exp}
\end{figure}

\subsection{E. coli Pangenome}

The performance of \texttt{grlBWT} is considerably better than that of \texttt{ropebwt2}, \texttt{pfp-ebwt}, and \texttt{r-pfpbwt} in highly-repetitive inputs (\ecoli). Using one thread, the elapsed time of \texttt{grlBWT} with \ecoli{} was 0.93 hours, while the elapsed times of \texttt{ropebwt2}, \texttt{pfp-ebwt}, and \texttt{r-pfpbwt} were 5.56, 1.08, and 1.19 hours, respectively. Thus, the average speed of \texttt{grlBWT} was 0.18 $\mu$secs per symbol, while the average speed of the other tools was 1.06, 0.21, and 0.23 $\mu$secs per symbol (respectively). We were also the most space-efficient method, with a memory peak of 0.82 GB (0.35 bits per symbol) for \texttt{grlBWT} versus memory peaks of 10.57, 9.95, and 17.30 GBs for \texttt{ropeBWT2}, \texttt{pfp-ebwt}, and \texttt{r-pfpbwt}, respectively. Our experiments on \ecoli{} also showed we could greatly improve our running time if we use parallelization. Four threads were enough to reduce \texttt{grlBWT}'s running time by more than half, from 0.93 hours to 0.34 hours (around 20 minutes). This improvement in the speed had a negligible impact on the memory peak as it increased from 0.82 GB to 0.99 GB. The only other tool that improved its performance significantly with parallelism was \texttt{ropebwt2}. It decreased its running time from 5.56 hours to 2.75 hours without changing its memory peak. However, its results were far from what we obtained with \texttt{grlBWT}. See Figure~\ref{fig:ecoli_exp} for more details on the E.~coli experiments.

\begin{figure}[t]
\centering
\includegraphics[width=0.60\textwidth]{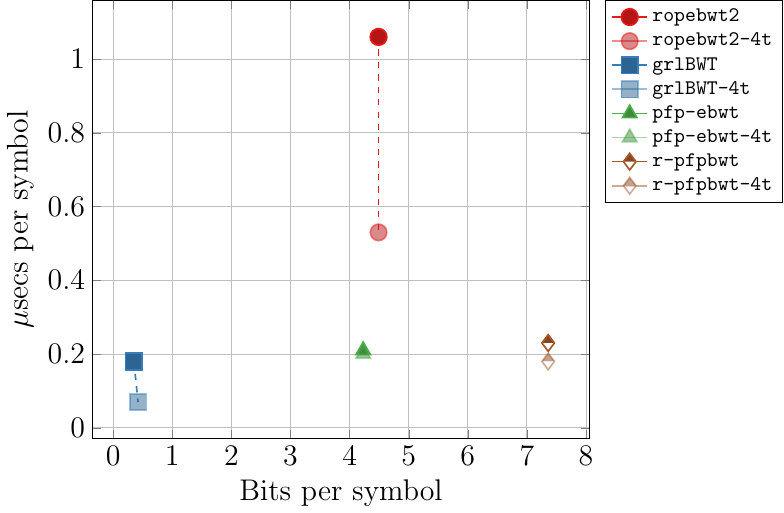}
\caption{Time-space tradeoffs for constructing the BWT of \ecoli{} (E.~coli pangenome). Time is given in microseconds per symbol and space in bits per symbol. The transparent points represent the instances using four threads (\texttt{-4t} suffix in the legend).}
\label{fig:ecoli_exp}
\end{figure}

\subsection{Big Human Pangenome}

The complete execution of \texttt{grlBWT} with the big human pangenome (file \hg{400}) took 41.21 hours and had a memory peak of 118.83 GB. Further inspection of this execution showed that the parsing phase of our algorithm (Section~\ref{ssec:parpha}) contributed to 96.1\% of the total running time of \texttt{grlBWT}, while the induction phase (Section~\ref{ssec:indpha}) contributed to the remaining 3.94\% (see Figure~\ref{fig:hg400_res2}A). Additionally, the first three parsing rounds contributed 93.2\% of the total running time of the parsing phase, with the first parsing round contributing more than 50\% (see Figure~\ref{fig:hg400_res2}B). These results indicate the bottleneck of the execution was in these rounds. A closer examination of the steps of the parsing rounds one, two, and three shows that transforming $T^{1}$ into $T^{2}$ is the most expensive step in the whole execution of \texttt{grlBWT} (see Figure~\ref{fig:hg400_res3}).

\begin{figure}[t]
\centering
\includegraphics[width=0.7\textwidth]{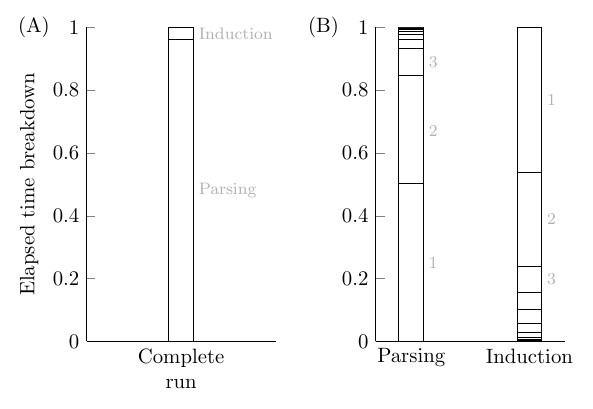}
\caption{Elapsed time breakdown of \texttt{grlBWT} using the file \hg{400}. (A) Breakdown of the phases. The bottom box is the parsing phase, and the upper box is the induction phase. The y-axis denotes the fraction of the total running time. (B) Breakdown of the rounds. Each box denotes one round in a phase. The y-axis, in this case, is the fraction of the total phase's running time. The rounds in the parsing phase (left bar) are read bottom-up, while the rounds in the induction phase (right bar) are read top-bottom. The numbers in grey to the right of the bars highlight the three most time-consuming rounds.}
\label{fig:hg400_res2}
\end{figure}

We believe this problem arises due to poor management of the page caches. Our tool \texttt{grlBWT} keeps $T^1$ and $T^{2}$ mostly on disk, loading small chunks of them (pages) into main memory to produce $T^{2}$ in a semi-external way. As explained in Section~\ref{sec:imp_det}, the Linux kernel speeds up disk accesses to these files by keeping recently-accessed pages cached in free RAM sections so they are available for future disk accessions. However, the problem in \hg{400} arises because $T^{1}$ and $T^{2}$ do not fit the page cache, so \texttt{grlBWT} triggers disk operation frequently due to page faults, making the parallel semi-external scans of $T^{1}$ and $T^{2}$ extremely slow. Despite the problem with $T^{1}$ and $T^{2}$, we note that the steps of \texttt{grlBWT} that operate over compressed data are remarkably efficient in terms of both time and space (see Figures~\ref{fig:hg400_res2} and~\ref{fig:hg400_res3}). At the end of this section, we propose an alternative solution to parse $T^{i}$ under page cache constraints.

%The problem is that \texttt{grlBWT} does not implement an efficient algorithm to decide which pages of $T^{1}$ and $T^{2}$ remain in the cache and which are evicted, it relies on the default kernel algorithm. . Notice this problem occurs only when the input file exceeds the available RAM designated for caching purposes.

The memory peak in the execution of \texttt{grlBWT} is dominated by the buffer of the parallel hash tables we use to construct the dictionary from the text (Section~\ref{sec:imp_det}). Notice the peak is 118.83 GB because we defined a buffer for these hash tables that uses at most $10\%$ of \hg{400} (about 120 GB). Put another way, the memory peak of \texttt{grlBWT} is a user-defined parameter when we process a large file in parallel.

\begin{figure}[t]
\centering
\includegraphics[width=0.8\textwidth]{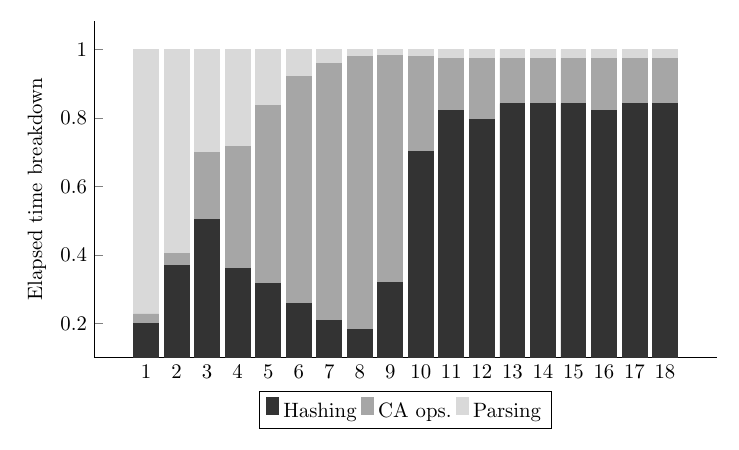}
\caption{Elapsed time breakdown of the parsing phase of \texttt{grlBWT} when executed with the input \hg{400} (i.e., the big human pangenome). Each $ith$ bar (x-axis) represents the time breakdown of the $ith$ round of parsing. ``Hashing'' (black box) refers to scanning $T^{i}$ and inserting its LMS phrases into a hash table. ``CA ops.'' (grey box) denotes the time spent performing compression-aware operations. That is, producing the dictionary's suffix array and the preliminary BWT, compressing the dictionary, and assigning symbols to the dictionary phrases. Finally, ``Parsing'' (light grey box) is the time spent transforming $T^{i}$ into $T^{i+1}$.}
\label{fig:hg400_res3}
\end{figure}

In practice, however, we are interested in the \emph{real} memory peak of our implementation. That is, the memory footprint produced by the data structures we keep in the heap (except for the aforementioned buffer, whose sole purpose is to enable a parallel execution). In every parsing round $i$, these data structures are the dictionary $D^{i}$ and its suffix $\gsame$ (plus some other minor data structures). On the other hand, during every induction phase $i$, the most relevant structures are $P^{i}$, the vector where we insert the symbols of $BWT^{i}_{bcr}$ that we induce from $BWT^{i+1}_{bcr}$ (see Section~\ref{sec:indbwtalgo}), and $G^{i}$, the grammar-like encoding of the expanded parsing set \exppset.

A close inspection of the memory footprint of the compressed data structures showed that the real memory peak of \texttt{grlBWT} is remarkably low compared to the input size: 21.05 GB ($1.7\%$ of \hg{400}'s size) during the first parsing round, and then another smaller peak of 14.09 GB ($1.1\%$ of \hg{400}'s size) during the construction of $BWT^{1}_{bcr}$ from $BWT^{2}_{bcr}$ in the last induction round (see Figure~\ref{fig:hg400_res1} for more details).

\begin{figure}[t]
\centering
\includegraphics[width=0.95\textwidth]{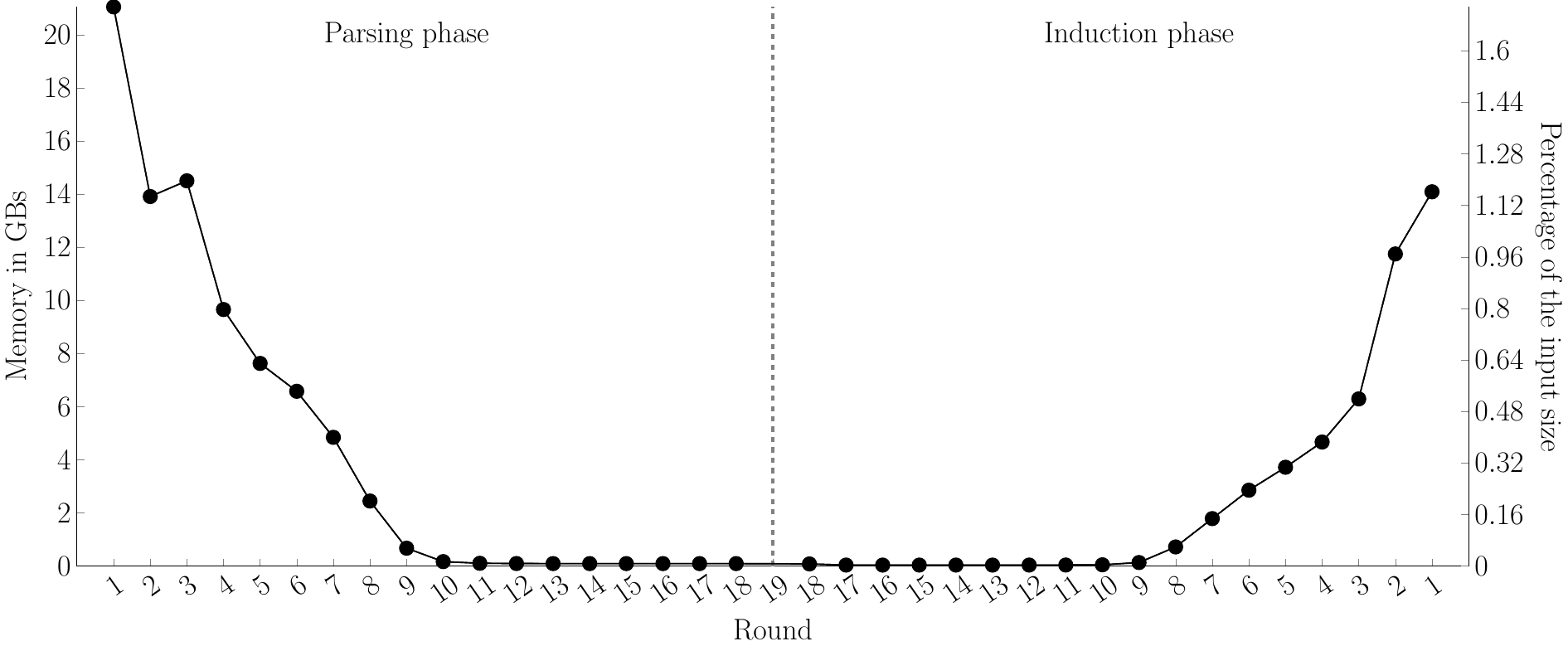}
\caption{Heap memory usage for the intermediate data structures when running \texttt{grlBWT} with \hg{400}. The x-axes are the iterations of $
\mathsf{grlBWT}$. The vertical dashed line marks the transition from the parsing phase to the induction phase. The left y-axis is the usage in GBs, while the right y-axis is the usage as a \emph{percentage} of \hg{400} in plain format. That is, the total bytes of the intermediate data structures divided by the bytes of the input (one byte per input symbol). The induction iterations are numbered backwards to match the way in which $\mathsf{grlBWT}$ works. %The gray horizontal line is the uncompressed size of \hg{400}.
The heap usage of every parsing iteration $i$ considers $\gsame$, $V^{i}$, and $D^{i}=(R^{i},L^{i}, N^{i})$ (Section~\ref{ssec:parpha}) plus other minor data structures. The heap usage of every induction iteration $i$ considers $V^{i}$, $P^{i}$, and $G^{i}$~(Section~\ref{sec:indbwtalgo}).}
\label{fig:hg400_res1}
\end{figure}

\paragraph{Parsing the Text Under Page Cache Constraints} We could tackle the problem of parsing a string $T^{i}$ that does not fit the page cache using a parallel procedure that implements a producer-consumer pattern. As before, assume we are allowed to use $p$ processes. We first initialize in main memory a set of $b>p$ buffers of $d$ bits each ($b$ and $d$ being parameters). Then, we create a producer process $t_{p}$ that reads chunks of $T^{i}$ semi-externally (and linearly) and puts them in the available buffers. After $t_p$ fills a buffer, it appends it into a queue $I$ flagged as ``ready to be processed''. On the other hand, we create a set of consumer processes $t_{c,1}, \ldots, t_{c,p-1}$ that actively check the state of $I$ to see if there are available chunks. When a consumer process $t_{c, j}$ pops a buffer from $I$, it parses its text using LMS parsing and then appends the consumed buffer into another queue $O$ that keeps the already processed data. Thus, once $t_{p}$ uses all the available buffers, it pops elements of $O$ to recycle buffers for new chunks of $T^{i}$, which appends into $I$ and the cycle begins again. As the consumer processes parse $T^{i}$ in parallel, they insert the phrases into one hash table $H$ that uses lock-free CPU instructions to support concurrent queries. We have to  choose $b$ and $d$ carefully so $t_{p}$ is always reading from disk, while the processes $t_{c,1},\ldots,t_{c, p-1}$ constantly consume chunks of $T^{i}$. This idea is more efficient than the scheme we presented in Section \ref{sec:imp_det} as it almost removes the need for a page cache. No matter how many disk accesses the producer process performs, the consumer processes do not remain idle.

\subsection{Effect of Super Phrases} Our heuristic of super phrases (Section~\ref{sec:prac_super_phrases}) has a notorious impact between parsing iterations three and five. In \texttt{hg10}, the number of phrases in $\mathcal{F}^{3}$, $ \mathcal{F}^{4}$, and $\mathcal{F}^{5}$ reduced by 14.4\%, 28.2\%, and 36.7\% (respectively) when using our heuristic (top-left plot in Figure~\ref{fig:super_phrases}). In \texttt{ill1}, the reductions in those iterations were 17.8\%, 32.8\%, and 6.6\%, respectively (top-right plot in Figure~\ref{fig:super_phrases}). However, our heuristic fails in iteration two (the one producing the largest $\mathcal{F}^{i}$) as it achieves a negligible reduction in both inputs. The reason could be that $T^{2}$ still has several repeated symbols, so our simple method, which relies on symbol frequencies to capture super phrases, does not work. Text $T^{3}$ is more likely to have unique symbols, so our heuristic probably works better than in $T^{2}$. This result is relevant as $\mathcal{F}^{3}$ is the second-largest parsing set in both \texttt{hg10} and \texttt{ill1}. Additionally, we noticed that the reduction of $\mathcal{F}^{3}$ is better in \texttt{ill1} than in \texttt{hg10} (17.8\% versus 14.4\%). A possible explanation is that \texttt{ill1} is less repetitive than \texttt{hg10}, so our heuristic becomes more efficient. Regarding the number of symbols in each $\mathcal{F}^{i}$ (second row of Figure~\ref{fig:super_phrases}), super phrases do not have a relevant impact. In \texttt{hg10}, the number of symbols in $\mathcal{F}^{3}, \mathcal{F}^{4}$, and $\mathcal{F}^{5}$ reduced by 2.4\%, 13.3\%, and 21.9\%, respectively. In contrast, in \texttt{ill1}, the reductions were 3.6\%, 22.8\%, and 34.6\%, respectively (slightly better than in \texttt{hg10}). The compression of $\mathcal{F}^{5}$ seems remarkable ($21.9\%$ and $34.6\%$), but keep in mind that this parsing set is considerably smaller than $\mathcal{F}^{2}$ and $\mathcal{F}^{3}$ in both inputs, so the overall space reduction is not big after all. We expected these results, as  we remove $p$ symbols from $\mathcal{F}^{i}$ for each sequence of $p$ consecutive parsing phrases we merge into one single super phrase, which is not much.

\begin{figure}[t]
\centering
\includegraphics[width=0.7\textwidth]{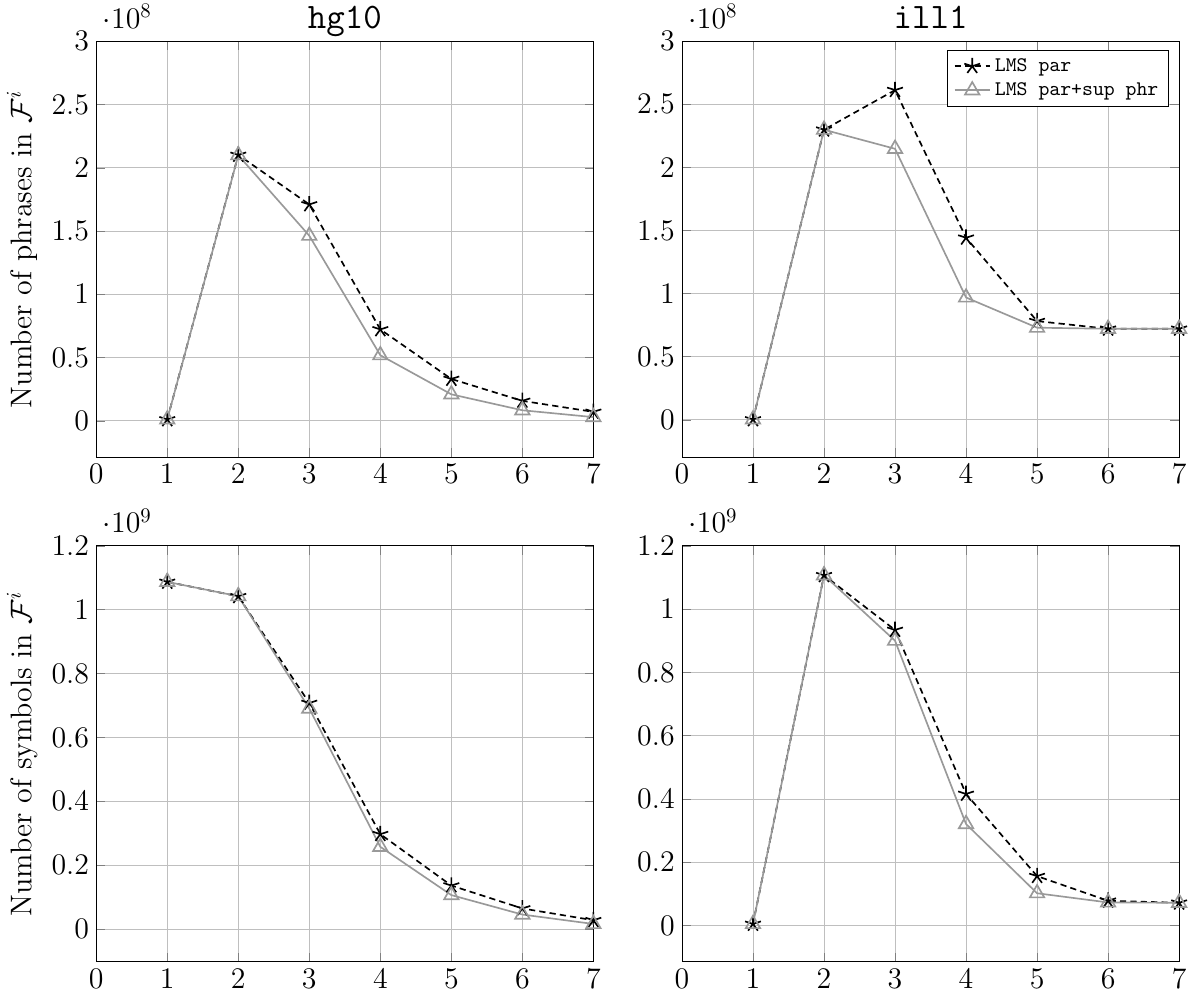}
\caption{Effect of super phrases in the first seven iterations of LMS parsing. The left column shows the results for \texttt{hg10} and the right column shows the results for \texttt{ill1}. In all the plots, the x-axis is the parsing iteration (e.g., the information of $\mathcal{F}^{3}$ is in $x=3$). The y-axes in the first row show the number of phrases in $\mathcal{F}^{i}$, with $1\leq i \leq 7$. The plots in the second row show the total number of symbols $||\mathcal{F}^{i}||$. The dashed line with star shapes is LMS parsing without super phrases (\texttt{LMS par}), and the grey line with triangle shapes is the LMS parsing with super phrases (\texttt{LMS par+sup phr}).}
\label{fig:super_phrases}
\end{figure}
\section{Concluding Remarks}

We introduced a method for building the BCR BWT that maintains the data of intermediate stages in compressed form. The representation we chose reduces not only working memory but also computation time. Our experimental results showed that our algorithm is competitive with the state-of-the-art tools under not-so-repetitive scenarios and greatly reduces the computational requirements when the input becomes more repetitive. This last feature proved an efficient solution for processing terabytes of redundant data under limited computational resources. For now, the hard drive is the main aspect that degrades our performance in large inputs and prevents us from running even more significant collections. However, we are confident we can solve these problems with a more careful implementation of our algorithm.

An important observation is that our framework enables the construction of the $r$-index~\cite{g2018op} for large collections in practice, as it is possible to obtain this data structure in $O(r)$ bits using the BWT as input. This idea certainly facilitates the indexation of large-scale pangenomes. However, more is needed to make the $r$-index a practical solution for pangenomes as it does not support all the relevant queries necessary for Genomics analyses.

We believe it is possible to use our repetition-aware strategy for other operations. For instance, update and merge multiple BWTs. The intuition to merge BWTs is that if we have several texts, we first run the parsing round of \textsf{grlBWT} independently in each of them to produce a list of dictionaries (one dictionary set for each text). Then, we combine the dictionaries in one set, and finally, we run the induction phase of \textsf{grlBWT} over the combined dictionary set. Our experiments showed that manipulating dictionary sets and running the induction phase of \textsf{grlBWT} is fast, even in terabytes of data. Thus, the merge of the BWTs should be fast too. The update of a BWT should work similarly. We first produce an initial BWT by running \textsf{grlBWT} over a text collection and save its dictionary set. Then, if we need to append more sequences to this BWT, we run the parsing phase of \textsf{grlBWT} on the new sequences to produce a new dictionary set, which we combine with the one we previously saved. As before, we run the induction phase over the combined dictionary set to produce the updated BWT. We can keep the combined dictionary set again if we need to append more sequences in the future. These ideas (merge and update) could enable the efficient construction of huge BWTs in distributed systems.

There are also other applications for our compression-aware technique we would like to explore, not just BWT-related topics. For instance, computing all-vs-all maximal exact matches in string collections, grammar or Lempel-Ziv compression, self-indexes, and approximate or multiple alignments, among other things.

\section*{Acknowledgements}

Funded in part by Basal Funds FB0001, ANID, Chile.

%% The Appendices part is started with the command \appendix;
%% appendix sections are then done as normal sections
%% \appendix

%% \section{}
%% \label{}

%% If you have bibdatabase file and want bibtex to generate the
%% bibitems, please use
%%
\bibliographystyle{elsarticle-num} 
\bibliography{references.bib}

%% else use the following coding to input the bibitems directly in the
%% TeX file.

%\begin{thebibliography}{00}
%% \bibitem{label}
%% Text of bibliographic item
%\bibitem{}
%\end{thebibliography}

\end{document}